\documentclass[11pt]{article}

\usepackage[utf8]{inputenc}
\usepackage[letterpaper, margin=1in]{geometry}
\usepackage{amsmath,amssymb,amsthm,amsfonts}
\usepackage{mathtools}
\usepackage{graphicx} 
\usepackage{float}
\usepackage{subcaption} 
\usepackage{arydshln}
\usepackage{multirow}
\usepackage{bm}

\usepackage{hyperref}
\hypersetup{colorlinks,linkcolor=blue,urlcolor=blue,citecolor=teal}
\usepackage{url}
\usepackage{xcolor}
\usepackage[title]{appendix}
\usepackage{cite}
\usepackage{tikz}  
\usepackage{tikz-cd}  
\usetikzlibrary{positioning}
\usepackage{tikz-qtree}
\usepackage{mathdots}
\usepackage{authblk}
\usepackage{dsfont}
\usepackage[version=4]{mhchem}

\usepackage{qcircuit}
\newcommand{\ket}[1]{| #1 \rangle}
\newcommand{\bra}[1]{\langle #1|}
\newcommand{\braket}[1]{\langle #1 \rangle}
\newcommand{\bracket}[3]{\langle #1|#2|#3 \rangle}

\DeclareMathOperator{\polylog}{polylog}
\DeclareMathOperator{\tr}{Tr}

\DeclareMathOperator{\ad}{ad}
\renewcommand{\Re}{\operatorname{Re}}

\def \eps {\varepsilon}

\renewcommand{\d}{\mathrm{d}}
\renewcommand{\i}{{\mathrm i}}

\allowdisplaybreaks

\newtheorem{thm}{Theorem}[section]

\newtheorem{prop}[thm]{Proposition}
\newtheorem{cor}[thm]{Corollary}
\newtheorem{lem}[thm]{Lemma}

\newtheorem{prob}[thm]{Problem}

\theoremstyle{definition}

\newtheorem{defn}[thm]{Definition}

\numberwithin{equation}{section}

\newtheorem*{thm*}{Theorem}
\newtheorem*{lem*}{Lemma}
\newtheorem*{prop*}{Proposition}
\newtheorem*{defn*}{Definition}
\newtheorem*{prob*}{Problem}
\newtheorem*{ques*}{Question}

\newcommand{\mb}{\mathbb}

\newcommand{\mc}{\mathcal}

\def\Tr{\mathrm{Tr}}
\def\E{\mathbb{E}}
\def\b{{\bf b}}

\usepackage{algorithm}
\usepackage{algorithmic}

\makeatletter
\newenvironment{breakablealgorithm}
  {
   \begin{center}
     \refstepcounter{algorithm}
     \hrule height1pt depth0pt \kern3pt
     \renewcommand{\caption}[2][\relax]{
       {\raggedright\textbf{\ALG@name~\thealgorithm} ##2\par}%
       \ifx\relax##1\relax 
         \addcontentsline{loa}{algorithm}{\protect\numberline{\thealgorithm}##2}%
       \else 
         \addcontentsline{loa}{algorithm}{\protect\numberline{\thealgorithm}##1}%
       \fi
       \kern3pt\hrule\kern3pt
     }
  }{
     \kern3pt\hrule\relax%
   \end{center}
  }
\makeatother

\newcommand{\be}{\begin{equation}}
\newcommand{\ee}{\end{equation}}
\newcommand{\bea}{\begin{eqnarray}}
\newcommand{\eea}{\end{eqnarray}}
\newcommand{\bes}{\begin{equation*}}
\newcommand{\ees}{\end{equation*}}
\newcommand{\beas}{\begin{eqnarray*}}
\newcommand{\eeas}{\end{eqnarray*}}

\title{Randomized Quantum Singular Value Transformation}

\author[1,2]{Xinzhao Wang\thanks{wangxz@stu.pku.edu.cn}}
\author[3]{Yuxin Zhang\thanks{zhangyuxin@amss.ac.cn}\thanks{The first two authors contributed equally.}}
\author[4]{Soumyabrata Hazra\thanks{soumyabrata.hazra@research.iiit.ac.in}}
\author[1,2]{Tongyang Li\thanks{tongyangli@pku.edu.cn}}
\author[3]{Changpeng Shao\thanks{changpeng.shao@amss.ac.cn}}
\author[4]{Shantanav Chakraborty\thanks{shchakra@iiit.ac.in}}

\affil[1]{\small{Center on Frontiers of Computing Studies, Peking University, Beijing,  China}}
\affil[2]{\small{School of Computer Science, Peking University, Beijing, China}}
\affil[3]{\small{SKLMS, Academy of Mathematics and Systems Science, Chinese Academy of Sciences, Beijing, China}}
\affil[4]{\small{CQST and CSTAR, International Institute of Information Technology Hyderabad, Telangana, India}}

\begin{document}

\maketitle
\thispagestyle{empty}
\begin{abstract}
We introduce the first randomized algorithms for Quantum Singular Value Transformation (QSVT), a unifying framework for many quantum algorithms. Standard implementations of QSVT rely on block encodings of the Hamiltonian, which are costly to construct, requiring a logarithmic number of ancilla qubits, intricate multi-qubit control, and circuit depth scaling linearly with the number of Hamiltonian terms. In contrast, our algorithms use only a single ancilla qubit and entirely avoid block encodings. We develop two methods: (i) a direct randomization of QSVT, where block encodings are replaced by importance sampling, and (ii) an approach that integrates qDRIFT into the generalized quantum signal processing framework, with the dependence on precision exponentially improved through classical extrapolation. Both algorithms achieve gate complexity independent of the number of Hamiltonian terms, a hallmark of randomized methods, while incurring only quadratic dependence on the degree of the target polynomial. We identify natural parameter regimes where our methods outperform even standard QSVT, making them promising for early fault-tolerant quantum devices. We also establish a fundamental lower bound showing that the quadratic dependence on the polynomial degree is optimal within this framework. We apply our framework to two fundamental tasks: solving quantum linear systems and estimating ground-state properties of Hamiltonians, obtaining polynomial advantages over prior randomized algorithms. Finally, we benchmark our ground-state property estimation algorithm on electronic structure Hamiltonians and the transverse-field Ising model with long-range interactions. In both cases, our approach outperforms prior work by several orders of magnitude in circuit depth, establishing randomized QSVT as a practical and resource-efficient alternative for early fault-tolerant quantum devices.
\end{abstract}

\newpage
\pagenumbering{roman}

\tableofcontents

\newpage
\pagenumbering{arabic}
\setcounter{page}{1}

\section{Introduction}

Quantum singular value transformation (QSVT) \cite{gilyen2019quantum} has emerged as a central paradigm in quantum algorithms, unifying and extending many of the most powerful techniques. 
At its core, QSVT enables polynomial transformations of the singular values of an operator, provided the operator is embedded as the top-left block of a unitary, known as a block encoding \cite{low2019hamiltonian, chakraborty_et_al:LIPIcs.ICALP.2019.33}. 
This simple but versatile abstraction subsumes a broad class of prior methods and provides a systematic framework for designing new algorithms. 
Consequently, QSVT underlies state-of-the-art approaches to widely studied problems of practical interest, including Hamiltonian simulation \cite{low2017qsp, low2019hamiltonian}, solving quantum linear systems \cite{gilyen2019quantum,lin2020optimal}, linear regression and quantum machine learning \cite{chakraborty_et_al:LIPIcs.ICALP.2019.33, chakraborty2023quantum}, quantum walks \cite{gilyen2019quantum, apers2021unified}, and optimization \cite{brandao2017SDP, vanapeldoorn2019improvedsdp, lin2020nearoptimalground}. 
It further serves as the key subroutine in many recent advances in quantum algorithms \cite{martyn2021grand}.

Despite its power, QSVT suffers from two major bottlenecks that limit its near-term implementability, even on early fault-tolerant quantum devices. 
The first is its reliance on block encodings, which can be costly to construct. 
For instance, consider an $n$-qubit Hamiltonian of the form
\begin{equation}
    H=\sum_{k=1}^{L}\lambda_k H_k,
\end{equation}
where each $H_k$ is unitary (or Hermitian) with $\|H_k\|=1$, and the total weight of the coefficients $\lambda=\sum_{k}|\lambda_k|$. 
Such Hamiltonians are ubiquitous in condensed matter \cite{sachdev2011quantum}, high-energy physics \cite{sachdev1993gapless, maldacena2016remarks, babbush2019quantum}, and quantum chemistry \cite{cao2019quantum, campbell2019random, mcclean2020openfermion, su2021fault, babbush2018low}. 
A block encoding of any such $H$ is typically obtained via the Linear Combination of Unitaries (LCU) method, using $O(\log L)$ ancilla qubits and a series of complicated controlled operations. 
Moreover, it was recently shown that this logarithmic ancilla overhead is unavoidable in general \cite{chakraborty2025quantum}. Consequently, implementing a degree-$d$ polynomial transformation via QSVT results in circuit depth scaling as $\widetilde{O}(L \lambda d)$, while also incurring an $O(\log L)$ ancilla overhead. 

The second bottleneck is the explicit linear dependence on $L$, which can be prohibitively large for many Hamiltonians of practical interest, particularly those arising in quantum chemistry \cite{cao2019quantum, campbell2019random, mcclean2020openfermion, su2021fault, babbush2018low} and high-energy physics \cite{sachdev1993gapless, maldacena2016remarks}. This challenge has motivated growing interest in randomized quantum algorithms, which can eliminate the dependence on $L$. 
The key idea is to replace deterministic constructions, such as Trotter decompositions \cite{lloyd1996universal, childs2021theory} or LCU \cite{childs2012hamiltonian}, with randomized procedures that reproduce the desired evolution in expectation.
Notable examples include qDRIFT \cite{campbell2019random, watson2024randomly} and qSWIFT \cite{nakaji2024qswift} for Hamiltonian simulation, randomized LCU methods \cite{chakraborty2024implementing, wang2024qubit} for more general transformations, and randomized implementations of linear combination of Hamiltonian simulation \cite{yang2025circuit}.

There has been growing interest in designing algorithms tailored to early fault-tolerant quantum devices, where the number of logical qubits is limited, ancilla qubits are scarce, and multi-qubit controlled gates remain expensive or unavailable \cite{dong2022ground, zhang2022computingground, lin2022heisenberg, Wang2023quantumalgorithm, katabarwa2024early}. 
In this regime, resource-efficient methods for broad algorithmic paradigms are of critical importance. Recent progress has been achieved for the LCU framework, yielding algorithms with very low hardware overhead and applications to problems of practical interest \cite{chakraborty2024implementing, wang2024qubit}. However, these approaches achieve only sub-optimal total complexity. 
Since QSVT is a more general paradigm, encompassing LCU and many other techniques, developing early fault-tolerant algorithms for QSVT is even more pressing. A recent advance in this direction is the work of \cite{chakraborty2025quantum}, which introduced a QSVT algorithm based on higher-order Trotterization \cite{childs2021theory} combined with classical extrapolation. Their method achieves circuit depth $\widetilde{O}(L\lambda d^{1+o(1)})$, which is nearly optimal, while avoiding block encodings and using only a single ancilla qubit. This resolves the first bottleneck (block-encoding and ancilla overhead), but the second (linear dependence on $L$) remains. 
Moreover, the use of higher-order Trotter–Suzuki decompositions introduces an additional challenge: the prefactor in the gate complexity grows exponentially with the order $2k$. In practice, only second-order $(k=1)$ and fourth-order $(k=2)$ formulas are typically feasible \cite{campbell2019random}, leading to complexities scaling as $\sim L d^{3/2}$ and $\sim L d^{5/4}$, respectively.

In this work, we develop the first randomized quantum algorithms for QSVT, simultaneously overcoming both bottlenecks described above. We introduce two distinct approaches. 
The first is a direct randomization of the standard QSVT framework. For Hamiltonians with an LCU decomposition, we incorporate stochastic sampling into the construction, thereby preserving the alternating-phase structure of QSVT while replacing costly block encodings with randomly sampled unitaries. Conceptually, this establishes the feasibility of randomized block encoding. However, it also inherits certain inefficiencies of prior randomized approaches \cite{zhang2022computingground, chakraborty2025quantum, wang2024qubit}, most notably the need for a suboptimal number of classical repetitions in regimes where block-encoding-based QSVT benefits from amplitude amplification.

Our second algorithm overcomes these limitations by embedding qDRIFT-based Hamiltonian simulation \cite{campbell2019random} into the framework of generalized quantum signal processing (GQSP) \cite{wang2023quantum, motlagh2024generalized}. This formulation applies to any Hamiltonian expressible as a linear combination of Hermitian operators (LCH).\footnote{Since Pauli operators are both unitary and Hermitian, for any Hamiltonian written in terms of its Pauli decomposition, i.e.\ $H=\sum_{k=1}^{L}\lambda_k P_k$, where $P_k\in \{I,X,Y,Z\}$, the notions of LCU and LCH coincide.} 
Building on the interleaved-sequence circuit framework of \cite{chakraborty2025quantum}, we consider general quantum circuits that alternate between arbitrary unitaries and Hamiltonian evolution operators, with GQSP (and consequently, QSVT) arising as a special case. 
Our second algorithm estimates the expectation value of any observable with respect to the output of such circuits, substituting qDRIFT in place of exact time-evolutions.
A naive substitution, however, would inherit the intrinsic $1/\eps$ precision scaling of qDRIFT \cite{campbell2019random}, resulting in prohibitively large circuit depths. To overcome this, we incorporate Richardson extrapolation \cite{low2019well} in a novel way, suppressing the error and yielding an end-to-end procedure with exponentially improved precision scaling.

Importantly, this result does not follow directly from prior work on Hamiltonian simulation with classical extrapolation \cite{watson2024exponentially, watkins2024time-dependent}. 
Our main technical contribution is to establish an upper bound on the diamond-norm distance between two states: one obtained from any interleaved sequence of qDRIFT-based Hamiltonian simulation and arbitrary unitaries, and the other from the corresponding ideal interleaved circuit with exact time evolutions. 
We show that this distance can be expressed explicitly in terms of the parameters of the qDRIFT channel, which in turn enables us to bound the time-dependent error in expectation values of arbitrary observables. By judiciously tuning these parameters through Richardson extrapolation, we reduce the precision dependence from $1/\eps$ to polylog$(1/\eps)$, thereby achieving substantially shorter circuit depths.

As a result, our second randomized algorithm is both more powerful and more versatile in practice, capable of implementing polynomial transformations with asymptotically better performance than prior randomized approaches. Crucially, both of our algorithms preserve the same core advantages: they require only a single ancilla qubit, achieve circuit depth $\widetilde{O}(\lambda^2 d^2)$ independent of the number of Hamiltonian terms $L$, and completely eliminate the need for block encodings.

We also establish a matching lower bound, proving that a quadratic dependence on the polynomial degree is unavoidable for any randomized procedure that samples each $H_k$ with probability $|\lambda_k|/\lambda$. This lower bound is fundamental: it applies not only to algorithms for implementing generic polynomial transformations within this access model, but also extends naturally to other prominent randomized approaches, including Hamiltonian simulation via qDRIFT \cite{campbell2019random, chen2021concentration} and qSWIFT \cite{nakaji2024qswift}, as well as randomized LCU approaches \cite{chakraborty2024implementing, wang2024qubit}.

The advantages of our second algorithm are demonstrated through two broad applications: quantum linear systems and ground-state property estimation. In both cases, it delivers significant polynomial improvements over randomized LCU-based methods \cite{chakraborty2024implementing, wang2024qubit}: the linear systems algorithm yields a fourth-power improvement in the condition number of $H$, while the ground-state property estimation algorithm achieves a quadratic improvement in the dependence on the overlap between the initial ``guess” state and the true ground state. 
More broadly, both of our algorithms are resource-efficient, and there exist natural parameter regimes where they even outperform standard, block encoding-based QSVT. 

To complement our theoretical results, we benchmark randomized QSVT for ground-state property estimation on two classes of Hamiltonians: electronic structure problems \cite{mcardle2020quantum} and spin models with long-range interactions \cite{defenu2021longrange}. In both settings, our methods exhibit substantial practical advantages as compared to block encoding-based standard QSVT \cite{lin2020nearoptimalground}, recent early fault-tolerant quantum algorithms \cite{dong2022ground}, as well as other randomized approaches \cite{chakraborty2024implementing, wang2024qubit}. For molecular Hamiltonians such as propane, carbon dioxide, and ethane, randomized QSVT achieves orders-of-magnitude reductions in circuit depth compared to prior work. Similarly, for one-dimensional transverse-field Ising models with long-range or hybrid interactions, we observe improvements in scaling with the system size, resulting in concrete reductions in gate complexity as observed in numerical simulations. These results underscore the practical relevance of randomized QSVT for early fault-tolerant quantum devices, where reducing circuit depth is paramount. Taken together, our results establish a resource-efficient framework for randomized quantum algorithms that is optimal within this access model and broadly applicable to problems of practical interest.

The remainder of this article is organized as follows. In Sec.~\ref{sec:prelim}, we introduce the notation used throughout the paper and review preliminary concepts that form the basis of our algorithms. Sec.~\ref{sec:randomized qsvt} presents our randomized quantum algorithms for QSVT, and Sec.~\ref{sec:lower bound} establishes matching lower bounds. Applications to quantum linear systems and ground-state property estimation are discussed in Sec.~\ref{sec:applications}, and we provide numerical benchmarks of our algorithm for the latter task on Hamiltonians from quantum chemistry and condensed matter physics in Sec.~\ref{sec:numerics}. Finally, Sec.~\ref{sec:conclusion} summarizes our contributions and outlines promising directions for future research.

\section{Preliminaries}
\label{sec:prelim}

In this section, we introduce the notation used throughout this work and the preliminary techniques and algorithms that underlie our results.

\subsection{Notations}
\label{sec:notations}

We use $\mb N=\mb Z_{>0}$ and $\mathbb{T}= \{ z\in \mathbb{C} : |z|=1 \}$. 
We use $\widetilde{O}(\cdot)$ to hide polylogarithmic factors, i.e., $\widetilde{O}(f(n))=O(f(n)\polylog(f(n)))$. 
Bold lowercase letters denote vectors, for example, $\bm{a}$ denotes $(a_1, \ldots, a_n)$. 
$\|\cdot\|_1$ denotes the vector $\ell_1$-norm, i.e., $\|\bm{a}\|_1 = \sum_{i=1}^n |a_i|$.
Unless otherwise specified, $\|A\|$ denotes the spectral norm of an operator $A$. 
For operators $A,B$, we write $A\approx_\varepsilon B$ if they are $\eps$-close in spectral norm, i.e.,\ $\|A-B\|\leq \varepsilon$. 
$M_n(\mathbb{C})$ denotes the space of $n \times n$ complex matrices. 
For $A,B \in M_n(\mathbb{C})$, we define $\ad_B(A)=[B,A]=BA-AB$, so that $e^{\ad_B}(A) = e^B A e^{-B}$. 
The trace of an operator $A$ is denoted by $\Tr[A]$, and $\|A\|_1:=\Tr[\sqrt{A^\dagger A}]$ denotes the Schatten $1$-norm, also called the trace norm.
For a quantum channel $\mc E\colon M_n(\mathbb{C}) \rightarrow M_n(\mathbb{C})$, its diamond norm is defined as
\[
\|\mc E\|_{\diamond}:=  \max_{\rho: \|\rho\|_1 \leq 1} \left\|\left(\mathcal{E} \otimes I_{n}\right)(\rho)\right\|_1,
\]
where $\rho \in M_{n^2}(\mathbb{C})$.
For $f(x):[a,b] \rightarrow  \mathbb{C}$, we denote $\|f\|_{[a,b]} = \max_{x\in[a,b]} |f(x)|.$  

The expectation value of a random variable $V$ is denoted by $\mathbb{E}[V]$. 
The probability of an event $E$ is denoted by $\Pr[E]$.
The Pauli matrices are
\[
I = \begin{bmatrix}
    1 & 0 \\ 0 & 1
\end{bmatrix}, \quad
X = \begin{bmatrix}
     0 & 1 \\ 1 & 0 
\end{bmatrix}, \quad
Y = \begin{bmatrix}
     0 & -\i \\ \i & 0 
\end{bmatrix}, \quad
Z = \begin{bmatrix}
    1 & 0 \\ 0 & -1
\end{bmatrix}.
\]
Any matrix $A \in M_{2^n}(\mathbb{C})$ can be uniquely expressed as a linear combination of tensor products of Pauli matrices:
\[
A = \sum_{P\in\{I,X,Y,Z\}^{\otimes n}}\lambda_P P,
\qquad 
\lambda_P = \frac{1}{2^n} \Tr[A^\dag P].
\] 
If $A$ is Hermitian, then $\lambda_P\in\mathbb{R}$.
For any unitary matrix $U$, there exists a unique skew-Hermitian matrix $V$ with eigenvalues in $ (-\i\pi, \i\pi]$ such that $e^{V}=U$, and we denote $V$ by $\log(U)$.

Throughout this article, the \emph{time complexity} of a quantum algorithm refers to the number of elementary gates used to implement the algorithm.

\subsection{Brief overview of QSVT}

We begin by recalling the standard framework of quantum singular value transformation (QSVT). At its core, QSVT provides a systematic way to implement polynomial transformations of the singular values of a matrix, provided the matrix is embedded as the top-left block of a larger unitary, known as a block encoding. 
This idea generalizes the earlier framework of quantum signal processing (QSP) \cite{low2017qsp}, which achieves polynomial transformations of single-qubit rotations. Formally, QSP is summarized by the following lemma:

\begin{lem}[QSP, Corollary 8 of \cite{gilyen2019quantum}]
\label{lem_complex quantum signal processing}
    Let $P(x) \in \mathbb{C}[x]$ be a degree-$d$ polynomial of parity-$(d \bmod 2)$, satisfying $\|P(x)\|_{[-1,1]} \leq 1$ and $|P(x)| > 1$ for all $x \notin [-1,1]$.
    Then there exists  $\Phi=(\phi_1,\ldots,\phi_d) \in \mathbb{R}^d$ such that  
    \[
    \prod_{j=1}^d\left(R(x) e^{\i  \phi_j Z}\right)=\left[\begin{array}{cc}
    P(x) & * \\
    * & *
    \end{array}\right].
    \]
    where 
    \be \label{single-qubit reflection}
    R(x)=\left[\begin{array}{cc}
    x & \sqrt{1-x^2} \\
    \sqrt{1-x^2} & -x
    \end{array}\right]
    \ee
    is the single-qubit reflection operator defined for $x\in [-1,1]$.
    Moreover, if $P(x)$ is real, then there exists a complex polynomial $\tilde{P}(x)$ with $\Re(\tilde{P}(x))=P(x)$, so that the interleaved gate product can also implement real polynomials.
\end{lem}

Block encoding \cite{chakraborty_et_al:LIPIcs.ICALP.2019.33, low2019hamiltonian, gilyen2019quantum} plays a central role in extending this single-qubit procedure to arbitrary operators:

\begin{defn}[Block encoding]
\label{def:block_encoding}
Let $H$ be an operator acting on $s$ qubits, $\alpha,\varepsilon>0$ and $a\in \mathbb{N}$, then we say that the $(s + a)$-qubit unitary $U$ is an $(\alpha, a, \varepsilon)$-block-encoding of $H$, if
\[
\left\|H - \alpha(\bra{0^a}\otimes I) U (\ket{0^a}\otimes I)\right\| \leq \varepsilon.
\]
That is, 
$$
U \approx_{\eps} \begin{bmatrix}
    H/\alpha & * \\ * & *
\end{bmatrix}.
$$ 
\end{defn}

Given access to such a block encoding, QSVT provides a general procedure to implement polynomial transformations of the singular values of $H$. This is formalized via the notion of singular value transformation:

\begin{defn}[Singular value transformation]
\label{def_svt of matrix}
Let $P\colon \mb{R} \rightarrow \mathbb{C}$ be an even or odd polynomial. Let $H \in \mathbb{C}^{m \times n}$ be a matrix with singular value decomposition $H=\sum_{i=1}^{n} \sigma_i\ket{u_i}\left\langle v_i\right|$ with $\sigma_i:= 0$ when $i> \min(m,n)$. We define the singular value transformation corresponding to $P$ as
\[
P(H):= \begin{cases}  \vspace{.2cm}
\sum_{i=1}^{n}  P\left(\sigma_i\right) \ket{u_i} \bra{v_i} & \text { if } P \text{ is odd, }\\ 
\sum_{i=1}^{n}  P\left(\sigma_i\right) \ket{v_i} \bra{v_i} & \text { if } P \text{ is even. }
\end{cases}
\]
\end{defn}

To achieve such transformations, QSVT employs alternating phase sequences $U_\Phi$, which interleave the block encoding $U$ with controlled single-qubit rotations:

\begin{defn}[Definition 8 of \cite{gilyen2019quantum}] 
\label{def_alternating phase sequence}
    Let $\mathcal{H}_U$ be a finite dimensional Hilbert space and $U, \Pi, \widetilde{\Pi} \in \operatorname{End}\left(\mathcal{H}_U\right)$ be linear operators such that $U$ is unitary and $\Pi, \widetilde{\Pi}$ are orthogonal projectors. Let $\Phi=(\phi_1,\ldots,\phi_d)  \in \mathbb{R}^d$, we define the alternating phase modulation sequences $U_\Phi$ as follows
    \begin{equation}\label{eq_alternating phase circuit}
    U_{\Phi}:= \begin{cases}  \vspace{.2cm}
    e^{\i  \phi_1(2 \widetilde{\Pi}-I)} U 
    \prod_{j=1}^{(d-1) / 2} \left(e^{\i  \phi_{2 j}(2 \Pi-I)} U^{\dagger} e^{\i  \phi_{2 j+1}(2 \widetilde{\Pi}-I)} U\right), & \text { if } d \text { is odd, } \\
    \prod_{j=1}^{d / 2}\left(e^{\i  \phi_{2 j-1}(2 \Pi-I)} U^{\dagger} e^{\i  \phi_{2 j}(2 \widetilde{\Pi}-I)} U\right), 
    & \text { if } d \text { is even.}
    \end{cases}
    \end{equation}
\end{defn}

The key theorem of QSVT shows that such sequences can realize polynomial transformations of the singular values of the block-encoded operator: 

\begin{lem}[QSVT, Theorem 10 of \cite{gilyen2019quantum}]
\label{thm for standard QSVT}
Assume that $P(x)\in \mathbb{C}[x]$ satisfies the conditions in Lemma \ref{lem_complex quantum signal processing}, then there exists $\Phi\in \mathbb{R}^{d}$ such that
\[
P(\widetilde{\Pi} U \Pi) = 
\begin{cases}
\widetilde{\Pi} U_\Phi \Pi & \text{if } d \text{ is odd}, \\
\Pi U_\Phi \Pi & \text{if } d \text{ is even}.
\end{cases}
\]
\end{lem}

Intuitively, if $\Pi=\widetilde{\Pi}=\begin{bmatrix}
    I & 0 \\ 0 & 0
\end{bmatrix},$ 
then Lemma \ref{thm for standard QSVT} states that 
$$
U_\Phi = \begin{bmatrix}
   P(H/\alpha)  & * \\ * & *
\end{bmatrix}, \quad
\text{if } U=\begin{bmatrix}
   H/\alpha  & * \\ * & *
\end{bmatrix}.
$$
Despite its generality, QSVT crucially relies on the efficient construction of block encodings, which is often costly. For Hamiltonians of the form
$$
H=\sum_{k=1}^{L} \lambda_k H_k,
$$
with $\|H_k\|=1$, a block encoding of $H/\lambda$ can be built using the LCU technique \cite{berry2015simulating}, typically requiring $O(\log L)$ ancilla qubits. Recent results show that this ancilla requirement is unavoidable \cite{chakraborty2025quantum}. 
Moreover, the LCU procedure entails a circuit depth of $O(L)$, yielding an overall cost of $\widetilde{O}(L\lambda d)$. These overheads (large ancilla requirements and linear dependence on $L$) constitute the main bottlenecks of standard QSVT and severely limit its practicality on near-term and early fault-tolerant devices. Our first algorithm is a direct randomized variant of QSVT, which replaces block encodings with randomly sampled unitaries.

\subsection{Generalized quantum signal processing}
\label{subsec:prelim-gqsp}

In standard QSVT, the implementable polynomials are restricted to being either complex-even or complex-odd. While real polynomials can also be realized, this typically requires additional ancilla qubits. 
The recently introduced framework of quantum phase processing \cite{wang2023quantum}, or equivalently generalized quantum signal processing (GQSP) \cite{motlagh2024generalized}, significantly extends this setting: by allowing access to controlled Hamiltonian evolution, it enables the implementation of arbitrary polynomials bounded within the complex unit circle, including Laurent polynomials (polynomials with both negative and positive powers) of the time evolution operator, without requiring additional ancilla qubits. Next, we outline key results used in this work and refer the reader to \cite{motlagh2024generalized} for further details.

Formally, suppose we can query the controlled Hamiltonian simulation unitaries
\begin{equation} \label{control U:eq}
c_0\text{-}U =
\begin{bmatrix}
U & 0 \\
0 & I
\end{bmatrix},
\quad
c_1\text{-}U^\dag =
\begin{bmatrix}
I & 0 \\
0 & U^\dagger
\end{bmatrix},
\end{equation}
where $U = e^{\i H}$ and $U^\dagger = e^{-\i H}$. Let
\begin{align*}
    R(\theta, \phi, \lambda)=\left[\begin{array}{cc}
        e^{\i (\lambda+\phi)} \cos (\theta) & e^{\i  \phi} \sin (\theta) \\
        e^{\i  \lambda} \sin (\theta) & -\cos (\theta)
        \end{array}\right] \otimes I, 
\end{align*}
be arbitrary $U(2)$ rotation of the single ancilla qubit. 
Motlagh and Wiebe \cite{motlagh2024generalized} showed that for any degree-$d$ Laurent polynomial $P$ satisfying $|P(x)|\le 1$ on $\mathbb{T}:= \{ x\in \mathbb{C} : |x|=1 \}$, there exists an interleaved sequence of $R(\theta_j, \phi_j, 0)$ and $c_0\text{-}U$, $c_1\text{-}U^\dag$, of length $2d+1$, that implements a block encoding of $P(U)$.  
We restate the result below:

\begin{thm}[Combining Corollary 5 and Theorem 6 of \cite{motlagh2024generalized}]
    \label{thm:generalized-QSP}
    For any Laurent polynomial $P(z) = \sum_{j=-d}^{d} a_{j} z^j$ such that $|P(z)| \le 1$ for all $z\in \mathbb{T}$, there exist $\Theta=(\theta_j)_j, \Phi=(\phi_j)_j \in \mathbb{R}^{2d+1}, \lambda \in \mathbb{R}$ such that 
    \begin{equation}
    \label{eq:gqsp}
        \begin{bmatrix}
            P(U) & * \\
            * & * 
        \end{bmatrix} = \Big(\prod_{j=1}^d R(\theta_{d+j}, \phi_{d+j}, 0) c_1\text{-}U^\dag \Big)\Big(\prod_{j=1}^{d} R(\theta_j, \phi_j, 0) c_0\text{-}U \Big) R(\theta_0, \phi_0, \lambda).
    \end{equation}
\end{thm}

Recently, it was shown that many commonly used functions can be approximated by Laurent polynomials in $e^{\i kx}$, bounded on $\mathbb{T}$, for some constant $k$ \cite{chakraborty2025quantum}. 
Standard QSVT implements functions expressed as polynomials in $x$ that are bounded on $[-1,1]$. Ref.~\cite{chakraborty2025quantum} further establishes the precise conditions under which a polynomial \( Q(x) \) can be approximated by a Laurent polynomial \( P(e^{\i k x}) \), bounded on $\mathbb{T}$, with $\deg(P)$ scaling linearly with $\deg(Q)$. Moreover, they prove that near-optimal complexity can be achieved by implementing the controlled Hamiltonian evolution operators in GQSP using higher-order Trotterization, combined with classical interpolation to exponentially reduce the quantum circuit depth.

Our second method extends this line of ideas by developing a fully randomized implementation of QSVT within the GQSP framework. Specifically, we use qDRIFT (see Sec.~\ref{subsec:prelim-qdrift})\cite{campbell2019random} for implementing Hamiltonian evolution, combined with a novel use of classical extrapolation. 
The resulting end-to-end randomized QSVT algorithm (i) avoids block encodings entirely, (ii) requires only a single ancilla qubit, and (iii) achieves circuit depth independent of the number of Hamiltonian terms, thereby improving on the complexity achieved in \cite{chakraborty2025quantum}.

\subsection{Time-ordered evolution in the interaction picture}
\label{subsec:prelim-interaction-picture}

Let $\mc H(\tau) = \mc A(\tau) + \mc B(\tau)$ be a continuous operator-valued function. The time-evolution of $\mc H(\tau)$ in the interaction picture of quantum mechanics, can be formally expressed via the following lemma.

\begin{lem}[\kern-1.5mm{\cite[Page 21]{dollard1984product}}]
\label{lem:interaction-pic}
    Let $\mc H(\tau) = \mc A(\tau) + \mc B(\tau)$ be an operator-valued function defined for $\tau \in \mathbb{R}$ with continuous summands $\mc A(\tau)$ and $\mc B(\tau)$. Then
    {\small\begin{align*}
    \exp_{\mc T}\left(\int_{0}^{t}  \mc H(\tau)\, \d \tau\right) = & \exp_{\mc T}\left(\int_{0}^{t}  \mc A(\tau)\, \d \tau\right) \\
    & \cdot \exp_{\mc T}\left(\int_{0}^{t}\exp_{\mc T}^{-1}\left(\int_{0}^{\tau_1}  \mc A(\tau_2)\, \d \tau_2\right)\mc  B(\tau_1) \exp_{\mc T}\left(\int_{0}^{\tau_1}\mc  A(\tau_2)\, \d \tau_2\right)\, \d \tau_1\right).
    \end{align*}}
\end{lem} 

The lemma, while stated for continuous functions, also holds for piecewise continuous functions, such as the ones we use in the derivations of the correctness of our second randomized quantum algorithm for QSVT (see Lemma \ref{lem_QSP+qDRIFT}).

\subsection{Richardson extrapolation}
\label{subsec:prelim-richardson}

Richardson extrapolation is a classical technique for estimating limits of the form $f(0)=\lim_{x\to 0} f(x)$ by evaluating $f$ at a finite set of nonzero points $x\in\{s_1,s_2,\ldots,s_m\}$ and combining these values with carefully chosen coefficients.
The key idea is to systematically cancel the leading error terms in the expansion of $f(x)$, thereby achieving higher-order accuracy without a corresponding increase in computational effort.
Originally developed for numerical analysis, Richardson extrapolation has recently found applications in quantum computing, in particular for suppressing discretization errors in Hamiltonian simulation via Trotter methods and qDRIFT \cite{watson2024randomly,watson2024exponentially}.

\begin{lem}[Richardson extrapolation, see \cite{watson2024randomly,low2019well}]
    Suppose that function $f\colon \mb R \to \mathbb{R}$ has a series expansion
    \begin{align*}
        f(x) = f(0) + \sum_{i=1}^{\infty} c_i x^{i}.
    \end{align*}
    Choose an initial point $s\in (0,1)$ and set points $s_i=s r_m/r_i$ with positive integers $r_i$ for $i\in [m]$. Then, using these $m$ sample points $\{s_i: i\in[m]\}$, we can derive an $m$-term Richardson extrapolation
    \be
        F^{(m)}(s)=\sum_{i=1}^m b_i f(s_i)
    \ee
    by choosing appropriate coefficients $b_i$ to cancel the first $m-1$ terms in the series expansion,
    such that
    \bes
        |F^{(m)}(s)-f(0)|\le \|\b\|_1 \left|R_{m}(s)\right|=O(s^{m}).
    \ees
    where $\b=(b_1,\ldots,b_m)$, and $R_{m}(x)$ is a function that only has terms of order $O(x^{m})$ and above.
    In particular, let
    \be 
        r_{i}=\left\lceil\frac{\sqrt{8} m}{\pi \sin (\pi(2 i-1) / 8 m)}\right\rceil^2,
    \qquad 
        b_i=\prod_{\ell \neq i} \frac{1}{1-r_\ell/r_i},
    \ee 
    then $\|\b\|_1=O(\log m), \max_i\{r_i\}=O(m^4)$, and $\max_i \{r_i/r_m\}=O(m^2)$.
\end{lem}
Here, we modify $r_i$ by first squaring it and then applying the ceiling function, using a rescaling technique similar to that in \cite{mohammadipour2025reducing}. This reduces the rounding error and ensures $t/s_j\in \mb Z$ for any given $t>0$.

\begin{lem}
\label{lem_richardson extrapolation}
    Suppose that function $f\colon \mb R \to \mathbb{R}$ can be decomposed as 
    \begin{align*}
        f(x) = P_m(x)+R_m(x),
    \end{align*} 
    where $P_m(x)$ is a degree-$(m-1)$ polynomial and $R_m(x)$ is the remainder. Given any $t > 0$ and $s\in (0,1)$, we can choose 
    \begin{align*}
        K = \max\Big\{\frac{m}{\pi}, \frac{2t}{s}\Big\}, \quad r_j =\left\lceil \frac{K}{\sin^2(\pi(2j-1)/8m)}\right\rceil,
    \end{align*}
    and
    \begin{align*}
        s_j = \frac{t}{r_j}, \quad b_j=\prod_{\ell \neq j} \frac{1}{1-r_\ell/r_j}
    \end{align*}
    so that $\|\b\|_1 \le C\log m$ for some absolute constant $C$, $r_j=O(\max\{m^3, m^2t/s\}/j^2)$, and $\max_{j}\{ s_j \}\le s$. Then define the $m$-term Richardson extrapolation \begin{align}
        \label{eq_richardson extrapolation}
        F^{(m)}(s)=\sum_{j=1}^m b_j f(s_j),
    \end{align}
    which satisfies \begin{align}
    \label{eq:error-F}
        |F^{(m)}(s)-f(0)|\le \|\b\|_1 \max_j|R_{m}(s_j)|.
    \end{align}
\end{lem}
\begin{proof}
    The choice of $b_j$ ensures to cancel all terms in $P_m(x)$, which directly implies Eq.~\eqref{eq:error-F}.
    Also, we have \begin{align*}
        r_j =  O(m^2K/j^2) = O(\max\{m^3, m^2t/s\}/j^2),
    \end{align*}
    and \begin{align*}
        \max_j \{s_j\} = s_m \le t\left\lceil \frac{K}{\sin^2(\pi/4)}\right\rceil^{-1} \le \frac{t}{2K} \le s.
    \end{align*}
    
    Now, we upper bound the coefficients $b_i$. 
    Let \begin{align*}
        r_j' :=  \frac{1}{\sin^2(\pi(2j-1)/8m)}= \Theta(m^2/j^2),
    \end{align*} and then by \cite[Theorem 1]{low2019well}, the corresponding coefficients $b_j'$ satisfy \begin{align*}
        |b_j'| \le \bigg|\frac{(-1)^{j+1}}{m} \cot\Big(\frac{\pi(2 j-1)}{8 m}\Big)\bigg| = O(1/j).
    \end{align*} 
    Note that $r_j$ satisfies \begin{align*}
        r_j = {\lceil K r_j'\rceil}.
    \end{align*}
    We first prove that for any $j\neq j'$, $r_j\neq r_{j'}$. For any $j\neq j'$, we have \begin{align*}
        |Kr_{j}'-Kr_{j'}'|\ge K \min_{x\in[1,m]}\bigg|\frac{\pi}{4m}\frac{2\cos(\pi(2 x-1)/8m)}{\sin^3(\pi(2 x-1)/8m)}\bigg|>\frac{\pi}{m}K >1.
    \end{align*}
    Therefore, for any $j\neq j'$, $r_j\neq r_j'$, and \begin{align*}
        r_j = Kr_j'+O(1),
    \end{align*}
    for any $j\in [m]$. 
    The ratio $\gamma_q' := r_q'/r_j'$ satisfies \begin{align*}
        \frac{r_q}{r_j} = \frac{Kr_q'+O(1)}{Kr_j'+O(1)} = \frac{r_q'}{r_j'}(1+O(1/(mr_j')+1/(mr_q')) = \frac{r_q'}{r_j'}\Big(1+O\Big(\frac{j^2+q^2}{m^3}\Big)\Big).
    \end{align*}
    Then we propagate the error to $b_j$ using \begin{align*}
        \frac{\partial b_j'}{\partial \gamma_q'} = b_j'\frac{1}{1-\gamma_q'}.
    \end{align*}
    The relative error of $b_j$ satisfies \begin{align*}
        \frac{|b_j-b_j'|}{|b_j'|} &= O\Big(\sum_{q\neq j} \frac{1}{(1-r_q'/r_j')}\frac{r_q'(j^2+q^2)}{r_j'm^3}\Big)\\
        &=O\Big(\sum_{q\neq j} \frac{1/r_j'}{1/r_q'-1/r_j'}\frac{j^2+q^2}{m^3}\Big) \\
        &=O\Big(\sum_{q\neq j} \frac{j^2}{|q^2-j^2|}\frac{j^2+q^2}{m^3}\Big)\\
        &=O\Big(\sum_{q\neq j} \frac{j^2}{|q-j|}\frac{j+q}{m^3}\Big).
    \end{align*}
    Therefore, the additive error of $\|\b\|_1$ can be bounded as \begin{align*}
        \|\b-\b'\|_1  &= O\Big(\sum_{j=1}^m |b_j'|\Big(\sum_{q=1}^{j-1} \frac{j^2}{j-q}\frac{j}{m^3}+\sum_{q=j+1}^m \frac{j^2}{q-j}\frac{q}{m^3}\Big)\Big)\\
        &=O\Big(\sum_{j=1}^m \frac{1}{j}\Big(\sum_{q=1}^{j-1} \frac{j^3}{(j-q)m^3}+\sum_{q=j+1}^m \frac{j^2q}{(q-j)m^3}\Big)\Big)\\
        &=O\Big(\sum_{j=1}^m \frac{1}{j}\Big(\frac{j^3}{m^3}\log j+\sum_{\ell=1}^{m-j} \frac{j^2(j+\ell)}{\ell m^3}\Big)\Big)\\
        & = O\Big(\sum_{j=1}^m \frac{j^2}{m^3}\log j+ \sum_{j=1}^{m}\frac{j^2}{m^3}\log(m-j+1)+\sum_{j=1}^{m}\frac{j(m-j)}{m^3}\Big)\\
        & = O(\log m).
    \end{align*}
    In conclusion, we have \begin{align*}
        \|\b\|_1 \le \|\b-\b'\|_1+\|\b '\|_1 = O\Big(\log m+\sum_{j=1}^m \frac{1}{j}\Big) = O(\log m),
    \end{align*}
    and there is an absolute constant $C$ such that $\|\b \|_1\le \log m$.
\end{proof}

\subsection{Hamiltonian simulation by qDRIFT}
\label{subsec:prelim-qdrift}

qDRIFT is a randomized algorithm for Hamiltonian simulation that uses importance sampling to achieve a circuit depth independent of the number of Hamiltonian terms \cite{campbell2019random}. This is in sharp contrast to Trotter–Suzuki methods \cite{lloyd1996universal, childs2021theory}, whose circuit depth always scales linearly with the number of Hamiltonian terms, regardless of the approximation order. 
We provide a brief overview of qDRIFT below, and refer the reader to \cite{campbell2019random, chen2021concentration} for further details.

Suppose
$$
H = \sum_{k=1}^L \lambda_k H_k
$$ 
is a Hermitian matrix with $\|H_k\|=1$, for all $k\in[L]$. For $\lambda=\sum_{k=1}^L |\lambda_k|$ and the distribution
$$
\mathcal{D} = 
\left\{\left(|\lambda_k|/{\lambda}, H_k\right): k\in[L]\right\},
$$
the qDRIFT procedure (outlined in Algorithm \ref{alg: qDRIFT}) randomly and independently samples matrices $H_{j_1},\ldots,H_{j_\ell}$ from $\mathcal{D}$, and applies
$
\prod_{r=1}^\ell e^{\i \lambda H_{j_r} t },
$
to an arbitrary input density operator $\rho$. With a suitable choice of $t$, this randomized evolution approximates the target evolution $e^{-\i HT}\rho e^{\i HT}$.

\begin{breakablealgorithm}
\caption{qDRIFT protocol}
\begin{algorithmic}[1]
\REQUIRE A Hermitian operator $H=\sum_{k=1}^{L} \lambda_k H_{k}$ with $\|H_k\|=1$. 

\qquad Classical sample access to the distribution $\mathcal{D} = \{(|\lambda_k|/\lambda, H_k): k\in[L]\}$, where $\lambda=\sum_k |\lambda_k|$.

\qquad A time-step parameter $t<1/2\lambda$.

\qquad An initial density operator $\rho$.

\ENSURE $\mc E(\rho)$, where $\mathcal{E}:= \sum_k p_k e^{-\i t  \lambda \mathrm{ad}_{H_k}}$ with $p_k=|\lambda_k|/\lambda$.

\STATE Compute $N=\lceil T/t \rceil$.

\STATE For $r=1,\ldots,N$

\quad Sample $H_{j_r}$ from the distribution $\mc D$.

\quad Let $\rho \leftarrow e^{-\i t \lambda H_{j_r}} \rho e^{\i t \lambda H_{j_r}}$

\STATE Return $\rho$.

\end{algorithmic}
\label{alg: qDRIFT}
\end{breakablealgorithm}

Mathematically, qDRIFT defines the quantum channel
\[
\mathcal{E}(\rho)=\sum_{k=1}^L p_k e^{-\i \lambda H_k t} \rho e^{\i \lambda H_k t} 
= \sum_{k=1}^L p_k e^{-\i t \lambda \mathrm{ad}_{H_k}}(\rho)
\]
for a time-step $t$, where $p_k=\lambda_k/\lambda$. Campbell proved that \cite[(B12) in the appendix]{campbell2019random} 
\[
\left\|\mathcal{E} - e^{-\i  {\rm ad}_H t} \right\|_{\diamond}
= O(t^2).
\]
Moreover, if we take $t=T/N$ and apply the qDRIFT channel $N$ times, the total error is bounded by
\[
\left\|\mathcal{E}^N - e^{-\i  {\rm ad}_H T} \right\|_{\diamond} 
= O\left({(\lambda T)^2}/{N}\right).
\]
This leads to a circuit depth of $O(\lambda^2t^2/\eps)$, which is independent of $L$. This independence makes qDRIFT particularly advantageous (over Trotter methods) in settings where $\lambda \ll L$, as is often the case in quantum chemistry \cite{campbell2019random, mcclean2020openfermion, babbush2018low}. 

The recent work of \cite{watson2024randomly} shows that if the goal is to measure the expectation value of any observable with respect to the time-evolved state, this dependence on precision can be exponentially improved from $1/\eps$ to $\mathrm{polylog}(1/\eps)$ with the help of classical extrapolation techniques. More precisely, Ref.~\cite{watson2024randomly} considered the parametrized qDRIFT channels with free parameter $s$ and step size $sT$:\footnote{Here, we use channel $\mc E_s$ to approximate $e^{\i s\ad_H}$ instead of $e^{-\i s\ad_H}$ for our convenience.}
\be
\label{eq_para qDRIFT channel}
    \mathcal{E}_s(\rho):= \sum_k p_k e^{\i s T \lambda \mathrm{ad}_{H_k}}(\rho).
\ee
It was shown in \cite{watson2024randomly} that for $sT < 1/2\lambda$, the channel $\mc E_s$ has the following form.

\begin{lem}[Section 4.2 and Lemma 10 of \cite{watson2024randomly}]
\label{lem:para qdrift channel}
If $s T<1/2\lambda$, then the channel $\mc E_s$ can be written as
\begin{equation}\label{exp form}
        \mathcal{E}_s=e^{\i s T \mathcal{G}(s)}, \quad \text{ where } \; \mathcal{G}(s)=\operatorname{ad}_H+\sum_{j=1}^{\infty} (sT)^j \mc D_{j+1},
\end{equation}
where $\left\{\mc D_{j+1}\right\}_{j=1}^{\infty}$ are superoperators satisfying $\left\|\mc D_j\right\|_{\diamond} \leq(4 \lambda)^j.$
\end{lem}

Then the difference between the true time evolved density matrix, and the one obtained by applying the qDRIFT channel can be written as a power of the step size as shown below (we set $T=1$):

\begin{lem}[Lemma 3 of \cite{watson2024randomly}]
\label{lem:watson qdrift}
    Let $K\geq 3$ be an integer.
    For any density operator $\rho$, we have
    \[
    \mathcal{E}_s^{1 / s}(\rho)- e^{\i H } \rho e^{-\i H } =\sum_{j \geq 1} s^j \tilde{\mc D}_{j+1, K}(\rho)
    +\tilde{\mc D}_{K,s} (\rho),
    \]
    where the diamond norm of the superoperators satisfies
    \[
    \begin{aligned}
    \big\|\tilde{\mc D}_{j+1, K}\big\|_{\diamond} & \leq(8 \lambda )^j \sum_{\ell=1}^{\min \{K-1, n\}} \frac{(8 \lambda )^\ell}{\ell!} , \\
    \big\|\tilde{\mc D}_{K,s} \big\|_{\diamond} & \leq \frac{(8 \lambda )^K}{K!} \sum_{j=1}^{\infty}(8 \lambda s )^j.
    \end{aligned}
    \]    
\end{lem}

Parametrizing the error as a polynomial in $s$ enables the use of Richardson extrapolation to exponentially reduce the circuit depth.

\section{Randomized algorithms for QSVT}
\label{sec:randomized qsvt}
In this section, we present two randomized quantum algorithms for QSVT.
Consider an $n$-qubit Hamiltonian expressed as a linear combination of local terms,
$$
H=\sum_{k=1}^{L}\lambda_k H_k, \quad \lambda=\sum_{k=1}^L |\lambda_k| .
$$
For the first algorithm, we assume that each $H_k$ is unitary and can be implemented efficiently on a quantum computer, using constant-depth circuits over elementary gates. This setting captures the vast majority of physically relevant Hamiltonians; in particular, in quantum chemistry and condensed matter physics, the terms $H_k$ are typically tensor products of Pauli operators, $H_k\in\{I,X,Y,Z\}^{\otimes n}$. The second algorithm applies to a more general setting where each $H_k$ is an arbitrary Hermitian operator of unit norm, provided that $e^{\i H_k}$ can be implemented efficiently.

Our key assumption is sampling access to the decomposition: namely, we can efficiently obtain i.i.d.\ samples from the distribution
$$
\mathcal{D} = 
\left\{\left(\left |\lambda_k\right |/\lambda, H_k\right): k\in[L]\right\}.
$$
This sampling model is what makes our algorithms inherently randomized. Given this access model, our goal is to implement a target function $f(H)$, approximated by a bounded polynomial of degree $d$, which is the central task of QSVT. 
Concretely, for any observable $O$ and initial state $\ket{\psi_0}$, we design randomized procedures that output an estimate $\mu$ such that
$$
\big|\mu-\braket{\psi_0|f(H)^{\dag} O f(H)|\psi_0}\big| \leq \eps\|O\|.
$$
Our challenge is to realize this while (i) avoiding block encodings, (ii) minimizing the number of ancilla qubits, and (iii) achieving circuit depth independent of the number of terms $L$.

\subsection{A direct randomization of QSVT}
\label{subsec:randomized-qsvt-1}
We begin by developing a method that can be viewed as a direct randomization of QSVT. Recall that in the standard framework, given a block encoding $U$ (see Definition~\ref{def:block_encoding}) of a Hermitian operator $H$ with $\|H\|\leq 1$ and a degree-$d$ even or odd polynomial $P(x)$, there exists a phase vector $\Phi=(\phi_1,\ldots,\phi_d)\in \mathbb{R}^d$ such that the alternating phase sequence
\be \label{phase sequence}
U_{\Phi} := 
\begin{cases}  \vspace{.2cm}
    e^{\i  \phi_1 Z} U 
    \prod_{j=1}^{(d-1) / 2} \left(e^{\i  \phi_{2 j}Z} U^{\dagger} e^{\i  \phi_{2 j+1}Z} U\right), & \text { if } d \text { is odd, } \\
    \prod_{j=1}^{d / 2}\left(e^{\i  \phi_{2 j-1}Z} U^{\dagger} e^{\i  \phi_{2 j}Z} U\right), 
    & \text { if } d \text { is even,}
\end{cases}
\ee
is itself a block encoding of $P(H)$. The phase sequence $\Phi$ is determined entirely by the target polynomial $P(x)$ and can be computed efficiently on a classical computer, e.g., see \cite{chao2020finding, dong2021efficient}.

However, in the sampling access model considered in this work, constructing a block encoding of $H$ is infeasible, since we only have access to individual terms $H_k$ sampled according to their weights $|\lambda_k|/\lambda$. This motivates a different approach. Instead of a block encoding, we introduce a simpler operator of the form
\be
U = \begin{bmatrix}
    cI & sH \\ sH^\dag & -cI
\end{bmatrix}, \qquad c,s\in (0,1),~c^2+s^2=1.
\ee
and study its behavior within the alternating phase sequence in Eq.~\eqref{phase sequence}. 
Although $U$ is not unitary in general, it becomes unitary whenever $H$ is unitary, and more importantly, it serves as a randomized analogue of a block encoding of $H$.
Concretely, by sampling a random term $H_k$ from $\mc D$, we can efficiently construct the unitary
\be \label{unitary Uk}
U_k=\begin{bmatrix}
    cI & sH_k \\ sH_k^\dag & -cI
\end{bmatrix},
\ee
satisfying $\mathbb{E}[U_k]=U$. Each run of the algorithm thus implements a randomized circuit: an interleaved sequence of single-qubit rotations and unitaries $U_k$, which in expectation reproduces the effect of the sequence $U_\Phi$. 

With this setup, we first provide a general characterization of alternating phase-modulation sequences $U_\Phi$ that incorporate the operator $U$, and demonstrate how such sequences can be used to implement polynomial transformations of $H$. In fact, the following theorem, which holds for arbitrary matrices (not necessarily Hermitian or square), can be viewed as the analogue of Lemma~\ref{thm for standard QSVT} in the standard QSVT framework.

\begin{thm}
\label{prop_non-unitary SVT circuit of A}
    Let $P(x)\in \mathbb{C}[x]$ be a degree-$d$ polynomial 
    such that $|P(x)| \leq 1/2$ for all $x\in[-1,1]$ and has parity-$(d \mod 2)$. Let $m, c,s,\alpha \in \mathbb{R}$ be such that
    \be
    \label{eq: some parameters}
    m = \Theta\!\left(d\cdot \log^2\!\left(d/\eps\right) \log^2\!\left(\log\left(d/\eps\right)\right)\right),  
    \quad  
    s = 1/\sqrt{m},  
    \quad  
    c = \sqrt{1-1/m},  
    \quad  
    \alpha = \sqrt{m-1}\,\log\!\left(d^2/\varepsilon\right).
\ee
Let $H\in\mathbb{C}^{k\times \ell}$ be a matrix with $\|H\|\leq 1$.
    Define
    \[
    U:= \begin{bmatrix}
        cI_k & sH \\ 
        sH^\dagger & -cI_\ell
    \end{bmatrix} \qquad 
    \widetilde{\Pi}:= \begin{bmatrix}
    I_k & 0 \\ 0 & 0
    \end{bmatrix}  \qquad  
    \Pi:= \begin{bmatrix}
    0 & 0 \\ 0 & I_{\ell}
    \end{bmatrix}.
    \]
    Then there exists $\Phi\in \mathbb{R}^m$ such that
    \[
    \begin{aligned}
        \ket{0}\bra{1}\otimes P(H/\alpha )& \approx_\varepsilon \widetilde{\Pi} U_{\Phi} \Pi &&\;\text{ if } m \text { is odd,} \\
        \ket{1}\bra{1} \otimes P(H/\alpha) & \approx_\varepsilon \Pi U_{\Phi} \Pi &&\;\text{ if } m \text { is even,} 
    \end{aligned}
    \]
    where $U_\Phi$ is the alternating phase sequence defined as

\begin{equation}
    U_{\Phi}:= \begin{cases}  \vspace{.2cm}
    e^{\i  \phi_1(2 \widetilde{\Pi}-I)} U 
    \prod_{j=1}^{(m-1) / 2} \left(e^{\i  \phi_{2 j}(2 \Pi-I)} U^{\dagger} e^{\i  \phi_{2 j+1}(2 \widetilde{\Pi}-I)} U\right), & \text { if } m \text { is odd, } \\
    \prod_{j=1}^{m / 2}\left(e^{\i  \phi_{2 j-1}(2 \Pi-I)} U^{\dagger} e^{\i  \phi_{2 j}(2 \widetilde{\Pi}-I)} U\right), 
    & \text { if } m \text { is even.}
    \end{cases}
    \end{equation}
\end{thm}

The above result realizes polynomial transformations of $H/\alpha$, in analogy with standard QSVT. Specifically,
$$
U_\Phi \approx_{\eps} \begin{bmatrix}
         * & P(H/\alpha ) \\ * & *
     \end{bmatrix},
$$
if $m$ is odd and
$$
U_\Phi \approx_{\eps} \begin{bmatrix}
         * & * \\ * & P(H/\alpha )
     \end{bmatrix},
$$ 
if $m$ is even. 
By conjugating $U_\Phi$ with appropriate swap operators, we can place $P(H/\alpha )$ in the top-left block in both cases.

\begin{proof}
We consider the case when $d$ is even. The odd case can be proved similarly.
Let $H=\sum_{i=1}^{\min(k,\ell)} \sigma_i \ket{u_i} \bra{v_i}$ be the singular value decomposition of $H$. 
    Denote 
    $$\ket{\psi_i}=\begin{bmatrix}
    0 \\ \ket{v_i}
    \end{bmatrix}, \qquad\ket{\tilde{\psi_i}}=\begin{bmatrix}
    \ket{u_i} \\ 0
    \end{bmatrix},$$ then we have $\Pi \ket{\psi_i}=\ket{\psi_i}, \widetilde{\Pi} \ket{\tilde{\psi_i}}=\ket{\tilde{\psi}_i}$, and $\widetilde{\Pi} U \Pi = s \sum_{i=1}^{\min(k,\ell)} \sigma_i \ket{\tilde{\psi}_i} \bra{\psi_i}$.
    For each $i$, define two vectors 
    \[
    \ket{\psi_i^{\perp}}:= \frac{(I-\Pi) U^{\dagger}\ket{\tilde{\psi}_i}}{\|(I-\Pi) U^{\dagger}\ket{\tilde{\psi}_i}\|}=\begin{bmatrix}
    \ket{u_i} \\ 0
    \end{bmatrix},
    \qquad
    \ket{\tilde{\psi}_i^{\perp}}:= \frac{(I-\widetilde{\Pi}) U\left|\psi_i\right\rangle}{\|(I-\widetilde{\Pi}) U\left|\psi_i\right\rangle \|}=\begin{bmatrix}
    0 \\ -\ket{v_i}
    \end{bmatrix}.
    \]
    Observe that\footnote{Note that $U$ is not unitary generally, hence the output quantum states are no longer normalized.}
    \beas
    U \ket{\psi_i} &=& s\sigma_i \ket{\tilde{\psi_i}} + c\ket{\tilde{\psi}_i^\perp} , \\
    U \ket{\psi_i^\perp} &=& c\ket{\tilde{\psi_i}} - s\sigma_i \ket{\tilde{\psi}_i^\perp} .
    \eeas
    Let $\mathcal{H}_i:= \operatorname{Span}\{\left|\psi_i\right\rangle,\left|\psi_i^{\perp}\right\rangle\}$ and $\tilde{\mathcal{H}}_i:= \operatorname{Span}\{\ket{\tilde{\psi}_i},\ket{\tilde{\psi}_i^{\perp}}\}$ for $i\in [\min(k,\ell)]$, then
    \begin{equation}\label{eq_matrix1}
    U=\bigoplus_{i\in [\min(k,\ell)]}  \left[\begin{array}{cc}
    s\sigma_i & c \\
    c & -s\sigma_i
    \end{array}\right]_{\tilde{\mathcal{H}}_i}^{\mathcal{H}_i} .
    \end{equation}
    We also have
    \begin{equation}\label{eq_matrix2}
    e^{\i  \phi(2 \Pi-I)}=\bigoplus_{i\in [\min(k,\ell)]} \left[\begin{array}{cc}
    e^{\i  \phi} & 0 \\
    0 & e^{-\i  \phi}
    \end{array}\right]_{\mathcal{H}_i}^{\mathcal{H}_i},
    \qquad 
    e^{\i  \phi(2 \widetilde{\Pi}-I)}=\bigoplus_{i\in [\min(k,\ell)]} \left[\begin{array}{cc}
    e^{\i  \phi} & 0 \\
    0 & e^{-\i  \phi}
    \end{array}\right]_{\tilde{\mathcal{H}}_i}^{\tilde{\mathcal{H}}_i}.
    \end{equation}  
    Note that 
    \[
    \begin{aligned}
        \tilde{R}(x)& :=  \begin{bmatrix}
        sx & c \\ c & -sx
        \end{bmatrix}\\
        &=\sqrt{c^2+s^2x^2} \cdot \begin{bmatrix}
        sx/\sqrt{c^2+s^2x^2} & c/\sqrt{c^2+s^2x^2} \\ c/\sqrt{c^2+s^2x^2} & -sx/\sqrt{c^2+s^2x^2}
        \end{bmatrix} \\
        &=\sqrt{c^2+s^2x^2} \cdot R\Big(\frac{sx}{\sqrt{c^2+s^2x^2}}\Big) ,
        \end{aligned}
    \]
where $R$ is the single-qubit reflection operator, see Eq.~\eqref{single-qubit reflection}.
For the given polynomial $P(x)$, we can derive a polynomial $\tilde{P}(z)$ according to Proposition \ref{prop_implement f(sx)} such that
\[
\left|\left(\sqrt{c^2+s^2 x^2}\right)^m \tilde{P}\Big(\frac{sx}{\sqrt{c^2+s^2 x^2}}\Big)-P(x/\alpha)\right|\le \varepsilon.
\]
Moreover, $|\tilde{P}(z)| \leq 1$ for $z=sx/\sqrt{c^2+s^2x^2}\in [-s,s]$, and $\tilde{P}$ has the same parity as $P(x)$. 
We may consider $\tilde{P}(z)$ as a polynomial in $z/s \in [-1,1]$. 
Then, by Lemma \ref{lem_complex quantum signal processing}, there exists a sequence of phase angles $\Phi=(\phi_1,\ldots,\phi_m)\in \mathbb{R}^m$ such that
\be \label{eq_matrix3}
\prod_{j=1}^{m}\left(e^{\i  \phi_j Z} R\Big(\frac{sx}{\sqrt{c^2+s^2x^2}}\Big)\right)=
\left[\begin{array}{cc}
\tilde{P}\left(\frac{sx}{\sqrt{c^2+s^2x^2}}\right) & * \\
* & *
\end{array}\right].
\ee
Combining Eqs.~\eqref{eq_matrix1}, \eqref{eq_matrix2} and \eqref{eq_matrix3}, we obtain
\begin{equation*}
\begin{aligned}
    U_{\Phi} &=\bigoplus_{i\in [\min(k,\ell)]}  \left[\prod_{j=1}^{m/2}\left( e^{\i  \phi_{2j-1} Z} \tilde{R}(\sigma_i) e^{\i  \phi_{2j} Z} \tilde{R}(\sigma_i) \right) \right]_{\mathcal{H}_i}^{\mathcal{H}_i} \\
    &= \bigoplus_{i\in [\min(k,\ell)]}  \left[\prod_{j=1}^{m}\left(e^{\i  \phi_{j} Z} \sqrt{c^2+s^2{\sigma_i}^2} \cdot R\Big(\frac{s\sigma_i}{\sqrt{c^2+s^2{\sigma_i}^2}}\Big) \right) \right]_{\mathcal{H}_i}^{\mathcal{H}_i} \\
    &= \bigoplus_{i\in [\min(k,\ell)]} \left[ \left(\sqrt{c^2+s^2{\sigma_i}^2}\right)^{m} \prod_{j=1}^{m} \left(e^{\i  \phi_{j} Z} R\Big(\frac{s\sigma_i}{\sqrt{c^2+s^2{\sigma_i}^2}}\Big) \right) \right]_{\mathcal{H}_i}^{\mathcal{H}_i}\\
    &\approx_\varepsilon \bigoplus_{i\in [\min(k,\ell)]} \left[\begin{array}{cc}
    P(\sigma_i/\alpha) & * \\
    * & *
    \end{array}\right]_{\mathcal{H}_i}^{\mathcal{H}_i}.
\end{aligned}
\end{equation*} 
Hence
\[
    \Pi U_\Phi \Pi \approx_\varepsilon \bigoplus_{i\in [\min(k,\ell)]} \left[\begin{array}{cc}
    P(\sigma_i/\alpha) & 0 \\
    0 & 0
\end{array}\right]_{\mathcal{H}_i}^{\mathcal{H}_i} 
= \sum_{i=1}^{\min(k,\ell)} P(\sigma_i/\alpha) \ket{\psi_i} \bra{\psi_i} 
= \ket{1} \bra{1} \otimes P(H/\alpha).
\]
This completes the proof.
\end{proof}

We first note that Theorem~\ref{prop_non-unitary SVT circuit of A} implements polynomial transformations on the singular values of an arbitrary operator 
$H$. 
In the special case where 
$H$ is Hermitian, these transformations act directly on its eigenvalues. 
From the proof, we see that the phase $\Phi$ can be computed using any of the previous algorithms for QSVT (e.g., \cite{chao2020finding, dong2021efficient}), applied to the polynomial $\tilde{P}(z)$ viewed as a polynomial in $z/s \in [-1,1]$.
The parameter $s$ in our construction plays a role analogous to the normalization factor $1/\alpha$ in block encodings (see Definition~\ref{def:block_encoding}), and directly impacts the overall complexity of the algorithm. 
Its choice is constrained by QSVT requirements: specifically, we need $(c^2+s^2x^2)^{m/2}$ to remain bounded by a constant for all $x \in [-1,1]$. Under this constraint, our choice of $s$ is essentially optimal. 
One might ask whether an alternative construction of $U$ could enable a significantly larger $s$ and hence reduce complexity further. However, our lower bound analysis shows that this is impossible in general, establishing the optimality of our approach.   

Naturally, the operators $U$ and $U_\Phi$ from Theorem~\ref{prop_non-unitary SVT circuit of A} are not unitary in general. 
To overcome this, we use randomized constructions of $U$ to recover the desired transformation. 
Specifically, by sampling terms $H_k$ of $H$ from the distribution $\mc D$ via importance sampling, we construct the unitary matrices $U_k$ defined in Eq.~\eqref{unitary Uk}, which satisfy $\mathbb{E}[U_k]=U$. 
We then prove that substituting these $U_k$ into the alternating phase sequence yields, in expectation, the exact operator $U_\Phi$. 
This key observation is formalized in the following lemma, stated for the normalized operator $H/\lambda$.

\begin{lem}\label{lem_expectation of Pk}
    Let $H$ be an operator
    written as a linear combination of unitaries 
    $$H=\sum_{k=1}^L \lambda_k H_k,$$ 
    with $\sum_k |\lambda_k|=1$. 
    Let $\mathcal{D}=\{(|\lambda_k| , H_k):k\in[L]\}$ denote the probability distribution that outputs $H_k$ with probability $|\lambda_k|$.
    Define unitary
    \[
    U_k:= \begin{bmatrix}
        cI & sH_k \\ s H_k^\dagger & -cI
    \end{bmatrix}.
    \]
    Let $U_\Phi^{(R)}$ be the random circuit obtained by replacing each $U$ in Eq.~\eqref{eq_alternating phase circuit} with an i.i.d. sample $U_k$ from $\mathcal{D}$. 
    Then $\mathbb{E}[U_\Phi^{(R)}]=U_\Phi$. 
\end{lem}
\begin{proof}
    Note that $\mb{E}[U_k]=U$ for $k\in [L]$. Since $H_k$ is sampled independently, for even $m$ we have
    \[
    \begin{aligned}
    \mathbb{E}[U_\Phi^{(R)}]&=\prod_{j=1}^{m/2} \left(e^{\i  \phi_{2 j-1}(2 \Pi-I)} \mathbb{E}[U_{k_j}^\dag] e^{\i  \phi_{2 j}(2 \widetilde{\Pi}-I)} \mathbb{E}[U_{k_j'}] \right) \\
    &=\prod_{j=1}^{m/2} \left(e^{\i  \phi_{2 j-1}(2 \Pi-I)} U^\dag e^{\i  \phi_{2 j}(2 \widetilde{\Pi}-I)} U \right)\\
    &=U_\Phi.
    \end{aligned}
    \]
    The case for odd $m$ is similar.
\end{proof}

As a direct corollary of the above lemma, we have the following result.

\begin{prop}\label{thm_SVT of non-singular A}
Let $P(x)\in \mathbb{C}[x]$ be a degree-$d$ polynomial such that $|P(x)| \leq 1/2$ for all $x\in[-1,1]$ and has parity-$(d\mod 2)$. Let $m,c,s,\alpha$ be described in \eqref{eq: some parameters}.
    Let $H=\sum_{k=1}^L \lambda_k H_k$ be an LCU decomposition of matrix $H\in \mb{C}^{\ell \times \ell}$ with $\sum_{k=1}^L |\lambda_k|= 1$. 
    Then, there exists a random quantum circuit $U_\Phi^{(R)}$, using only one ancilla qubit and $m$ applications of 
    $$
    \left\{U_k=\begin{bmatrix}
        cI & s H_k \\ s H_k^\dagger & -cI
    \end{bmatrix}: \quad k\in [L] \right\},
    $$ whose expectation is the circuit $U_\Phi$ in Theorem \ref{prop_non-unitary SVT circuit of A}.
\end{prop}
\begin{proof}
    This follows directly from Theorem \ref{prop_non-unitary SVT circuit of A} and Lemma \ref{lem_expectation of Pk}.
\end{proof}

With these ingredients, we now state our first randomized QSVT algorithm. To implement a target polynomial $P(H)$, we construct a randomized alternating-phase circuit $U^{(R)}_{\Phi}$ by drawing $m=\widetilde{O}(d)$ i.i.d. samples $H_{j_1}, \ldots, H_{j_m}$ from $\mc D$, and substituting the corresponding unitaries $U_{j_1}, \ldots, U_{j_m}$. On input $\ket{\psi_0}$, we measure $O$ at the end of each coherent run. Since $\E[U^{(R)}_{\Phi}]=U_{\Phi}$, the expected measurement outcome is 
$$
\braket{\psi_0|P(H')^{\dag}OP(H')|\psi_0},
$$ 
where $H'=H/(\lambda\alpha)$. 
As each outcome is bounded by $\|O\|$, Hoeffding's inequality implies that $O(1/\eps^2)$ independent repetitions suffice to obtain an $\eps\|O\|$-additive estimate.
  
\begin{breakablealgorithm}
\caption{Direct randomization of quantum singular value transformation [Theorem \ref{thm: randomised QSVT}]}
\begin{algorithmic}[1]
\REQUIRE An LCU decomposition of $H=\sum_{k=1}^L \lambda_k H_k$ with $\lambda=\sum_{k=1}^L |\lambda_k|$.

\qquad A quantum state $\ket{\psi_0}$ and an observable $O$.

\qquad A polynomial $P(x)$ satisfying conditions stated in Theorem \ref{thm: randomised QSVT}. 

\qquad Let $\alpha=\sqrt{m-1} \log (d^2/\varepsilon)$, $H'=\frac{H}{\lambda\alpha}$, and set $
T = ({2}/{\varepsilon^2}) \log(2/\delta).
$

\ENSURE An estimate $\mu$ such that
$$
\left |\mu-\braket{\psi_0|P(H')^\dag O P(H')|\psi_0}\right | \leq \varepsilon \|O\|
$$

\STATE Generate the distribution $\mathcal{D}=\{(|\lambda_k|/\lambda, H_k): k\in [L]\}$.

\STATE Construct polynomial $\tilde{P}(z)$ via Proposition \ref{prop_implement f(sx)} and generate phase $\Phi\in \mathbb{R}^m$ 
via Lemma \ref{lem_complex quantum signal processing} (see Eq.~\eqref{eq_matrix3}).

\STATE For $j=1,2,\ldots,T$ 

Generate a random quantum circuit $U_\Phi^{(R)}$ according to $\mathcal{D}$.
\hfill [Proposition \ref{thm_SVT of non-singular A}]

\begin{itemize}
    \item If $P(x)$ is even, measure $\ket{1}\bra{1}\otimes O$ in the state $U^{(R)}_\Phi \ket{1,\psi_0}$ and denote the measurement outcome as $\mu_j$.
    \item If $P(x)$ is odd, measure $\ket{0}\bra{0}\otimes O$ in the state $U^{(R)}_\Phi \ket{1,\psi_0}$ and denote the measurement outcome as $\mu_j$.
\end{itemize}

\STATE Output $\mu = \frac{1}{T} \sum_{j=1}^T \mu_j$. 
\end{algorithmic}
\label{alg: randomised QSVT}
\end{breakablealgorithm}

The full procedure is given in Algorithm~\ref{alg: randomised QSVT} and its correctness and complexity are analyzed in the following theorem.

\begin{thm}
\label{thm: randomised QSVT}
Let $P(x)\in \mathbb{C}[x]$ be a degree-$d$ polynomial 
such that $|P(x)| \leq 1/2$ for all $x\in[-1,1]$ and has parity-$(d\mod 2)$. Let $m,c,s,\alpha$ be as defined in Eq.~\eqref{eq: some parameters}.
    Let $H=\sum_{k=1}^L \lambda_k H_k$ be the LCU decomposition of matrix $H\in \mb{C}^{\ell \times \ell}$ with $\lambda=\sum_k |\lambda_k|$. 
    Define
    $$
    H'={H}/{\alpha\lambda}, 
    $$
    let $\ket{\psi_0}$  a quantum state and $O$ be an observable. Then Algorithm \ref{alg: randomised QSVT}, using 
    $$
    T=O(\varepsilon^{-2} \log(1/\delta)),
    $$ 
    repetitions of the quantum circuit in Proposition \ref{thm_SVT of non-singular A}, outputs an estimate $\mu$ such that
    \[
    \left| \mu- 
     \bra{\psi_0}  P(H')^\dag O P(H')
     \ket{\psi_0} 
     \right| \leq \varepsilon \|O\|
    \]
    with probability at least $1-\delta$. Moreover, each coherent run requires circuit depth $\widetilde{O}(d)$ and uses only a single ancilla qubit.
\end{thm}


\begin{proof}
We consider the case when $d$ is even. The odd case can be proved similarly.  
Note that
\[
\bra{\psi_0}  P(H')^\dag O 
     P(H')\ket{\psi_0} 
= \Tr\left[O ~ P(H')\ket{\psi_0} \bra{\psi_0}  P(H')^\dag\right],
\]
which can be viewed as the expectation value of $O$ with respect to the unnormalized state $P(H')\ket{\psi_0}$. By Theorem \ref{prop_non-unitary SVT circuit of A} and Proposition \ref{thm_SVT of non-singular A}, we have 
\[
\mb{E} \left[ U^{(R)}_\Phi \ket{1,\psi_0}\right] = 
U_\Phi \ket{1,\psi_0} \approx_\varepsilon
\ket{0} \otimes \ket{G} +
\ket{1} \otimes P(H')\ket{\psi_0},
\]
for some garbage state $\ket{G}$. 
To estimate the desired quantity, randomly generate $T$ quantum circuits $U^{(R_1)}_\Phi,\ldots,U^{(R_{T})}_\Phi$ and measure $\ket{1}\bra{1}\otimes O$ in the state $U^{(R_j)}_\Phi \ket{1,\psi_0}$ for all $j\in[T]$. 
Denote the measurement outcome as $\mu_j$, and let $\mu:= \sum_{j=1}^T \mu_j /T$. Since $\mu_j \in [-\|O\|,\|O\|]$, Hoeffding's inequality yields
    \[
    \Pr\left[ \,\left|\mu- \bra{\psi_0}  P(H')^\dag O P(H')\ket{\psi_0} \right| \geq \varepsilon \|O\| \, \right] 
    \leq 2 \exp(-T\varepsilon^2/2).
    \]
    To ensure failure probability at most $\delta$, it suffices to choose
    \[
    T \geq \frac{2}{\varepsilon^2} \log \left(\frac{2}{\delta}\right).
    \]
    This method uses only a single ancilla qubit, namely the one used in the construction of $U_\Phi$.
\end{proof}

Note that in Algorithm~\ref{alg: randomised QSVT}, $H$ does not need to be Hermitian.
In many cases, we are ultimately interested in estimating expectation values of the form
$$
\braket{\psi_0|f(H)^{\dag} O f(H)|\psi_0},
$$
where $f(x)$ is a target function that can be approximated by a bounded polynomial of degree $d$ and definite parity. 
We now show that each coherent run of our algorithm has circuit depth $\widetilde{O}(\lambda^2d^2)$.

Let $f_d(x)$ be a target function we want to approximate, where $d$ denotes its approximate polynomial degree, i.e., the minimum degree of the polynomial approximating $f_d(x)$. 
By Theorem~\ref{thm: randomised QSVT}, our randomized construction does not directly implement $f_d(H)$, but rather a rescaled version $f_{d'}(H')$ for some $d', H'$. 
Thus, to recover $f_d(H)$, one must enlarge the degree parameter so that $f_{d'}(H')=f_d(H)$. 
For many functions of interest, this requirement amounts to choosing
$$
\dfrac{d'}{\alpha\lambda}=d.
$$
Since the scaling factor satisfies $\alpha=\widetilde{O}(\sqrt{d'})$, it follows that $d'=\widetilde{O}(\lambda^2d^2)$. 
Hence, Algorithm \ref{alg: randomised QSVT} estimates the desired expectation value using $O(1/\eps^2)$ independent repetitions of a circuit of depth
$$
\widetilde{O}(\lambda^2d^2).
$$
This quadratic dependence holds quite generally, as most target functions of interest satisfy this property. 
As an illustration, consider Hamiltonian simulation $e^{\i Ht}$. 
It is known that $\sin(xt)$ and $\cos(xt)$ admit polynomial approximation of degree $d=\widetilde{\Theta}(t)$ \cite{low2019hamiltonian, gilyen2019quantum}. 
In our setting, the rescaling factor is $\alpha=O(\sqrt{t}\log(t/\eps))$, so the algorithm only approximately implements $e^{\i Ht/(\alpha\lambda)}$. 
To recover the true evolution, we must instead simulate with an effective parameter $t'$ satisfying 
$$
\dfrac{\sqrt{t'}}{\lambda\log(t'/\eps)}=O(t),
$$
which implies $t'=\widetilde{O}(\lambda^2 t^2)$. 
Thus, Theorem \ref{thm: randomised QSVT} yields a randomized quantum algorithm for Hamiltonian simulation: if $\ket{\psi_t}=e^{-\i Ht}\ket{\psi_0}$, then Algorithm \ref{alg: randomised QSVT} outputs $\mu$ such that
    $$
    \left|\mu-\braket{\psi_t|O|\psi_t}\right|\leq \eps\|O\|,
    $$
using $O(1/\eps^2)$ repetitions of a quantum circuit of depth $\widetilde{O}(\lambda^2t^2)$. The same quadratic overhead extends to other important tasks. 
For instance, in quantum linear systems, the target function  $f(x)=1/\kappa x$ can be approximated by a polynomial of degree $d=\widetilde{\Theta}(\kappa)$ \cite{gilyen2019quantum}. 
For matrix powers $f(x)=x^{t}$, the approximate degree is $d=\widetilde{\Theta}(\sqrt{t})$ \cite{sachdeva2014faster}. In all such cases, setting $d'=\widetilde{O}(\lambda^2 d^2)$ suffices to implement the desired transformation.

In contrast, standard QSVT requires circuit depth $\widetilde{O}(L\lambda d)$ for the same task, together with $O(\log L)$ ancilla qubits. 
Our randomized approach, by comparison, requires only a single ancilla qubit and eliminates the linear dependence on $L$. 
Consequently, it is not only more resource-efficient, but also provably faster whenever the polynomial degree is relatively small, i.e., when $\lambda d \ll L$. 
This regime naturally arises in many near-term applications, making our method particularly well-suited for early fault-tolerant quantum devices.

Although Algorithm~\ref{alg: randomised QSVT} can be viewed as a direct randomization of standard QSVT, it also inherits some of its drawbacks. For example, as in the deterministic setting, implementing a real polynomial requires two additional ancilla qubits. 
More importantly, in many applications the ultimate goal is not just to apply $f(H)$, but to estimate the expectation value of observables with respect to the normalized state
$$
\ket{\psi}=\dfrac{f(H)\ket{\psi_0}}{\|f(H)\ket{\psi_0}\|}, 
$$
that is, the quantity $\braket{\psi|O|\psi}$. 
If the normalization factor is bounded as $\|f(H)\ket{\psi_0}\|\geq \eta$, standard QSVT prepares $\ket{\psi}$ using $\widetilde{O}(1/\eta)$ rounds of fixed-point amplitude amplification \cite{yoder2014fixed}, followed by $O(1/\eps^2)$ repetitions of the resulting circuit to estimate the observable.   

In contrast, Algorithm~\ref{alg: randomised QSVT} cannot benefit from amplitude amplification. 
The only option is to estimate $\braket{\psi|O|\psi}$ through repeated classical sampling. 
If we let $\ell=\|f(H)\ket{\psi_0}\|$, the algorithm effectively outputs $\mu/\ell^2$. 
However, the measurement outcomes are now random variables bounded in  
$$\left[-\|O\|/\eta^2, \|O\|/\eta^2\right],$$
which significantly increases the variance. 
As a result, although each run still has circuit depth $\widetilde{O}(\lambda^2 d^2)$, the number of classical repetitions required increases to
$$
T=\widetilde{O}\left(\dfrac{1}{\eps^2\eta^4}\right).
$$
Hence, the overall complexity becomes
$$
\widetilde{O}\left(\dfrac{\lambda^2d^2}{\eps^2\eta^4}\right),
$$
which is highly suboptimal compared to what can be achieved by amplitude amplification or estimation in standard QSVT. 
Notably, this matches the complexity reported in \cite{chakraborty2025quantum, wang2024qubit}. 
Thus, although our randomized construction avoids block encodings and reduces hardware overhead, this inefficiency is the main bottleneck of the first approach, and overcoming it is the key motivation for our second randomized algorithm. 

For our second algorithm, we take inspiration from \cite{chakraborty2025quantum} and integrate qDRIFT into the generalized QSP framework, considering linear combinations of Hermitian operators. 
This yields two key advantages. First, unlike the direct randomization approach, the number of ancilla qubits is fixed at one, regardless of whether the target polynomial is real or complex \cite{motlagh2024generalized}. 
Second, combining qDRIFT with classical extrapolation enables circuits of shorter depth. 
More importantly, this approach naturally supports amplitude amplification, allowing us to compute $\braket{\psi|O|\psi}$ with the same number of classical repetitions as standard QSVT, thereby eliminating the main inefficiency of Algorithm~\ref{alg: randomised QSVT}. 
Consequently, our second algorithm provides a more powerful and resource-efficient framework for randomized QSVT.

\subsection{Randomized QSVT with qDRIFT}
\label{sec:qsvt-with-qdrift}

Our second randomized algorithm embeds qDRIFT within the framework of generalized quantum signal processing (GQSP) \cite{motlagh2024generalized, wang2023quantum}. For any Hamiltonian $H$, GQSP (see Sec.~\ref{subsec:prelim-gqsp}) implements a Laurent polynomial $P(e^{\i H})$ (bounded on the complex unit circle $\mathbb{T}$) using circuits that alternate between parametrized $U(2)$ rotations and controlled $e^{\i H}$ (or $e^{-\i H}$). 
Unlike standard QSVT, this framework natively handles both real and complex polynomials without additional ancilla qubits. Moreover, since Laurent polynomials can approximate (i) most functions of interest and (ii) nearly any polynomial achievable by QSVT, GQSP has emerged as a versatile framework for polynomial transformations \cite{chakraborty2025quantum}. Although the procedure requires $\widetilde{O}(d)$ queries to $e^{\i H}$ (or $e^{-\i H}$) and only one ancilla qubit, its end-to-end efficiency depends critically on how $e^{\i H}$ is realized. Prior work \cite{motlagh2024generalized} assumed block-encoding access to $H$, enabling the use of optimal Hamiltonian simulation methods \cite{low2019hamiltonian}, but at the cost of substantial ancilla overhead, intricate controlled operations, and circuit depth scaling with $L$.

In our randomized setting, we replace exact Hamiltonian evolutions with qDRIFT \cite{campbell2019random}, thereby eliminating both the ancilla overhead and the explicit dependence on $L$. 
A direct substitution, however, would inherit qDRIFT’s $1/\eps$ depth scaling (see Sec.~\ref{subsec:prelim-qdrift}), which substantially increases the overall cost. To address this, we introduce a novel use of Richardson extrapolation that reduces the scaling to $\mathrm{polylog}(1/\eps)$, dramatically improving efficiency while being resource-efficient.

Building on the framework of \cite{chakraborty2025quantum}, we consider a general class of interleaved quantum circuits composed of Hamiltonian evolutions (possibly of different Hamiltonians) interspersed with arbitrary unitaries. Formally, let
$$
H^{(1)}, H^{(2)}, \dots, H^{(M)}$$ 
be a family of $M$ Hermitian operators of the same dimension. Each $H^{(\ell)}$ admits a linear combination of Hermitian
decomposition 
$$
H^{(\ell)} = \sum_{k=1}^{L_\ell} \lambda_k^{(\ell)} H_k^{(\ell)}, \quad \|H_k^{(\ell)}\|=1, \quad \lambda^{(\ell)}=\sum_{k=1}^{L_\ell} |\lambda_k^{(\ell)}|,
$$
and let $$\lambda := \max_\ell \{\lambda^{(\ell)}\}.$$
Then, the interleaved sequence circuit $W$ is written as follows
\be
\label{eq:inteleaved-circuit-definition}
W = V_0\cdot e^{\i H^{(1)}}\cdot V_1 \cdot e^{\i H^{(2)}}\cdots e^{\i H^{(M)}}\cdot V_M = V_0 \prod_{\ell=1}^M e^{\i H^{(\ell)}} V_{\ell},    
\ee
where $\{V_\ell\}_{\ell=0}^{M}$ are arbitrary unitary operators. 
Observe that $W$ subsumes GQSP: set $V_j=R_j(\theta_j,\phi_j,\gamma_j)\otimes I$, and $H^{(j)}=\widetilde{H}$, where $\widetilde{H}=\mathrm{diag}(H,0)$ or $\mathrm{diag}(0,-H)$ depending on whether control qubit is $0$ or $1$. For an $n$-qubit Hamiltonian, the circuit $W$ uses $n+1$ qubits overall. 

We establish a general result: for any initial state $\ket{\psi_0}$ and observable $O$, one can estimate
$$
\braket{\psi_0|W^{\dag}OW|\psi_0},
$$
with an additive error at most $\eps\|O\|$, using no block encodings, only a single ancilla qubit, and a circuit depth that scales polylogarithmically in $1/\eps$. Specifically, our method requires $\widetilde{O}(\eps^{-2})$ repetitions of a quantum circuit of depth $\widetilde{O}(\lambda^2 M^2+M)$.

The key idea is to replace the exact Hamiltonian evolutions in $W$ with qDRIFT-based simulations. Our main technical contribution is to show that the time-dependent error in observable expectation values between this randomized circuit and the exact $W$ can be expressed as a polynomial in the parameters of the qDRIFT channel. This insight enables us to apply Richardson extrapolation \cite{low2019well} in a novel way to mitigate the error, thereby improving the circuit depth scaling.

Importantly, this analysis does not follow directly from prior extrapolation-based error reduction methods for randomized Hamiltonian simulation \cite{watson2024randomly}, nor from \cite{chakraborty2025quantum}, where Trotterization was combined with extrapolation. Instead, we introduce two time-dependent superoperators, one for the exact $W$ and one for its qDRIFT implementation, and analyze their difference in the interaction picture. This formulation is crucial for deriving a tight upper bound on the error in expectation values. For clarity of exposition, we present the arguments below in terms of density operators, even though our main theorems are stated for pure states.

The error series between the density operator evolved under the parametrized qDRIFT channel $\mathcal{E}_s^{1 / s}(\rho)$ and the exact evolution $e^{\i \ad_H}( \rho)$ was established in \cite{watson2024randomly} (restated in Lemma \ref{lem:watson qdrift}). Our subsequent results can be seen as a nontrivial generalization of this difference to arbitrary interleaved sequence $W$. For any density operator $\rho$, define 
\be 
\label{eq:exact qdrift state}
 \quad \rho_{2M+1} =  \mc  N_{2M+1}(\rho):=W\rho W^{\dag}
\ee
to be the output state of the exact interleaved circuit $W$. On the other hand, for a family of Hamiltonians 
$\{H^{(\ell)}\}_{\ell=1}^{M}$, define the corresponding parametrized qDRIFT channel as 
\[
\mathcal{E}^{(\ell)}_s(\rho):= \sum_k p_k^{(\ell)} e^{\i s \lambda^{(\ell)} \mathrm{ad}_{H_k^{(\ell)}}}(\rho),
\]
where $p_k^{(\ell)}=|\lambda_k^{(\ell)}|/\lambda^{(\ell)}$.
The output quantum state from the approximate interleaved circuit (with arbitrary unitaries and qDRIFT) is given by
\begin{equation}
\label{eq:appro qdrift state}
\tilde{\rho}_{s,2M+1} :=  \Big(\mc V_0\prod_{\ell=1}^M \big(\mathcal{E}^{(\ell)}_{sM}\big)^{1/(sM)}\mc V_\ell\Big)(\rho),
\end{equation}
where $\mc V_\ell$ is the unitary channel defined via $\mc V_\ell(\rho) := V_\ell \rho V_\ell^\dagger$. Then we have the following result on the difference between density operators output by the two circuits.

\begin{lem}
\label{lem_QSP+qDRIFT}
Let $K\in \mb Z_{>3}$ and assume $s \in(0,1/2\lambda M)$, then we have
\begin{equation}
        \tilde{\rho}_{s,2M+1}-\rho_{2M+1}=\sum_{j \geq 1} s^j \tilde{\mc D}_j(\rho) ,
\end{equation}
where the diamond norm of the superoperators $\tilde{\mc D}_j$ satisfies
\[
    \begin{aligned}
    \big\|\tilde{\mc D}_j\big\|_{\diamond} & \leq(8 \lambda M)^j \sum_{\ell=1}^{j} \frac{(8 \lambda M)^\ell}{\ell!}.
    \end{aligned}
\]
\end{lem}

\begin{proof}
Let $\delta=sM < 1/2\lambda$, then as stated in Lemma \ref{lem:para qdrift channel}, each channel $\mathcal{E}^{(\ell)}_\delta$ has the form
\[
\mathcal{E}^{(\ell)}_\delta=e^{\i s \mc G^{(\ell)}(\delta)}, \quad \text{ where } \; \mc G^{(\ell)}(\delta)=\ad_{H^{(\ell)}}+\sum_{j=1}^{\infty} \delta^j \mc D_{j+1}^{(\ell)},
\]
with the superoperators bounded as $\|\mc D_j^{(\ell)}\|_{\diamond} \leq(4 \lambda)^j$.

Define two time-dependent superoperators $\mc K(t)$ and $\tilde{\mc K}(t)$ as follows:
\begin{equation*}        
\mc K(t)=\begin{cases} \mc \ad_{\log( V_{\lfloor t\rfloor/2})} & \text{ if } \lfloor t\rfloor \text{ is even}, \\ \i \ad_{H^{(\lceil t/2\rceil)}}  & \text{  if } \lfloor t\rfloor \text{ is odd}, \end{cases}  
\quad 
\tilde{\mc K}(t) =\begin{cases} \ad_{\log( V_{\lfloor t\rfloor/2})} & \text{ if } \lfloor t\rfloor \text{ is even}, \\ \i \mc G^{(\lceil t/2 \rceil)}(\delta) & \text{  if } \lfloor t\rfloor \text{ is odd}, \end{cases} 
\end{equation*}
for $t\in[0,2M+1]$. Let $\mathds{1}_A$ be the indicator function of $A=\{t\in [0,2M+1]: \lfloor t\rfloor \text{ is odd}\}$, then we have \begin{align*}
    \tilde{\mc K}(t)-\mc K(t) = \i \mathds{1}_{A}(t) \sum_{j=1}^{\infty} \delta^j \mc D_{j+1}^{(\lceil t/2\rceil)}.
\end{align*} 
Define the corresponding evolution operators $\mc N_t, \tilde{\mc N}_{s,t}$ as \begin{align*}
    \partial_t \mc N_t = \mc K(t)\mc N_t, \quad \partial_t \tilde{\mc N}_{s,t}=\tilde{\mc K}(t) \tilde{\mc N}_{s,t},
\end{align*}
which is consistent with the previous definition of ${\mc N}_{2M+1}$. 
For density operator $\rho$, define $\rho_t, \tilde{\rho}_{s, t}$ as \[
\rho_t = \mc N_t(\rho)=\exp_{\mc{T}}\left( \int_{0}^{t} \mc K(\tau)\,\d \tau \right)(\rho), \quad \tilde{\rho}_{s,t}= \tilde{\mc N}_{s,t}(\rho)=\exp_{\mc{T}}\left( \int_{0}^{t} \tilde{\mc K}(\tau)\,\d \tau \right)(\rho).\]
Then we have
\[
\rho_{2M+1}=\Big(\mc V_0\prod_{\ell=1}^M e^{\i \ad_{H^{(\ell)}}} \mc V_\ell\Big)(\rho),     \qquad \tilde{\rho}_{s,2M+1} =  \Big(\mc V_0\prod_{\ell=1}^M \big(\mathcal{E}^{(\ell)}_{sM}\big)^{1/(sM)}\mc V_\ell\Big)(\rho),
\]
as in Eqs.~\eqref{eq:exact qdrift state} and~\eqref{eq:appro qdrift state}. We can use Lemma \ref{lem:interaction-pic} to provide a series expansion in $\delta$ of the difference between $\tilde{\rho}_{s,t}$ and $\rho_t$. Indeed, we have 
{\small\begin{align*}
    \tilde{\rho}_{s,2M+1} & = \mc N_{2M+1}\exp_{\mc T}\left(\int_{0}^{2M+1}\mc N_{\tau}^{-1} (\tilde{\mc K}(\tau)-\mc K(\tau))\mc N_{\tau}\, \d \tau\right)(\rho) \\
    & = \mc N_{2M+1}\bigg(\sum_{j = 0}^{\infty} \int_{0}^{2M+1} \d \tau_1\int_{0}^{\tau_1}\d \tau_2 \cdots\int_{0}^{\tau_{j-1}}\d \tau_j \prod_{\ell=j}^{1} \Big(\mc N_{\tau_\ell}^{-1} (\tilde{\mc K}(\tau_\ell)-\mc K(\tau_\ell))\mc N_{\tau_\ell}\Big)\bigg)(\rho)\\
    & =  \rho_{2M+1}+\bigg(\sum_{j = 1}^{\infty} \int_{0}^{2M+1} \d \tau_1\int_{0}^{\tau_1}\d \tau_2 \cdots\int_{0}^{\tau_{j-1}}\d \tau_j \prod_{\ell=j}^{1} \Big(\sum_{k=1}^{\infty}\i\mathds{1}_A(\tau_\ell)\delta^{k}\mc N_{\tau_\ell}^{-1}\mc D_{k+1}^{(\lceil \tau_{\ell}/2\rceil)}\mc N_{\tau_\ell}\Big)\bigg)(\rho) \\
    & = \rho_{2M+1}+ \bigg(\sum_{k = 1}^{\infty}  \delta^{k} \sum_{j=1}^{\infty} \sum_{\substack{k_1,\ldots, k_j\ge 1\\ k_1+\cdots k_j=k}}\int_{0}^{2M+1} \d \tau_1\int_{0}^{\tau_1}\d \tau_2 \cdots\int_{0}^{\tau_{j-1}}\d \tau_j \prod_{\ell=j}^1 \Big(\i \mathds{1}_A(\tau_\ell)\mc N_{\tau_\ell}^{-1} \mc D_{k_\ell+1}^{(\lceil\tau_\ell/2\rceil)}\mc N_{\tau_l}\Big)\bigg)(\rho)\\
    & =: \rho_{2M+1}+\sum_{j=1}^{\infty}\tilde{\mc D}_{k}(\rho).
\end{align*} }The diamond norm of $\tilde{\mc D}_{k}$ can be bounded by
{\small\begin{align*}
    \|\tilde{\mc D}_{k}\|_{\diamond} & \le \sum_{j=1}^{\infty} \sum_{\substack{k_1,\ldots, k_j\ge 1\\ k_1+\cdots k_j=k}}\int_{0}^{2M+1} \d \tau_1\int_{0}^{\tau_1}\d \tau_2 \cdots\int_{0}^{\tau_{j-1}}\d \tau_j \prod_{\ell=j}^1 \Big( \mathds{1}_A(\tau_\ell)\mc \|\mc N_{\tau_\ell}^{-1}\|_{\diamond} \|\mc D_{k_\ell+1}^{(\lceil\tau_\ell/2\rceil)}\|_{\diamond}\mc \|\mc N_{\tau_l}\|_{\diamond}\Big) \\
    &\le \sum_{j=1}^{\infty} \sum_{\substack{k_1,\ldots, k_j\ge 1\\ k_1+\cdots k_j=k}}\int_{0}^{2M+1} \d \tau_1\int_{0}^{\tau_1}\d \tau_2 \cdots\int_{0}^{\tau_{j-1}}\d \tau_j \prod_{\ell=j}^1 \Big( \mathds{1}_A(\tau_\ell)(4\lambda)^{k_{\ell}+1}\Big)
    \\
    &=\sum_{j=1}^{k} (4\lambda)^{k}\frac{(4\lambda M)^{j}}{j!} \sum_{\substack{k_1,\ldots, k_j\ge 1\\ k_1+\cdots +k_j=k}}1\\
    &\le  (8\lambda)^{k}\sum_{j=1}^{k}\frac{(8\lambda M)^{j}}{j!},
\end{align*}}where the second line follows from that $\mc N_t$ and $\mc N_t^{-1}$ are unitary channels preserving diamond norm, the third line follows from  \[
\int_{0}^{2M+1} \d \tau_1\int_{0}^{\tau_1}\d \tau_2 \cdots\int_{0}^{\tau_{j-1}}\d \tau_j \prod_{\ell=j}^1  \mathds{1}_A(\tau_\ell)= \frac{\mathrm{vol}(A^{\otimes \ell})}{\ell !}=\frac{M^\ell}{\ell!},
\]
and the fourth line follows from \[
\sum_{\substack{k_1,\ldots, k_j\ge 1\\ k_1+\cdots +k_j=k}}1=\binom{k+j-1}{j-1}\le 2^{k+j}.
\]
\end{proof}

We can apply the bound obtained in Lemma \ref{lem_QSP+qDRIFT} to choose an appropriate parameter $s$ used in the Richardson extrapolation procedure to obtain the desired accuracy. Formally, we have

\begin{lem}
\label{lem_qsp, qDRIFT, richardson}
    Suppose that $8\lambda M \ge 1$. For any $\eps\in (0,1)$, set $m = \lceil \log(1/\eps)\rceil$ and  \[
    s=(8 \lambda M)^{-2} ({4C\log m}/{\varepsilon})^{-1 / m},\]
    with the constant $C$ defined in Lemma~\ref{lem_richardson extrapolation}.
    Let $O$ be an observable, define the $m$-th order Richardson estimator as
    \begin{equation*}
        \langle \tilde{O}_{m,s}\rangle:= \sum_{i=1}^m b_i \Tr\left[ O  \tilde{\rho}_{s_i, 2M+1}\right],
    \end{equation*}
    where $b_i, r_i$, and $s_i=1/(Mr_i)$ are given as
    in Lemma \ref{lem_richardson extrapolation} with $t =1/M$. 
    Then, the extrapolation error is bounded by
    \[
    \left|\langle \tilde{O}_{m,s}\rangle-\Tr\left[O \rho_{2M+1}\right]\right| \leq \varepsilon \|O\|/2,
    \]
\end{lem}

\begin{proof}
    Define the function $f(x)$ as
    \[\begin{aligned}
        f(x)&:= \Tr\left[O \tilde{\rho}_{x,2M+1}\right] \quad \text{ for }x\in(0,1) ,\\
        f(0)&:= \Tr\left[O \rho_{2M+1}\right].
    \end{aligned}
    \]
    Then for $x\in(0,1/2\lambda M)$, by Lemma \ref{lem_QSP+qDRIFT}, we have 
    \begin{equation*}
        f(x)=f(0)+\sum_{j \geq 1} x^j \Tr\left[O \tilde{\mc D}_{j}(\rho)\right].
    \end{equation*}
    The $m$-th order Richardson extrapolation procedure removes all terms up to $O(s^{m-1})$ in the series as shown in Lemma \ref{lem_richardson extrapolation}, which yields
    \begin{equation*}
        \bigg|\sum_{i=1}^m b_i \Tr\left[ O  \tilde{\rho}_{s_i, 2M+1}\right]- \Tr\left[O \rho_{2M+1}\right]\bigg|\leq \|\b\|_1\max_{i} \left|R_m(s_i) \right|.
\end{equation*} By Lemma~\ref{lem_richardson extrapolation}, we have $s_i\le s$ for any $i\in[m]$.
   Thus, by Lemma~\ref{lem_QSP+qDRIFT},
    \begin{align*}
         \max_{i} \left|R_m(s_i)\right| 
        &\leq \|O\|  \sum_{j \geq m} s^j (8 \lambda M)^j \sum_{\ell=1}^{j} \frac{(8 \lambda M)^\ell}{\ell!} \\
        &\leq \|O\|  \sum_{j \geq m} s^j (8 \lambda M)^{2j} \sum_{\ell=1}^{j} \frac{1}{\ell!} \\
        & = \|O\|(s (8\lambda M)^2)^{m}\frac{1}{1-s(8\lambda M)^2 }(e-1) \\
        & \le 2\|O\| \frac{\eps}{8\|\b\|_1} \frac{1}{1-(\eps/8\|\b\|_1)^{1/m}} \\
        & \le\frac{\eps \|O\| }{4\|\b\|_1}\frac{1}{1-(\eps/8)^{1/\log(1/\eps)}} \le  \frac{\eps\|O\|}{2\|\b\|_1} 
    \end{align*}
   Combining all the above, we obtain
    \[
    \left|\langle \tilde{O}_{m,s}\rangle-\Tr\left[O \rho_{2M+1}\right] \right| \leq \varepsilon \|O\|/2.
    \]
    This completes the proof.
\end{proof}

Combining Lemma \ref{lem_qsp, qDRIFT, richardson} with the parameter choice specified in Lemma \ref{lem_richardson extrapolation}, we can estimate the complexity for computing $\langle \tilde{O}_{m,s}\rangle$. 
We summarize this as the main theorem of this section. In the following, we assume that each $V_{\ell}$  can be implemented in constant time. 

\begin{thm}[Interleaved Hamiltonian evolution with qDRIFT]
\label{thm_HSVT with qDRIFT}
For $\ell\in[M]$, let 
$$
H^{(\ell)}=\sum_{k=1}^{L_\ell} \lambda_k^{(\ell)} H_{k}^{(\ell)}
$$ 
be Hermitian decompositions of Hermitian operators $H^{(\ell)}$ of the same dimension, where $\lambda^{(\ell)}=\sum_k |\lambda_k^{(\ell)}|$ and $\|H_k^{(\ell)}\|=1$ for all $\ell,k$. Let $\lambda=\max_\ell\{\lambda^{(\ell)}\}$, and $V_0,\ldots,V_M$ be unitaries. Assume that $\lambda M \ge 1$. Define
\[
W := V_0 \prod_{\ell=1}^M e^{\i H^{(\ell)}} V_{\ell}.    
\] 
Suppose $\varepsilon\in(0,1)$ is the precision parameter and each evolution $e^{\i H^{(\ell)}_k\tau}$ can be implemented in $\widetilde{O}(1)$ time for any $\tau\in \mathbb{R}$. Then there is an algorithm (see Algorithm \ref{alg: HSVT with qDRIFT}) that, for any observable $O$ and initial state $\ket{\psi_0}$, estimates 
$$ \bra{\psi_0} W^{\dagger} O W  \ket{\psi_0}, $$ 
with an additive error at most $\varepsilon\|O\|$, with time complexity $$
\widetilde{O}\left((\lambda^2 M^2+M)/\varepsilon^2\right)
$$ and maximum quantum circuit depth $$\widetilde{O}\left(\lambda^2 M^2+M\right).$$
\end{thm}

\begin{proof}
Let $\langle \tilde{O}_{m,s}\rangle$ be the Richardson estimator defined in Lemma~\ref{lem_qsp, qDRIFT, richardson} with $\rho=\ket{\psi_0} \bra{\psi_0}$ and $t = 1/M$. As $8\lambda M \ge 1$, by Lemma~\ref{lem_qsp, qDRIFT, richardson}, we have \[
    \bigl| \langle \tilde{O}_{m,s}\rangle-\Tr[O (W\ket{\psi_0}\bra{\psi_0}W^{\dagger})] \bigr|  = \bigl| \langle \tilde{O}_{m,s}\rangle-\bra{\psi_0} W^{\dagger} O W  \ket{\psi_0}  \bigr| \leq \varepsilon \|O\|/2.
    \]
The Richardson estimator is a linear combination of several terms, $\langle \tilde{O}_{m,s}\rangle = \sum_i b_i \Tr[O\tilde{\rho}_{s_i, 2M+1}]$, where the coefficient vector $\b$ satisfies $\|\b\|_1=O(\log m) = O(\log\log(1/\eps))$ by Lemma \ref{lem_richardson extrapolation}. To achieve a total additive error of $\varepsilon \|O\|/2$ for $\langle \tilde{O}_{m,s}\rangle$, we can estimate each individual term \begin{align*}
    \tr[O\tilde{\rho}_{s_i, 2M+1}] = \tr\bigg[O\Big(\mc V_0\prod_{\ell=1}^M \big(\mathcal{E}^{(\ell)}_{s_iM}\big)^{1/(s_iM)}\mc V_\ell\Big)(\rho) \bigg],
\end{align*}
with an additive error at most $\eps\|O\|/(2\|\b\|_1)=O(\eps\|O\|/\log\log(1/\eps))$. 
From Lemma \ref{lem_richardson extrapolation}, we have \begin{align*}
    \frac{1}{s_i M} = \frac{t}{s_i} = r_i\in \mathbb{Z}. 
\end{align*} Therefore, we can prepare $\tilde{\rho}_{s_i, 2M+1}$ by a quantum circuit of depth $O(M+M/(s_iM)) = O(M+Mr_i)$. 
As $\|\b\|_1=\Theta(\log (m))$, $(\|\b\|_1)^{1 / m}=\Theta(1)$. 
Then, we have \[
1/s =  O\left(\lambda^2 M^2\dfrac{(\|\b\|_1)^{1/m}}{\eps^{1/\log(1/\eps)}}\right) = {O}\left(\lambda^2 M^2\right).\]  Then we can estimate  $\tr[O\tilde{\rho}_{s_i, 2M+1}]$ with an additive error at most $\eps\|O\|/(2\|\b\|_1)$, with a success probability at least $1-O(1/m)$, by \begin{align*}
    O\left(\frac{\|\b\|_1^2}{\eps^2}\log m\right)=O\left(\frac{(\log\log(1/\eps))^3}{\eps^2}\right)
\end{align*} repeated running of the quantum circuit of depth $O(M+Mr_i)$ and then measure the observable $O$. By the union bound, the overall success probability is at least a constant.

From Lemma \ref{lem_richardson extrapolation}, we also have $r_1=O(\max\{m^3, m^2t/s\})$, and hence the maximum quantum circuit depth is 
\begin{align*}
    O(M+M\max_i \{r_i\}) 
    &= O(M+Mr_1)\\
    &= O\left(\max\{Mm^3, m^2/s\}\right)\\
    &=O\left(\max\{\lambda^2 M^2 \log^2(1/\varepsilon),M\log^3(1/\eps)\}\right).
\end{align*}
Constructing the Richardson estimator $\langle \tilde{O}_{m,s}\rangle$ requires estimating all $m$ sample points. Each point requires $O((\log\log(1/\eps))^3/\eps^2)$ runs of a circuit with depth up to $O(M+Mr_i)$. This implies that the total number of gates is 
\begin{align*}
    O\left(\frac{(\log\log(1/\eps))^3}{\eps^2}\sum_{i=1}^m(M+Mr_i)\right) 
    &=O\left(\frac{(\log\log(1/\eps))^3}{\eps^2}\sum_{i=1}^mM\max\{m^3, m^2t/s\}/i^2\right)\\
    &=O\left(\frac{(\log\log(1/\eps))^3}{\eps^2}\max\{\lambda^2 M^2\log^2(1/\eps),M \log^3(1/\eps)\}\right)\\
    & =\widetilde{O}\left(\dfrac{\lambda^2 M^2 +M}{\eps^2}\right),
\end{align*}
where the second line follows from $\sum_{i=1}^m i^{-2}\le \pi^2/6 = O(1)$.
\end{proof}

To accurately estimate the target expectation value, the algorithm repeatedly executes the circuit $W$, each time using a different qDRIFT step size $s_i$ chosen according to Lemma~\ref{lem_richardson extrapolation}. Each run yields an estimate $f(s_i)$, the expectation value of $O$ with respect to the corresponding output state. The final estimate is then obtained as the linear combination $\sum_i b_i f(s_i)$, where the coefficients $b_i$ are specified by Lemma~\ref{lem_richardson extrapolation}. The complete procedure in the proof of Theorem~\ref{thm_HSVT with qDRIFT} is summarized via Algorithm~\ref{alg: HSVT with qDRIFT}.

\begin{breakablealgorithm}
\caption{Interleaved Hamiltonian evolution with qDRIFT [Theorem \ref{thm_HSVT with qDRIFT}]}
\begin{algorithmic}[1]
\REQUIRE $M$ Hermitian operators $H^{(\ell)}=\sum_{k=1}^{L_\ell} \lambda_k^{(\ell)} H_{k}^{(\ell)}$ in Hermitian decomposition, 

\qquad with $\|H_k^{(\ell)}\|=1,~\lambda^{(\ell)}=\sum_k |\lambda_k^{(\ell)}|,~\text{and}~\lambda=\max_j\{\lambda^{(\ell)}\}$ for all $\ell\in[M]$. 

\qquad qDRIFT protocol that implements channel $\mathcal{E}^{(\ell)}_s:= \sum_k p_k^{(\ell)} e^{-\i s  \lambda^{(\ell)} \mathrm{ad}_{H_k^{(\ell)}}}$.
\hfill [See Algorithm \ref{alg: qDRIFT}]

\qquad A quantum state $\ket{\psi_0}$ and an observable $O$.

\qquad $M+1$ unitaries $V_0,\ldots,V_M$, with channel $\mc V_\ell$ defined via $\mc V_\ell(\rho) := V_\ell^\dagger \rho V_\ell$.

\qquad $W := V_0 \prod_{\ell=1}^M e^{\i H^{(\ell)}} V_{\ell}$.\\~\\

\ENSURE $ \bra{\psi_0} W^{\dagger} O W  \ket{\psi_0}  \pm \varepsilon\|O\|$.\\~\\

\STATE Choose $m = \lceil \log(1/\eps)\rceil$, $s=(8 \lambda M)^{-2} ({4C\log m}/{\varepsilon})^{-1 / m}$ with $C$ defined in Lemma~\ref{lem_richardson extrapolation}, and compute $b_1,\ldots,b_m$ and $s_1,\ldots,s_m$ via Lemma~\ref{lem_richardson extrapolation} with $t=1/M$.

\FOR{$i=1, \ldots, m$}
    \STATE Estimate $f(s_i) := \Tr\left[O \left(\mc V_0\prod_{\ell=1}^M (\mathcal{E}_{s_i M}^{(\ell)})^{1/(s_i M)}\mc V_\ell\right)\left(\ket{\psi_0} \bra{\psi_0} \right) \right]$ with an additive error at most $\frac{\|O\|\varepsilon}{2\|\b\|_1}$ and denote the result by $\hat{f}_i$.
\ENDFOR
\STATE \textbf{return} $\sum_{i=1}^m b_i \hat{f}_i$.

\end{algorithmic}
\label{alg: HSVT with qDRIFT}
\end{breakablealgorithm}

As noted earlier, the generalized QSP circuit in Eq.\eqref{eq:gqsp} is a special case of the interleaved structure $W$. For any Hamiltonian $H$, Theorem~\ref{thm:generalized-QSP} guarantees that this interleaved circuit of controlled Hamiltonian evolutions and single-qubit rotations implements a Laurent polynomial $P(e^{\i H})$, with $|P(x)|\leq 1$ on the unit circle $\mathbb{T}$. By Theorem~\ref{thm_HSVT with qDRIFT}, we therefore obtain a procedure which, given an initial state $\ket{\psi_0}$ and observable $O$, estimates
$$
\braket{\psi_0|P(e^{\i H})^{\dag} O P(e^{\i H})|\psi_0},
$$
to additive accuracy $\eps\|O\|$. In this setting, the circuit length satisfies $M=\widetilde{O}(d)$, where $d$ is the degree of the Laurent polynomial. Formally, we have the following result:

\begin{thm}[Generalized QSP with qDRIFT]
\label{thm_gqsp with qdrift}
    Let $H=\sum_{k=1}^L \lambda_k H_k$ be a Hermitian decomposition of $H$, with $\|H_k\|=1$, and $\lambda=\max\{\sum_{k}|\lambda_k|,1\}$. Also, let $\varepsilon\in(0,1/2)$ be the precision parameter. 
    Furthermore, suppose $P(z)=\sum_{j=-d}^d a_j z^j$ is a $d$-degree Laurent polynomial bounded by $1$ on $\mathbb{T}:=\{ x\in \mathbb{C} : |x|=1 \}$. 
    Then there exists a quantum algorithm that, for any observable $O$ and initial state $\ket{\psi_0}$, estimates 
    $$
    \bra{\psi_0} (P(e^{\i H})^\dagger O P(e^{\i H}) \ket{\psi_0},$$ 
    with an additive error at most $\varepsilon \|O\|$ and constant success probability, using only a single ancilla qubit. 
    The maximum quantum circuit depth is
    \be
    \widetilde{O}\left(\lambda^2 d^2\right),
    \ee
    and the total time complexity is
    \be
    \widetilde{O}\left(\lambda^2 d^2/\varepsilon^2\right).
    \ee
\end{thm}
\begin{proof}
    Combining Theorems \ref{thm:generalized-QSP} and \ref{thm_HSVT with qDRIFT}, set $U=e^{\i H}$ in the GQSP circuit of Eq.~\eqref{eq:gqsp}, and use the parametrized qDRIFT channel to approximate the Hamiltonian evolution. Then the stated results follow by observing $M=O(d)$ and \[
    \widetilde{O}\left(\lambda^2 M^2+M\right) = \widetilde{O}(\lambda^2M^2) = \widetilde{O}(\lambda^2d^2),
    \] 
    as $\lambda,M \ge 1$.
\end{proof}

Following Theorem~\ref{thm_gqsp with qdrift}, we now discuss its application to more general functions and Hamiltonians. While the theorem is stated for a Laurent polynomial $P(e^{\i H})$, which is periodic in the spectrum of $H$, our method can be applied to approximate a general, potentially non-periodic function $Q(H)$ for a Hamiltonian $H$ with a norm bounded by $\|H\| \le B$ for some $B > 0$. We further suppose that $\lambda \ge B$.  A practical choice of $B$ is $B=\lambda$.

The strategy is to employ a rescaling of the Hamiltonian. Let $H' = H/B$, such that the rescaled Hamiltonian has a norm $\|H'\| \le 1$. To implement $Q(H) = Q(B H')$, we need to find a suitable degree-$d'$ Laurent polynomial $P(z)$ that approximates the rescaled function $Q(Bx)$ for $x \in [-1, 1]$. That is, we seek $P$ such that $P(e^{\i x}) \approx Q(Bx)$ for $x \in [-1, 1]$. This rescaling maps the spectral domain of interest of $H$, namely $[-B, B]$, onto the interval $[-1, 1]$, which can be embedded within a single $2\pi$ period of $P(e^{\i x})$. 
We can then apply Theorem~\ref{thm_gqsp with qdrift} to the rescaled Hamiltonian $H'$ and the polynomial $P(z)$. The complexity of this procedure is determined by two competing factors arising from the rescaling:
\begin{enumerate}
    \item The one-norm of the coefficients of the rescaled Hamiltonian $H' = \sum_k (\lambda_k/B) H_k$ becomes $\lambda' = \sum_k |\lambda_k/B| = \lambda/B$.
    
    \item   The degree $d'$ of the Laurent polynomial $P(e^{\i x})$ required to approximate $Q(Bx)$ on the interval $[-1,1]$. As $B$ increases, the function $Q(Bx)$ effectively ``compresses'' the behavior of $Q(x)$ from the larger interval $[-B,B]$ into the fixed interval $[-1,1]$. This causes the rescaled function to become more oscillatory or rapidly varying. To capture these features, a higher-degree polynomial is needed. The required degree $d'$ typically scales linearly with $B$ (see, e.g., the examples in \cite[Appendix C]{chakraborty2025quantum}). We therefore assume $d' \propto d \cdot B$, where $d$ is the base degree for approximating $Q$ on a constant-sized interval.
\end{enumerate}

By applying the complexity result from Theorem~\ref{thm_gqsp with qdrift} to $H'$ and $P(z)$, the total time complexity scales as:
\[
\widetilde{O}\left((\lambda')^2 (d')^2/\varepsilon^2\right) = \widetilde{O}\left( (\lambda/B)^2 (d \cdot B)^2/\varepsilon^2 \right) = \widetilde{O}\left(\lambda^2 d^2/\varepsilon^2\right).
\]
The dependence on the norm bound $B$ cancels out, and we recover the same complexity scaling with respect to the one-norm of the coefficients in the decomposition of the original Hamiltonian, $\lambda$, and the base degree $d$ required for the function approximation. 
\\

\noindent
\textbf{Comparison with prior works:~}From Ref.~\cite{chakraborty2025quantum}, it is known that Laurent polynomials can approximate most functions of interest, and moreover, any $d$-degree polynomial achievable by QSVT \cite{gilyen2019quantum} can also be approximated by a Laurent polynomial of the same degree. This means our framework effectively implements QSVT: given an initial state $\ket{\psi_0}$ and an observable $O$, we can estimate
$$
\braket{\psi_0|f(H)^{\dag} \, O \, f(H)|\psi_0},
$$ 
for any function $f$ that admits a bounded polynomial approximation of degree $d$. Standard QSVT \cite{gilyen2019quantum} requires a circuit depth of $\widetilde{O}(L\lambda d)$, and $\widetilde{O}(\eps^{-2})$ classical repetitions, while using $O(\log L)$ ancilla qubits, and sophisticated multi-qubit controlled operations. In contrast, both our algorithms exhibit no dependence on $L$, use only a single ancilla qubit, and has a shorter circuit depth whenever $\lambda d\ll L$. We summarize the ancilla requirement, circuit depth per coherent run, as well as the number of classical repetitions required for both our randomized algorithms, and the standard QSVT, in Table \ref{table:comparison-qsvt}. 
Finally, the recent method of \cite{chakraborty2025quantum} implements QSVT without block encodings, using only a single ancilla qubit. This approach requires circuit depth $\widetilde{O}(L(d\lambda)^{1+\frac{1}{2k}})$, using $2k$-th order Suzuki–Trotter formulas combined with classical extrapolation. By contrast, our method eliminates the dependence on $L$ and achieves shorter depth whenever $(d\lambda)^{1-1/(2k)} \ll L$. Since only low-order formulas ($k=1,2$) are feasible in near-term devices, this translates into broad parameter regimes where our randomized method offers an advantage, while using the same ancilla resources. We also note that the depth of the algorithm in \cite{chakraborty2025quantum} scales with the nested commutator norm $\lambda_{\mathrm{comm}}=O(\lambda)$; for certain Hamiltonians, it is possible that $\lambda_{\mathrm{comm}}\ll \lambda$, in which case that approach may be preferable.
\begin{table}[ht!!]
\begin{center}
    \resizebox{\columnwidth}{!}{
    \renewcommand{\arraystretch}{1.5} 
    \begin{tabular}{cccc}
    \hline
    Algorithm & Ancilla & Circuit depth per coherent run & Classical repetitions \\ \hline\hline

    Standard QSVT \cite{gilyen2019quantum} & $O(\log L) $ & $\widetilde{O}\left(Ld\lambda\right)$ & $\widetilde{O}\left(\varepsilon^{-2}\right)$ \\
    QSVT with Trotterization\cite{chakraborty2025quantum} & 1 & $\widetilde{O}\left(L(d\lambda)^{1+o(1)}\right)$ & $\widetilde{O}\left(\varepsilon^{-2}\right)$ \\ 
    \hline
     This work (direct randomized QSVT) & $1~\mathrm{or}~3$ & $\widetilde{O}\left(\lambda^2d^2\right)$ & $\widetilde{O}\left(\varepsilon^{-2}\right)$ \\
    This work (QSVT with qDRIFT) & 1 & $\widetilde{O}(\lambda^2d^2)$ & $\widetilde{O}(\eps^{-2}) $ \\
    \hline
\end{tabular}}
\caption{\small{Comparison with the complexities of different approaches to implement QSVT. Consider a Hermitian operator $H=\sum_{k=1}^{L}\lambda_k H_k$, with $\lambda=\max\{\sum_{k}|\lambda_k|,1\}$, and $\|H_k\|=1$. Then, for any initial state $\ket{\psi_0}$, and any observable $O$, each of these algorithms estimates $\braket{\psi_0|f(H)^{\dag}Of(H)|\psi_0}$ to additive accuracy $\eps\|O\|$, where $f(H)$ can be approximated by a $d$-degree polynomial bounded in the interval $[-1,1]$. For our first randomized algorithm, the number of ancilla qubits required increases when $f$ is approximated by a real polynomial, much like standard QSVT. Overall, for each procedure, we compare the total number of ancilla qubits required, the circuit depth per coherent run, and the total number of classical repetitions needed.
\label{table:comparison-qsvt}}}
\end{center}
\end{table}

From Theorem~\ref{thm_gqsp with qdrift}, our second randomized algorithm based on qDRIFT requires only a single ancilla qubit whenever the target function admits a Laurent polynomial approximation. Hence, it uses just one ancilla qubit regardless of whether the polynomial is real or complex. By contrast, our first randomized algorithm, structurally similar to standard QSVT, requires three ancilla qubits to implement real polynomials.
 
\subsubsection{Expectation value with respect to the normalized quantum state}

Often, for many applications of QSVT, we are interested in estimating the expectation value $\braket{\psi|O|\psi}$, where
\begin{align*}
    \ket{\psi}:=\frac{P(e^{\i H})\ket{\psi_0}}{\|P(e^{\i H}) \ket{\psi_0}\|}.
\end{align*}
is the normalized state. When $\|P(e^{\i H})\ket{\psi_0}\|\geq \eta$, standard QSVT \cite{gilyen2019quantum} prepares $\ket{\psi}$ using $\widetilde{O}(1/\eta)$ rounds of fixed-point amplitude amplification \cite{yoder2014fixed}. This increases the circuit depth to $\widetilde{O}(L\lambda d/(\eta\eps))$. By contrast, as discussed in Sec.~\ref{subsec:randomized-qsvt-1}, our first randomized algorithm cannot accommodate this procedure, and estimating the target expectation value instead requires $O(\eta^{-4}\eps^{-2})$ classical repetitions.

Our second algorithm, however, integrates fixed-point amplitude amplification seamlessly. Applying this procedure to the output of any interleaved circuit $W$ yields a modified circuit
$$
\widetilde{W}=\prod_{j=1}^{K/2} W e^{\i\phi_j\ket{0,\psi_0}\bra{0,\psi_0}}W^{\dag} e^{\i\theta_j (\ket{0}\bra{0}\otimes I)},
$$
where $\theta_j,\phi_j \in [0,2\pi]$ are QSP phase angles prescribed by the fixed point amplitude amplification algorithm, and $K=O(\eta^{-1}\log(1/\eps))$. Here, $e^{\i\phi_j\ket{0,\psi_0}\bra{0,\psi_0}}$ is a reflection about the initial state $\ket{\psi_0}$, controlled by a single ancilla qubit. Crucially, $\widetilde{W}$ retains the same interleaved structure as $W$, with only an increased circuit depth.

Thus, Theorem~\ref{thm_HSVT with qDRIFT} can be extended to incorporate the modified sequence $\widetilde{W}$, enabling the preparation of $\ket{\psi}$ and the subsequent estimation of $\bra{\psi}O\ket{\psi}$. Formally we have

\begin{thm}[QSVT with qDRIFT]
    \label{thm_gqsp with qdrift and aa}
    Let $H=\sum_{k=1}^L \lambda_k H_k$ be the Hermitian decomposition of $H$, with $\|H_k\|=1$, and $\lambda=\max\{\sum_{k}|\lambda_k|,1\}$. Also let $\varepsilon\in(0,1/2)$ be the precision parameter. Suppose $P(z)=\sum_{j=-d}^d a_j z^j$ is a Laurent polynomial bounded by $1$ on $\mathbb{T}:=\{ x\in \mathbb{C} : |x|=1 \}$ and $\|P(e^{\i H})\ket{\psi_0}\|\ge \eta$ for some $\eta\in (0,1)$. Furthermore, define 
    $$
    \ket{\psi}=P(e^{\i H})\ket{\psi_0}/\|P(e^{\i H})\ket{\psi_0}\|.
    $$ 
    Then there exists a quantum algorithm that, for any observable $O$ and initial state $\ket{\psi_0}$, estimates $$\bra{\psi} O \ket{\psi},$$ 
    with an additive error at most $\varepsilon \|O\|$ and constant success probability, using two ancilla qubits. 
    The maximum quantum circuit depth is
    \be
    \widetilde{O}\left(\dfrac{\lambda^2 d^2}{\eta^{2}}\right),
    \ee
    and the total time complexity is
    \be
    \widetilde{O}\left(\dfrac{\lambda^2d^2}{\eta^2\eps^2}\right).
    \ee
\end{thm}
\begin{proof}
By Theorem~\ref{thm:generalized-QSP}, there exists a GQSP circuit, which we denote $U_f$, that prepares the unnormalized state. This circuit takes the form of an interleaved sequence of Hamiltonian evolutions and single-qubit rotations:
\begin{align*}
    (\bra{0} \otimes I) U_f \ket{0}\ket{\psi_0} = P(e^{\i H})\ket{\psi_0}.
\end{align*}
The length of this interleaved sequence is $M = 2d+1$, involving $d$ applications of (controlled) $e^{\i H}$ and $d$ applications of (controlled) $e^{-\i H}$.

We employ fixed-point amplitude amplification \cite{yoder2014fixed} to prepare the normalized state $\ket{\psi}$ with an additive error at most $\eps/4$. This ensures that the expectation value of $O$ is with an additive error at most $\|O\|\eps/2$. Given that the success probability the initial preparation is $\|P(e^{\i H})\ket{\psi_0}\|^2 \ge \eta^2$, the amplification procedure requires $O(
\log(1/\eps)/\eta)$ rounds of $U_f$ and its inverse, interleaved with reflection operators of the form $e^{\i\phi_j\ket{0,\psi_0}\bra{0,\psi_0}}$ and $ e^{\i\theta_j (\ket{0}\bra{0}\otimes I)}$. 

The complete quantum circuit for preparing $\ket{\psi}$ can be viewed as a single interleaved Hamiltonian evolution sequence. The length of this new sequence, denoted by $M'$, is determined by the length of the GQSP circuit ($\widetilde{O}(d)$) and the number of amplification rounds ($O(\log(1/\eps)/\eta)$). Therefore, the total length is $M' = \widetilde{O}(d/\eta)$. The Hamiltonians involved are still simply $H$, so the parameter $\lambda$ remains unchanged.

We can now directly apply Theorem~\ref{thm_HSVT with qDRIFT} to estimate the expectation value $\bra{\psi} O \ket{\psi}$.  Substituting $M'$ yields the maximum quantum circuit depth of $$
\widetilde{O}\left(\lambda^2 d^2\eta^{-2}\right),
$$ 
and a total time complexity of $$
\widetilde{O}\left(\lambda^2 d^2\eta^{-2}/\varepsilon^2\right).
$$
Let us now analyze the number of ancilla qubits our procedure requires. The GQSP protocol itself requires a single ancilla qubit. Now, one additional ancilla qubit is needed to implement the multi-controlled reflection operators required for amplitude amplification. So overall, the entire procedure requires two ancilla qubits. This completes the proof.
\end{proof}
~\\
\textbf{Comparison with prior works:~}A wide range of quantum algorithms have been developed for estimating expectation values of the form $\braket{\psi|O|\psi}$, where
$$
\ket{\psi}=\dfrac{f(H)\ket{\psi_0}}{\|f(H)\ket{\psi_0}\|},
$$
and $f(H)$ is some function of a given Hamiltonian $H$. These algorithms fall under the umbrella of quantum linear algebra, since $f(H)$ typically represents a transformation of the eigenvalues of $H$, such as $H^{-1}$, $e^{-H}$, or $H^t$. 
\begin{table}[ht!!]
\begin{center}
    \resizebox{\columnwidth}{!}{
    \renewcommand{\arraystretch}{1.5} 
    \begin{tabular}{cccc}
    \hline
    Algorithm & Ancilla & Circuit depth per coherent run & Classical repetitions \\ \hline\hline

    Standard QSVT \cite{gilyen2019quantum} & $O(\log L) $ & $\widetilde{O}\left(Ld\lambda/\eta\right)$ & $\widetilde{O}\left(\varepsilon^{-2}\right)$ \\
    QSVT with Trotterization\cite{chakraborty2025quantum} & 2 & $\widetilde{O}\left(L(d\lambda/\eta)^{1+o(1)}\right)$ & $\widetilde{O}\left(\varepsilon^{-2}\right)$ \\ 
    \hline
     This work (direct randomized QSVT) & $1~\mathrm{or}~3$ & $\widetilde{O}\left(\lambda^2d^2\right)$ & $\widetilde{O}\left(\varepsilon^{-2}\eta^{-4}\right)$ \\
    This work (QSVT with qDRIFT) & 2 & $\widetilde{O}(\lambda^2d^2)$ & $\widetilde{O}(\eps^{-2}) $ \\
    \hline
\end{tabular}}
\caption{\small{Comparison with the complexities of different approaches to implement QSVT. Consider a Hermitian operator $H=\sum_{k=1}^{L}\lambda_k H_k$, with $\lambda=\max\{\sum_{k}|\lambda_k|,1\}$, and $\|H_k\|=1$. Then, for initial $\ket{\psi_0}$, and any observable $O$, each of these algorithms estimates $\braket{\psi|O|\psi}$ to additive accuracy $\eps\|O\|$, where $f(H)$ can be approximated by a $d$-degree polynomial bounded in the interval $[-1,1]$, and $\ket{\psi}=f(H)\ket{\psi_0}/\|f(H)\ket{\psi_0}\|$. For our first randomized algorithm, the number of ancilla qubits required increases when $f$ is approximated by a real polynomial, much like standard QSVT. Overall, for each procedure, we compare the total number of ancilla qubits required, the circuit depth per coherent run, and the total number of classical repetitions needed.
\label{table:comparison-qsvt-aa}}}
\end{center}
\end{table}
For standard QSVT, given a block encoding of $H$ with a Hermitian decomposition into $L$ local terms, any $f(H)$ that admits a bounded polynomial approximation of degree $d$ can be implemented. The expectation value $\braket{\psi|O|\psi}$ can then be estimated in three ways, each of which requires a block encoding of $H$, using $O(\log L)$ ancilla qubits, sophisticated controlled operations, and a circuit depth depending on $L$. These variants are as follows:

\begin{itemize}
    \item \emph{Without quantum amplitude amplification or estimation:~}The QSVT circuit prepares the state
    $$
    \approx \ket{\bar{0}} \otimes f(H)\ket{\psi_0}+\ket{\Phi}^{\perp},
    $$
    where $(\ket{\bar{0}}\bra{\bar{0}}\otimes I)\ket{\Phi}^{\perp}=0$. This requires a circuit depth of $\widetilde{O}(L\lambda d)$ and $O(\log L)$ ancillas. The estimation proceeds by measuring the first register at the end of each run; whenever it is in $\ket{\bar{0}}$, one measures $O$ on the second register. Since this occurs with probability $\Theta(1/\eta^2)$, one needs $O(\eta^{-2}\eps^{-2})$ repetitions. By comparison, our randomized algorithm avoids block encodings and uses only a single ancilla qubit, with circuit depth $\widetilde{O}(\lambda^2 d^2)$ per run. However, it requires $\widetilde{O}(\eta^{-4}\eps^{-2})$ repetitions. Thus, we achieve shorter circuit depth whenever $\lambda d \ll L$, although our total complexity is higher.

    \item \emph{With quantum amplitude amplification and classical repetitions:~}Here, fixed-point amplitude amplification is used to directly prepare $\ket{\psi}$, followed by $\widetilde{O}(\eps^{-2})$ classical repetitions. The circuit depth per coherent run increases to $\widetilde{O}(L\lambda d / \eta)$, with $O(\log L)$ ancillas.
    In contrast, from Theorem~\ref{thm_HSVT with qDRIFT}, our randomized algorithm achieves circuit depth $\widetilde{O}(\lambda^2 d^2/\eta^2)$ with only two ancillas. This yields a depth advantage whenever $\lambda d / \eta \ll L$. A detailed side-by-side comparison is given in Table~\ref{table:comparison-qsvt-aa}.

    \item \emph{Coherent estimation using quantum amplitude estimation (QAE):~}If block encodings of both $H$ and $O$ are available, one may use QAE to coherently estimate $\braket{\psi|O|\psi}$, reducing the dependence on precision from $1/\eps^2$ to $1/\eps$. However, this comes at the price of exponentially larger circuit depth. Moreover, for observables of the form
    $$
    O=\sum_{j=1}^{L_O} h_j O_j,
    $$
    the block encoding of $O$ incurs additional cost: ancilla requirements increase to $O(\log(L\cdot L_O))$, and the total depth becomes
    $$
    \widetilde{O}\left(\dfrac{L\lambda d}{\eps\eta}+\dfrac{L_O}{\eps}\right).
    $$
    These overheads make QAE-based approaches impractical for early fault-tolerant architectures, so we do not benchmark against them.
\end{itemize}

Finally, we compare our randomized algorithm with the recent block-encoding-free QSVT approach based on $2k$-th order Suzuki–Trotter formulas \cite{chakraborty2025quantum}. Their method achieves circuit depth scaling as $\widetilde{O}(L(d\lambda/\eta)^{1+{1}/{(2k)}})$, while our randomized algorithm removes the dependence on $L$. Consequently, our approach achieves a shorter depth per coherent run whenever 
$$(d\lambda/\eta)^{1-1/(2k)}\ll L.
$$ 
This advantage is particularly pronounced for small Trotter orders ($k=1,2$), which are the only practically feasible choices in early fault-tolerant architectures.

Overall, both our randomized algorithms require a circuit depth that has a quadratic dependence on the degree of the underlying polynomial. In the next section, we show that this is quite fundamental by proving a matching lower bound for any generic quantum algorithm implementing polynomial transformations within this access model.

\section{Lower bounds}
\label{sec:lower bound}

The circuit depth of standard QSVT scales linearly with the degree of the target polynomial, and this linear dependence is known to be optimal \cite{gilyen2019quantum, montanaro2024quantum}. 
In contrast, both of our randomized algorithms exhibit a quadratic dependence on the polynomial degree. 
This naturally raises the question: can this scaling be improved? In this section, we show that the answer is negative. 
Within the sampling-access model considered here, a quadratic dependence is unavoidable. 
More precisely, we establish lower bounds on both the sample complexity and circuit depth of any randomized quantum algorithm in this model that aims to implement a target function $f(x)$. 
These bounds demonstrate that our algorithms already achieve optimal scaling in this setting. 
Interestingly, the same limitations extend beyond our constructions, applying equally to other randomized approaches such as Hamiltonian simulation via qDRIFT \cite{campbell2019random} and qSWIFT \cite{nakaji2024qswift}, as well as the more recent randomized LCU methods \cite{wang2024qubit, chakraborty2024implementing}.

We focus on the canonical target function $f(x)=e^{\i xt}$, for which it is well known that a polynomial of degree $\widetilde{\Theta}(t)$ suffices to approximate the evolution. 
In the sampling-access model, we prove that both the sample complexity and the circuit depth per round of any randomized quantum algorithm must scale as $\Omega(t^2)$. 
This shows that a quadratic dependence on the polynomial degree is unavoidable for any method, including QSVT, that aims to implement general polynomial transformations of a Hermitian operator.

To make this concrete, we restrict to operators $H$ that are linear combinations of Pauli strings, a special but representative subclass of the Hamiltonians considered in this work. 
In this setting, we do not have access to the full description of $H$. Instead, only sampling access to individual Pauli terms according to their probability weights is available. 
We refer to this setting as the \emph{Pauli sample access model}, formally defined below.

\begin{defn}[Pauli sample access model]
Suppose $H = \sum_{k=1}^L \lambda_k P_k$ is the Pauli decomposition of the Hermitian operator $H$. In the \textit{Pauli sample access model}, the algorithm does not have access to the full description of $H$. Instead, it can only access $H$ by sampling from the distribution
$$\mc D = \left\{ ( |\lambda_k|/\lambda,~e^{\i \theta_k} P_k) : k\in[L] \right\},$$ where $\lambda = \sum_k |\lambda_k|$ and $\lambda_k=e^{\i \theta_k} |\lambda_k|$.
\end{defn}

As discussed, each query to the distribution $\mc D$ yields a Pauli operator $e^{\i \theta_k} P_k$ up to some phase $e^{\i \theta_k}$ with probability $|\lambda_k|/\lambda$. 
We treat this distribution as a black box, meaning that the algorithm may sample from it, but has no direct knowledge of the probabilities themselves. Given a target function $f(x)$ and two states $\ket{\psi_0}, \ket{\psi_1}$, the task is to estimate 
$$
\bra{\psi_1} f(H) \ket{\psi_0},
$$
with an error at most $\varepsilon$, while minimizing the number of queries to $\mc D$. 
To prove our lower bound, we establish a connection between this estimation task and a tomography problem, formulated as follows.

\begin{lem}[Theorem 9.2 of \cite{van2023quantum}]
\label{lemma for state tomography}
Let $\ket{\psi}=\sum_{k=1}^{L} x_k \ket{k}$ be a unknown quantum state. 
Suppose we are given only classical samples obtained from computational-basis measurements of $\ket{\psi}$. 
Then, in order to estimate all amplitudes up to their magnitudes, i.e.\ to compute $|x_k|\pm \delta$ for every $k\in[L]$, the sample complexity is lower bounded by $\Omega(1/\delta^2)$.
\end{lem}

From the proof of Lemma 9.3 of \cite{van2023quantum}, this lower bound continues to hold even in the case where half of the amplitudes are
$$
x_k=\sqrt{(1 + \eta)/L},
$$ 
and the other half are 
$$
x_k=\sqrt{(1 - \eta)/L}
$$ 
for some small $\eta$.
Note that the difference between these two quantities is $\Theta(\eta/\sqrt{L})$. So the error in Lemma \ref{lemma for state tomography} should be set as $\delta=\Theta(\eta/\sqrt{L})$. In what follows, we use this structured setting to derive two distinct lower bounds on the sample complexity, depending on the target accuracy of the estimate. 
In particular, we consider the case $L=O(1)$, in which the sample complexity in Lemma \ref{lemma for state tomography} is lower bounded by $\Omega(1/\eta^2)$.

\begin{thm}[Sample complexity lower bound: Arbitrary accuracy]
\label{thm_sample lower bound}
In the Pauli sample access model, the sample complexity of estimating $\bra{\psi_1}e^{\i H t} \ket{\psi_0}$ up to arbitrary accuracy is lower bounded by $\Omega(t^2)$.
\end{thm}

\begin{proof}
Consider a quantum state $\ket{\psi}=\sum_{k=1}^{L} x_k \ket{k}$, such that half of its amplitudes are $$x_k=\sqrt{(1 + \eta)/L},$$
and the other half 
$$x_k=\sqrt{(1 - \eta)/L}.$$  
Furthermore, suppose 
$$
H = \sum_{k=1}^L \alpha_k X_k,\quad~\text{with~} \alpha_k = |x_k|^2 \in\left\{ (1+\eta)/L,~(1-\eta)/L\right\},
$$
and $X_k$ being the Pauli $X$ acting on the $k$-th qubit. Consider the distribution $$
\mc D=\{(|x_k|^2, X_k): k\in[L]\}.
$$
Then, measuring $\ket{\psi}$ in the computational basis is equivalent to sampling from $\mc D$. 
Observe that
\[
e^{\i H t} = \prod_{k=1}^L e^{\i \alpha_k X_k t},
\]
furthermore,
\be
\label{eq1-0215}
e^{\i \alpha_k X_k t}=\left[\begin{array}{cc}  \vspace{.2cm}
\cos \left(\alpha_k t\right) & \i \sin \left(\alpha_k t\right) \\
\i \sin \left(\alpha_k t\right) & \cos \left(\alpha_k t\right)
\end{array}\right]. 
\ee
By choosing $t=\frac{L}{\eta}\cdot \frac{\pi}{4} = \Theta(1/\eta)$, we obtain 

\be
e^{\i \alpha_k X_k t}=
\begin{cases}
    \left[\begin{array}{cc} \vspace{.2cm}
    \cos\left(\big(1+\tfrac{1}{\eta}\big)\tfrac{\pi}{4}\right) 
    & \i \sin\left(\big(1+\tfrac{1}{\eta}\big)\tfrac{\pi}{4}\right)  \\
    \i \sin\left(\big(1+\tfrac{1}{\eta}\big)\tfrac{\pi}{4}\right)  
    & \cos\left(\big(1+\tfrac{1}{\eta}\big)\tfrac{\pi}{4}\right)
\end{array}\right], & \text{~if~}\alpha_k={(1+\eta)}/{L} \\~\\
\left[\begin{array}{cc} \vspace{.2cm}
    \cos\left(\big(1-\tfrac{1}{\eta}\big)\tfrac{\pi}{4}\right) 
    & -\i \sin\left(\big(1-\tfrac{1}{\eta}\big)\tfrac{\pi}{4}\right)  \\
    -\i \sin\left(\big(1-\tfrac{1}{\eta}\big)\tfrac{\pi}{4}\right)  
    & \cos\left(\big(1-\tfrac{1}{\eta}\big)\tfrac{\pi}{4}\right)
\end{array}\right], & \text{~if~}\alpha_k={(1-\eta)}/{L}.
\end{cases}
\ee

The first columns of both matrices form an orthogonal basis, written as
\beas
    &\ket{\alpha_+} := \cos\left(\big(1+\tfrac{1}{\eta}\big)\tfrac{\pi}{4}\right) \ket{0}
     + \i \sin\left(\big(1+\tfrac{1}{\eta}\big)\tfrac{\pi}{4}\right)\ket{1}, \\
    &\ket{\alpha_-} := \cos\left(\big(1-\tfrac{1}{\eta}\big)\tfrac{\pi}{4}\right) \ket{0}
     - \i \sin\left(\big(1-\tfrac{1}{\eta}\big)\tfrac{\pi}{4}\right)\ket{1} .
\eeas
Measuring in the $\ket{\alpha_\pm}$ basis tells us whether $\alpha_k={(1+\eta)}/{L}$ or $\alpha_k={(1-\eta)}/{L}$. Therefore, if we could compute $e^{\i Ht}\ket{0^L}$, and measure the resulting state, we could estimate every $\alpha_k$ and consequently every $x_k$, exactly. 

From Lemma \ref{lemma for state tomography}, we know that the sample complexity of computing all $x_k$ with an additive accuracy of $\Theta(\eta)$ is $\Omega(1/\eta^2)$ when $L=O(1)$. So, the sample complexity lower bound for computing $e^{\i Ht}\ket{0^L}$ is $\Omega(1/\eta^2)=\Omega(t^2)$. As $L=O(1)$, this lower bound also holds for the task of computing 
$$
\bra{\psi_1}e^{\i Ht}\ket{\psi_0},$$ 
as we could simply check all $\bra{\psi}e^{\i Ht}\ket{0^L}$ with $\ket{\psi}\in\{\ket{\alpha_+}, \ket{\alpha_-}\}^{\otimes L}$ which would return the information about $x_k$, for all $k\in [L]$.
\end{proof}

It is possible to slightly modify the above proof to show that the sample complexity is lower bounded by $\Omega(t^2/\varepsilon^2)$ when the task is to estimate 
$$
\braket{\psi_1|e^{\i H t}|\psi_0}
$$
with an error at most $\eps$. 

\begin{thm}[Sample complexity lower bound: Fixed accuracy]
\label{thm_sample lower bound 2}
In the Pauli sample access model, the sample complexity of estimating $\bra{\psi_1}e^{\i H t} \ket{\psi_0}$ up to accuracy $\varepsilon$ is lower bounded by $\Omega(t^2/\varepsilon^2)$.
\end{thm}

\begin{proof}
We consider the case $L=2$. Let $\eta>0$ be a small parameter. Define
\[
x^{(+)} := \sqrt{\tfrac{1+\eta}{2}}, \qquad
x^{(-)} := \sqrt{\tfrac{1-\eta}{2}}.
\]
As in Lemma~\ref{lemma for state tomography}, we construct two promise instances with amplitudes $(x^{(+)},x^{(-)})$ and $(x^{(-)},x^{(+)})$, respectively. Define the operator
\[
H = \alpha_1 X_1 + \alpha_2 X_2, \qquad
\alpha_1, \alpha_2 \in \Bigl\{\tfrac{1+\eta}{2},\, \tfrac{1-\eta}{2}\Bigr\},
\]
and set $t = \Theta\!\left(1/\sqrt{\eta}\right)$.
Then, for each $k \in \{1,2\}$,
\[
e^{\i \alpha_k X_k t}
= \begin{bmatrix}
\cos(\alpha_k t) & \i \sin(\alpha_k t)\\[4pt]
\i \sin(\alpha_k t) & \cos(\alpha_k t)
\end{bmatrix}.
\]
Let $\ket{\phi_1}$ denote the first column corresponding to $\alpha_k=(1+\eta)/2$, and $\ket{\phi_2}$ the first column for $\alpha_k=(1-\eta)/2$.  
So, we have
\[
\langle \phi_1 | \phi_2 \rangle = 1 - \Theta(\eta),
\]
and hence
\[
\|\ket{\phi_1}-\ket{\phi_2}\| = \Theta(\sqrt{\eta}).
\]
 
Take the input $\ket{\psi_0}=\ket{0}\otimes\ket{0}$. Then
\[
e^{\i Ht}\ket{\psi_0} \in \bigl\{\, \ket{\phi_1}\otimes\ket{\phi_2},\; \ket{\phi_2}\otimes\ket{\phi_1}\,\bigr\}.
\]
Define
\[
\ket{\tau} := \frac{\ket{\phi_1}-\ket{\phi_2}}{\|\ket{\phi_1}-\ket{\phi_2}\|}, 
\qquad
\ket{\psi_1} := \ket{\tau}\otimes\ket{0}.
\]
Consider
\[
a =\bra{\psi_1} e^{\i Ht} \ket{\psi_0}.
\]
If the time evolved state is $\ket{\phi_1}\otimes\ket{\phi_2}$, then
$$
a= \bra{\psi_1}\left(\ket{\phi_1}\otimes\ket{\phi_2}\right)=\braket{\tau|\phi_1}\cdot \braket{0|\phi_2}=\Theta(\|\ket{\phi_1}-\ket{\phi_2}\|)=\Theta(\sqrt{\eta}),
$$  
while if it is $\ket{\phi_2}\otimes\ket{\phi_1}$, then the value of $a$ is the negative of this. Thus, the two promise instances yield values of $a$ differing by
\[
\Delta a = \Theta\left(\left\|\ket{\phi_1}-\ket{\phi_2}\right\|\right) = \Theta(\sqrt{\eta}).
\]

To resolve the two cases, an algorithm must estimate $a$ to additive accuracy
\[
\varepsilon = \Theta(\sqrt{\eta}).
\]
Since $t=\Theta(1/\sqrt{\eta})$, we equivalently have $\eta = \Theta\!\left(1/t^2\right)$, and $ \varepsilon = \Theta\!\left(1/t\right)$.
So,
\[
\frac{t^2}{\varepsilon^2} = \Theta\!\left({1}/{\eta^2}\right).
\]

By Lemma~\ref{lemma for state tomography}, distinguishing the two promise families requires $\Omega(1/\eta^2)$ queries to $\mathcal D$.  
Substituting the relations above gives the claimed lower bound of $\Omega(t^2/\varepsilon^2)$ queries.
\end{proof}

The sample complexity lower bound (Theorem~\ref{thm_sample lower bound 2}) also yields a lower bound on the circuit depth of randomized algorithms in the Pauli sample access model. 
Consider algorithms of the form
\[
W_{Q_1,\ldots,Q_c} \;=\; U_1 Q_1 U_2 Q_2 \cdots U_c Q_c U_{c+1},
\]
where each $Q_i$ is generated independently from a constant number $s=O(1)$ of samples drawn from $\mathcal D$, and $U_1,\ldots,U_{c+1}$ are fixed unitaries independent of the samples. Moreover, we assume that the depth of each $Q_i$ is constant. 

Suppose the algorithm executes $R$ independent runs of this circuit and estimates some classical quantity (e.g.\ the empirical mean) from the outcomes. Then the total number of oracle queries is $N_{\mathrm{tot}}=R \cdot (c s)$. 
By Hoeffding’s inequality, estimation with additive precision $\varepsilon$ requires $R=\Theta(1/\varepsilon^2)$ repetitions \cite{canetti1995lower}, so $N_{\mathrm{tot}}=\Theta(cs/\varepsilon^2)$. 
Since Theorem~\ref{thm_sample lower bound 2} asserts that $N_{\mathrm{tot}}=\Omega(t^2/\varepsilon^2)$, we conclude that
\[
c \cdot s \;=\; \Omega(t^2).
\]
Since each sampled unitary has constant depth, the circuit depth per run is therefore $\Omega(t^2)$. This is formally stated as follows.

\begin{cor}[Circuit complexity lower bound]
\label{Circuit complexity lower bound}
In the Pauli sample access model, assume that the number of classical repetitions is independent of the evolution time $t$.
Then the circuit depth of any randomized quantum algorithm for estimating $\bra{\psi_1}e^{\i H t} \ket{\psi_0}$ is lower bounded by $\Omega(t^2)$.
\end{cor}


Our lower bound results are fundamental. They demonstrate that any generic quantum algorithm within this access model that estimates properties of the state $f(H)\ket{\psi_0}$, for any broad class of functions $f$, must incur circuit depth scaling quadratically with the minimum degree of the polynomial approximating $f(H)$. This establishes that both of our randomized algorithms for QSVT are essentially optimal in their dependence on the polynomial degree. Beyond QSVT, these lower bounds also apply to recently introduced randomized LCU methods \cite{chakraborty2024implementing, wang2024qubit}. Moreover, they extend naturally to randomized Hamiltonian simulation techniques such as qDRIFT \cite{campbell2019random}, its higher-order variant qSWIFT \cite{nakaji2024qswift}, and randomized multiproduct formulas \cite{faehrmann2022randomizingmulti}. All of these methods exhibit a quadratic dependence on the evolution time $t$, and our results show that this scaling is in fact optimal within this access model. 
In this sense, our lower bounds delineate the fundamental performance limits of randomized quantum algorithms for implementing $f(H)$ under sample access to  $H$, making explicit which trade-offs are unavoidable and providing a benchmark for what future advances in this framework can realistically achieve. In the next section, we consider concrete applications of our randomized algorithms for QSVT.

\section{Applications}
\label{sec:applications}
We apply our algorithms to two problems of broad practical interest: (i) solving quantum linear systems, and (ii) ground state property estimation.

\vspace{.3cm}

\subsection{Quantum linear systems}
\label{subsec:qls}
The quantum linear systems algorithm \cite{harrow2009quantum, childs2017qls, costa2022optimal} has been a cornerstone of quantum algorithms, with wide-ranging applications to regression \cite{ chakraborty_et_al:LIPIcs.ICALP.2019.33, chakraborty2023quantum}, machine learning, and the solution of differential equations \cite{childs2020quantum, Krovi2023improvedquantum}. Formally, let $A\in \mathbb{C}^{N\times N}$ is a matrix with operator norm $\|A\|=1$, whose singular values lie in the range $[-1,-1/\kappa]\cup[1/\kappa, 1]$, where $\|A^{-1}\|=\kappa$ is the condition number of $A$. 
Given efficient access to a state preparation oracle for $\ket{b}$, the goal is to estimate the expectation value 
$\braket{x|O|x}$ to additive accuracy $\eps$, where $$
\ket{x}=\dfrac{A^{-1}\ket{b}}{\left\|A^{-1}\ket{b}\right\|}.$$ 

The complexity of quantum linear systems algorithms has been refined through successive improvements, with most prior work assuming access to a block encoding $U_A$ of matrix $A$. Algorithms based on LCU or QSVT typically achieve linear dependence on $\kappa$ by employing sophisticated subroutines such as variable-time amplitude amplification (VTAA) \cite{ambainis2012variabletime, childs2017qls, chakraborty_et_al:LIPIcs.ICALP.2019.33}. An alternative line of work, inspired by adiabatic quantum computation \cite{costa2022optimal}, makes only $O(\kappa\log(1/\eps))$ queries to the block encoded operator, which is optimal, without having to use VTAA. 

More recently, Ref.~\cite{chakraborty2025quantum} demonstrated that, given the full description of a Pauli decomposition of $A$, the target expectation value can be estimated using a single interleaved sequence of Hamiltonian evolutions. This sequence, which combines a discretized adiabatic evolution with an eigenstate filtering step, takes the form 
$$
W := V_0 \prod_{\ell=1}^M e^{\i H^{(\ell)}} V_{\ell},
$$ 
with total length of $M = \widetilde{O}(\kappa)$ and acts on a system of $\log(N)+4$ qubits. In this construction, Hamiltonian evolution is approximated using higher-order Trotterization, yielding a quasi-linear dependence on $\kappa$.

In this work, we develop a randomized quantum linear systems algorithm that builds on the above framework. Specifically, we estimate the final expectation value by replacing the Hamiltonian evolution operators in $W$ with qDRIFT. In this way, the entire procedure can be viewed as a particular instance of the operator $W$, and its complexity follows directly by applying the substitutions prescribed in Theorem~\ref{thm_HSVT with qDRIFT}. 

\begin{thm}[Randomized Quantum linear systems algorithm]
\label{thm:qls}
    Let $Ax=b$ be a system of linear equations, where $A\in\mb C^{N \times N}$ with $\|A\|=1$ and $\|A^{-1}\|=\kappa$, and let $\varepsilon\in(0,1/2)$ be the precision parameter.
    Let $A=\sum_{k=1}^L \lambda_k P_k$ be the Pauli decomposition of $A$.
    Given a unitary oracle $U_b$ that prepares the state $\ket{b}$ and an observable $O$.
    There exists a randomized quantum algorithm that estimates 
    $$
   \braket{x|O|x},
    $$ 
    with an additive error at most $\varepsilon\|O\|$, using only $\log(N)+4$ qubits. 
    The maximum quantum circuit depth is $\widetilde{O}(\lambda^2 \kappa^{2})$, while the total time complexity and the number of queries to $U_b$ and $U_b^{\dag}$ are $\widetilde{O}(\lambda^2 \kappa^2/\varepsilon^2)$, where $\lambda = \sum_{k=1}^L |\lambda_k|$.
\end{thm}
\begin{proof}
    Ref.~\cite{chakraborty2025quantum} establishes that the quantum linear systems problem can be reduced to estimating an expectation value for an interleaved Hamiltonian evolution sequence $W$, acting on a space of $\log(N)+4$ qubits. This sequence has a total length of $M = \widetilde{O}(\kappa)$. The Hamiltonians $H^{(\ell)}$ within it are constructed from the Pauli decomposition of $A$ such that the norm-sum parameter $\lambda$ scales as $O(\sum_{k=1}^L |\lambda_k|)$. 

    A key property of this construction is that each Hamiltonian $H^{(\ell)}$ is a sum of $O(L)$ simple terms, $H^{(\ell)} = \sum_{k} H_k^{(\ell)}$, where the evolution $e^{\i H_k^{(\ell)} \tau}$ under each individual term (for any $\tau \in \mathbb{R}$) can be implemented efficiently. Specifically, each such evolution, as well as each unitary $V_\ell$, requires a circuit with $\mathrm{polylog}(N)$ elementary gates and at most two queries to the oracle $U_b$ and its inverse.

    This setup is perfectly suited for the qDRIFT-based framework presented in Theorem \ref{thm_HSVT with qDRIFT}. Substituting the parameters $M = \widetilde{O}(\kappa)$ and $\lambda$ yields a total time complexity and query complexity of $\widetilde{O}(\lambda^2 M^2/\varepsilon^2) = \widetilde{O}((\lambda \kappa/\varepsilon)^2)$, and a maximum circuit depth of $\widetilde{O}(\lambda^2 M^2) = \widetilde{O}(\lambda^2 \kappa^{2})$. 
\end{proof}

Remarkably, our algorithm requires only four ancilla qubits, avoids block encodings entirely, and achieves a circuit depth independent of $L$. We compare this to other works in Table \ref{table:comparison-qls}.

\begin{table}[ht!!]
\begin{center}
    \resizebox{\columnwidth}{!}{
    \renewcommand{\arraystretch}{1.5} 
    \begin{tabular}{cccc}
    \hline
    Algorithm & Ancilla & Circuit depth per coherent run & Classical repetitions \\ \hline\hline

    Standard QSVT \cite{gilyen2019quantum} & $\lceil \log_2 L \rceil+1 $ & $\widetilde{O}\left(L\lambda\kappa^2\right)$ & $\widetilde{O}\left(\varepsilon^{-2}\right)$ \\
    State-of-the-art\cite{costa2022optimal} & $\lceil \log_2 L \rceil+6 $ & $\widetilde{O}\left(L\lambda\kappa\right)$ & $\widetilde{O}\left(\varepsilon^{-2}\right)$ \\ 
    QSVT with Trotterization\cite{chakraborty2025quantum} & 4 & $\widetilde{O}\left(L(\lambda\kappa)^{1+o(1)}\right)$ & $\widetilde{O}\left(\varepsilon^{-2}\right)$ \\ 
    \hline
    Other randomized methods \cite{wang2024qubit, chakraborty2024implementing} & $1$ & $\widetilde{O}(\lambda^2\kappa^2)$ & $\widetilde{O}\left(\kappa^4\eps^{-2}\right)$\\
    This work (Theorem \ref{thm:qls}) & 4 & $\widetilde{O}(\lambda^2\kappa^2)$ & $\widetilde{O}(\eps^{-2}) $ \\
    \hline
\end{tabular}}
\caption{\small{Comparison of the complexities of different quantum linear systems algorithms. Consider a Hermitian operator $A=\sum_{k=1}^{L}\lambda_k P_k$, with $\lambda = \sum_{k=1}^L |\lambda_k|$, $\|A\|=1$, and $\|A^{-1}\|=\kappa$. Then, given a procedure for preparing the state $\ket{b}$, and any observable $O$, each of these algorithms estimates $\braket{x|O|x}$ to additive accuracy $\eps\|O\|$. We compare the total number of ancilla qubits required, the circuit depth per coherent run, and the total number of classical repetitions needed.
\label{table:comparison-qls}}}
\end{center}
\end{table}

Standard QSVT requires access to a block encoding of $H$, which, as discussed earlier, is costly to construct: it demands several ancilla qubits and intricate controlled operations \cite{gilyen2019quantum}. The resulting algorithm has circuit depth $\widetilde{O}(\lambda L \kappa^2)$ and uses $\lceil\log_2 L\rceil+1$ ancilla qubits. Moreover, if the observable $O$ is measured at the end of each run, the entire procedure must be repeated $\widetilde{O}(1/\eps^2)$ times. In contrast, our randomized approach achieves strictly shorter depth whenever $\lambda\ll L$, while using only four ancilla qubits.

For QSVT, the dependence on $\kappa$ can be improved to linear using variable-time amplitude amplification (VTAA), but this comes at the cost of even more ancilla qubits and sophisticated controlled operations, making it impractical for early fault-tolerant devices. Similarly, if a block encoding of $O$ is available, one may employ quantum amplitude estimation to quadratically reduce the dependence on inverse precision from $1/\eps^2$ to $1/\eps$. However, this improvement is offset by an exponential increase, in terms of $\varepsilon$, in circuit depth, to $\widetilde{O}(L\lambda\kappa^2/\varepsilon)$ in each round of repetition, along with the practical difficulty of implementing amplitude estimation using early fault-tolerant architectures.

The state-of-the-art algorithm by Costa et al.~\cite{costa2022optimal} requires a circuit depth of $\widetilde{O}(L\lambda\kappa)$, requiring block encoding access to $H$, while incurring an overhead of $\lceil\log_2 L\rceil+6$ ancilla qubits. Despite being resource-efficient, our randomized algorithm requires shorter circuit depth whenever $\lambda\kappa \ll L$.

Compared to the recent Trotter–extrapolation approach of \cite{chakraborty2025quantum}, which achieves a time complexity of 
$
\widetilde{O}(L(\lambda\kappa)^{1+\frac{1}{2k}})
$
using $2k$-th order Suzuki–Trotter formulas and classical extrapolation, our method requires a shorter depth whenever $(\lambda\kappa)^{1-1/(2k)}\ll L$. Since only low-order formulas ($k=1,~2$) are practical in near-term settings, there are broad parameter regimes where our randomized method provides an advantage, while using the same number of ancilla qubits.

Finally, our method provides a polynomial advantage over randomized LCU-based linear systems algorithms \cite{chakraborty2025quantum, wang2024qubit}. Although the circuit depths are comparable, our approach reduces the required number of classical repetitions by a factor of $\kappa^4$, leading to a substantial overall speedup.

\subsection{Ground state property estimation} 
\label{subsec:gspe}
Consider a Hamiltonian $H$ with spectral decomposition $H=\sum_{i} \xi_i \ket{v_i}\bra{v_i}$, where the eigenvalues $\xi_i$ are ordered increasingly. The eigenstate $\ket{v_0}$ corresponding to the smallest eigenvalue $\xi_0$ is called the ground state, and $\xi_0$ is the ground state energy. Preparing the ground state and estimating its energy are among the most fundamental problems in quantum computing and quantum many-body physics. In full generality, these tasks are QMA-hard \cite{kempe2006complexity}, but they become efficiently solvable under additional assumptions. In particular, one typically assumes (i) knowledge of the spectral gap $\Delta$ separating the ground state from the first excited state, and (ii) access to a guess state $\ket{\phi_0}$ that has at least $\gamma$ overlap with the ground state, i.e.\ $|\braket{\phi_0|v_0}|\geq \gamma$.

Furthermore, we assume that any Hamiltonian $H$ of interest can be expressed as a sum of $L$ local terms, namely,  $H=\sum_{k=1}^L \lambda_k H_k$, with $\|H_k\|=1$, and $\lambda=\max\{\sum_{k}|\lambda_k|,1\}$. We now formally state the ground state property estimation problem:

\begin{prob}[Ground state property estimation]
\label{prob:ground state}
    Let $H=\sum_{k=1}^L \lambda_k H_k$ be a Hamiltonian with $\|H_k\|=1$ and $\lambda=\max\{\sum_{k}|\lambda_k|,1\}$. Suppose $\varepsilon\in(0,1/2)$ and that we can efficiently prepare an initial guess state $\left|\phi_0\right\rangle$ such that it has an overlap of at least $\gamma$ with the ground state $\left|v_0\right\rangle$ of $H$, i.e.\ 
    $\left|\braket{\phi_0|v_0} \right| \geq \gamma$, for $\gamma\in (0,1)$. Furthermore, assume that there is a spectral gap $\Delta$ separating the ground state energy $\xi_0$ from the rest of the spectrum. Then, the goal is to compute $\bracket{v_0}{O}{v_0}$ up to additive error $\varepsilon\|O\|$ for any given observable $O$. 
\end{prob}

We use Theorem \ref{thm_gqsp with qdrift and aa}, to obtain the following result. 

\begin{thm}[Randomized Ground state property estimation]
\label{thm:Ground state property estimation using GQSP with qDRIFT}
    Suppose that there is a spectral gap $\Delta$ of $H$ separating the ground state energy $\xi_0$ from the first excited state energy $\xi_1$ such that
    \[
    \xi_0 \leq \mu-\Delta / 2<\mu+\Delta / 2 \leq \xi_1,
    \]
    for some given $\mu$. 
    Then there is a randomized quantum algorithm for the ground state property estimation Problem~\ref{prob:ground state} using two ancilla qubits. 
    The maximum circuit depth is
    \[
    \widetilde{O}\left((\lambda/\Delta\gamma)^2\right),
    \]
    and the total time complexity is
    \[
    \widetilde{O}\left((\lambda/\Delta\gamma\eps)^2\right).
    \]
\end{thm}
\begin{proof}
    We first rescale $H$ to $H'=H/\lambda$ so that $\|H'\|\le 1$. The spectral gap of $H'$ is $\Delta' = \Delta/\lambda$ and the norm-sum parameter $\lambda'=\lambda/\lambda = 1$. The core of the algorithm is to implement an approximate projection onto the ground state. As shown in  \cite{chakraborty2025quantum}, this can be done by constructing a Laurent polynomial $P(z)$ that approximates a shifted sign function, which acts as a filter for the ground state energy. Specifically, there exists a Laurent polynomial $P(z)$ bounded by $1$ on $\mb T$, such that $\|P(e^{\i H'})\ket{\phi_0}-\braket{v_0|\phi_0}\ket{v_0}\|\le \eps/(8\gamma)$, and the degree of $P(z)$ is $$d = O\left(\dfrac{1}{\Delta'}\log \dfrac{1}{\gamma\varepsilon}\right) = \widetilde{O}\left(\dfrac{\lambda}{\Delta}\right).
    $$
    Thus the normalized QSVT state $\ket{\phi}:=P(e^{\i H'})\ket{\phi_0}/\|P(e^{\i H'})\ket{\phi_0}\|$ is an $(\eps/4)$-approximation of the ground state $\ket{v_0}$ up to a global phase factor. This ensures that $|\bra{\phi}O\ket{\phi}-\bra{v_0}O\ket{v_0}|\le \eps\|O\|/2$.

    We can now directly apply Theorem \ref{thm_gqsp with qdrift and aa} to estimate $\bra{\phi}O\ket{\phi}$ with an error at most $\eps\|O\|/2$, which yields a total time complexity of 
    $$\widetilde{O}\bigg(\frac{(\lambda')^2 d^2}{\gamma^2\varepsilon^2}\bigg) = \widetilde{O}\left(\dfrac{\lambda^2}{\Delta^2\gamma^2\varepsilon^2}\right).
    $$ 
    and a maximum circuit depth of $$\widetilde{O}\bigg(\frac{(\lambda')^2 d^2}{\gamma^2}\bigg) = \widetilde{O}\left(\dfrac{\lambda^2}{\Delta^2\gamma^2}\right),
    $$
while using two ancilla qubits. This provides an estimate of $\bra{v_0}O\ket{v_0}$ within an error $\eps\|O\|$.
\end{proof}
We summarize a detailed comparison of our approach with prior methods in Table~\ref{table:ground state problem}. The state-of-the-art algorithm of Lin and Tong \cite{lin2020nearoptimalground} constructs a block encoding of $H$ and then applies QSVT together with fixed-point amplitude amplification to prepare the ground state. This requires $O(\log L)$ ancilla qubits and yields a quantum circuit depth of $\widetilde{O}(L\lambda\Delta^{-1}\gamma^{-1})$. The target expectation value must then be estimated using $\widetilde{O}(1/\eps^{2})$ classical repetitions of this circuit. By contrast, our method operates with only two ancilla qubits, completely avoids block encodings, and achieves a shorter circuit depth whenever $\lambda \ll L\Delta\gamma$, while retaining the same $\widetilde{O}(1/\eps^{2})$ classical repetition cost. 
As noted earlier, quantum amplitude estimation can, in principle, reduce the dependence on the precision parameter from $1/\eps^2$ to $1/\eps$. 
However, this comes at the price of an exponential increase in circuit depth, making it unfavorable for early fault-tolerant devices. Furthermore, implementing amplitude estimation requires block-encoding access to the observable $O$, which introduces additional ancilla overhead and further complicates the procedure.  

\begin{table}[h!]
\centering
\resizebox{1.0\columnwidth}{!}{
\renewcommand{\arraystretch}{1.5} 
\begin{tabular}{c | c c c c} 
 \hline
 Algorithm & Ancilla  & Circuit depth per coherent run & Classical repetitions \\
 \hline\hline
 Lin and Tong~\cite{lin2020nearoptimalground}   & $\lceil \log_2L\rceil+3$ & $\widetilde{O}(L \lambda \Delta^{-1}\gamma^{-1})$ & $\widetilde{O}(\eps^{-2})$ \\
  QETU~\cite{dong2022ground} with Trotter   & 2 & $\widetilde{O}\big(L \big(\lambda \Delta^{-1}\gamma^{-1}\big)^{1+o(1)}\varepsilon^{-o(1)}\big) $ & $\widetilde{O}(\varepsilon^{-2})$ \\ 
 QSVT with Trotter~\cite{chakraborty2025quantum} & 2 & $ \widetilde{O}(L \big(\lambda\Delta^{-1}\gamma^{-1}\big)^{1+o(1)} ) $ & $\widetilde{O}(\varepsilon^{-2})$\\ \hline
   QETU~\cite{dong2022ground} with qDRIFT   & 2 & $\widetilde{O}\left( \lambda^2 \Delta^{-2}\gamma^{-2}\varepsilon^{-1}\right) $ & $\widetilde{O}(\varepsilon^{-2})$ \\ 
    Other randomized methods \cite{wang2024qubit, chakraborty2024implementing} & $1$ & $\widetilde{O}(\lambda^2 \Delta^{-2})$ & $\widetilde{O}\left(\eps^{-2}\gamma^{-4}\right)$\\
 This work (Thm.~\ref{thm:Ground state property estimation using GQSP with qDRIFT}) & 2 & $\widetilde{O}\big((\lambda\Delta^{-1}\gamma^{-1})^2\big)$& $\widetilde{O}(\eps^{-2})$ 
\\\hline 
\end{tabular}
}
\caption{Comparison of the complexities of different algorithms for ground state property estimation. Consider any Hamiltonian $H=\sum_{k=1}^{L}\lambda_k H_k$, with $\|H_k\|=1$, and spectral gap $\Delta$. Here, $\lambda = \max\{\sum_{k}|\lambda_k|,1\}$. Then, given an initial ``guess state'' with overlap at least $\gamma$ with the ground state $\ket{v_0}$ of $H$, and an observable $O$, these algorithms estimate $\braket{v_0|O|v_0}$ with additive accuracy $\eps\|O\|$. For each algorithm, we compare the total number of ancilla qubits required, the circuit depth per coherent run, and the total number of classical repetitions needed.}
\label{table:ground state problem}
\end{table}

Several recent works have proposed alternatives that avoid block encodings by assuming oracle access to the time-evolution operator $U = e^{\i H}$, with complexity measured in terms of the number of queries to $U$ \cite{lin2022heisenberg, dong2022ground, wan2022randomized, Wang2023quantumalgorithm, zhang2022computingground}. These approaches typically require only a few ancilla qubits, making them attractive in principle. However, to implement them as fully end-to-end quantum algorithms without introducing additional ancilla overhead, one must approximate $U$ using either higher-order Trotterization or qDRIFT. Both substitutions substantially increase the circuit depth, resulting in an exponentially worse dependence on the precision parameter compared to our randomized algorithm.

In Ref.~\cite{chakraborty2025quantum}, a complementary end-to-end algorithm for ground-state property estimation was developed. Instead of qDRIFT, the authors employ $2k$-th order Suzuki–Trotter formulas to implement the Hamiltonian evolutions in $W$, together with fixed-point amplitude amplification and classical extrapolation to estimate the target expectation value. This approach uses three ancilla qubits and requires $\widetilde{O}(\eps^{-2})$ independent runs of a quantum circuit with depth
$
\widetilde{O}(L({\lambda}/{\Delta\gamma})^{1+\frac{1}{2k}}).
$
As discussed earlier, only low-order Trotter formulas are practical in the early fault-tolerant setting. For any order $k$, our randomized algorithm achieves strictly shorter circuit depth whenever 
$
({\lambda}/{\Delta\gamma})^{1-\frac{1}{2k}}\ll L.
$
This condition is often satisfied in realistic regimes, particularly for low Trotter orders, allowing our method to provide a clear advantage while also requiring fewer ancilla qubits.

Finally, randomized LCU-based methods for ground-state property estimation \cite{chakraborty2025quantum, wang2024qubit} require only a single ancilla qubit, but incur a total cost of $\widetilde{O}(\eps^{-2}\gamma^{-4})$ runs of a quantum circuit of depth $\widetilde{O}(\lambda^2/\Delta^2)$ to estimate the target expectation value. This leads to a suboptimal dependence on $1/\gamma$, and, moreover, these methods cannot exploit amplitude amplification to reduce the number of classical repetitions. In contrast, our algorithm achieves matching circuit depth and total complexity when amplitude amplification is not used. More importantly, Theorem~\ref{thm:Ground state property estimation using GQSP with qDRIFT} shows that fixed-point amplitude amplification can be seamlessly integrated into $W$, yielding a quadratically improved dependence on $1/\gamma$ in the overall complexity.

Overall, our randomized algorithms achieve shorter circuit depth across broad parameter regimes, while remaining fully independent of the number of terms in $H$. They completely avoid block encodings, require only a small constant number of ancilla qubits, and reduce classical repetition costs relative to randomized LCU methods. Taken altogether, these features make our approach especially well-suited for the early fault-tolerant era, where circuit depth and ancilla resources are the dominant bottlenecks.

\section{Numerical Benchmarking}
\label{sec:numerics}
We evaluate the practical performance of our randomized QSVT algorithms on the ground-state property estimation task introduced in Sec.~\ref{subsec:gspe}. In particular, we compare the asymptotic circuit depth, measured in terms of single- and two-qubit gates, with that of prior approaches. Our benchmarks include the algorithm of Dong, Lin, and Tong~\cite{dong2022ground}, which is suited for early fault-tolerant quantum computers, but assumes black box access to the time evolution operator $e^{\i H}$. To enable a fair comparison, we explicitly implement this oracle using standard Hamiltonian simulation techniques, namely the first- and fourth-order Suzuki–Trotter product formulas and the qDRIFT method. In each case, we report the end-to-end circuit depth required to achieve the target accuracy in ground state property estimation. 

In addition, we compare against recent QSVT algorithms based on Suzuki-Trotter decompositions~\cite{chakraborty2025quantum}. As noted previously, only the second- and fourth-order formulas are considered in practice, since the pre-factor associated with higher-order Trotter methods grows exponentially and quickly outweighs their asymptotic advantages. Finally, we benchmark our algorithm against prior randomized approaches~\cite{chakraborty2024implementing, wang2024qubit}. As established in Sec.~\ref{subsec:gspe}, our randomized QSVT enjoys a quadratic improvement in the overlap parameter $\gamma$, leading to a provable reduction in overall gate complexity. We demonstrate this advantage explicitly by comparing total gate counts across the spin and molecular Hamiltonians studied in the following subsections.

Having outlined the algorithms used for comparison, we now turn to the systems we consider. In Sec.~\ref{subsec:q-chem-system}, we examine quantum chemistry Hamiltonians, in particular, those of elementary molecules such as propane, ethane, and carbon dioxide. Following this, in Sec.~\ref{subsec:spin-hamiltonians}, we also explore quantum spin models with long-range interactions, where the number of terms $L$ grows quadratically with system size, while $\lambda$ grows linearly.

Both these systems are particularly well-suited to randomized approaches: although the number of Pauli terms $L$ is extremely large, the total coefficient sum $\lambda$ remains moderate, making the quadratic dependence on 
$\lambda$ more favorable than the linear dependence on $L\lambda$ that arises in Trotter-based methods and standard QSVT.

\subsection{Quantum chemistry Hamiltonians}
\label{subsec:q-chem-system}
Electronic structure Hamiltonians provide a natural and practically relevant testing ground for ground-state property estimation algorithms. These systems have been extensively studied in the context of Hamiltonian simulation \cite{whitfield2011simulation, babbush2018low, campbell2019random, motta2021low, gunther2025phase}. The Hamiltonian of a molecule in second quantization takes the form
\begin{equation}
    H=\sum_{p q} h_{p q} a_p^{\dagger} a_q+\frac{1}{2} \sum_{p q r s} h_{p q r s} a_p^{\dagger} a_q^{\dagger} a_r a_s, 
\end{equation}
where $a_p^\dagger$ and $a_q$ are fermionic creation and annihilation operators, and the coefficients $h_{pq}, h_{pqrs}$ are determined by the one- and two-electron integrals of the molecular system. To simulate such Hamiltonians on a quantum computer, we map the fermionic operators to qubits using the Jordan–Wigner transformation, yielding a weighted sum of Pauli operators. 
\begin{table}[h!]
\centering
\resizebox{0.9\columnwidth}{!}{
\renewcommand{\arraystretch}{1.5} 
\begin{tabular}{| c | c | c | c |} 
 \hline
Molecule & Number of qubits $(n)$ & Total number of terms $(L)$  & 1-norm of the coefficients $(\lambda)$ \\
 \hline\hline
Propane (STO-3G)  & $46$ & $390,441$ & $435.98$ \\ \hline
Carbon dioxide (6-31G) & $54$  & $182,953$ & $679.04$ \\ \hline 
Ethane (6-31G) & $60$ & $301,718$ & $ 711.67$\\ \hline 
\end{tabular}
}
\caption{Number of Pauli terms $L$ and 1-norm $\lambda$ (in Hartree) of the truncated Hamiltonians for the three molecules.}
\label{table:qchem-Ham}
\end{table}

In our numerical experiments, we study three representative molecules at their ground-state equilibrium geometries: Propane (\ce{C3H8}), carbon dioxide (\ce{CO2}), and ethane (\ce{C2H_6}). For propane, we construct the molecular orbitals using the minimal STO-3G basis set, while for carbon dioxide and ethane we adopt the larger 6-31G basis. These choices specify the set of Gaussian functions used to represent the atomic orbitals, and therefore determine the size and accuracy of the one- and two-electron integrals entering the electronic structure Hamiltonian.

To enable simulations within reasonable resources and to provide a fair comparison across methods, we apply a systematic truncation procedure to the Pauli decompositions of the Hamiltonians. Specifically, the terms are sorted by the magnitude of their coefficients, and those with the smallest magnitudes are discarded, subject to the constraint that the sum of the absolute values of the removed coefficients does not exceed $\eps/(\gamma \Delta)$. This threshold ensures that the truncation error is controlled and enables a fair comparison with Trotter-based methods. After truncation, we subtract the identity contribution from each Hamiltonian, since it merely shifts all eigenvalues by a constant and does not affect ground-state properties. The resulting truncated Hamiltonians have the number of Pauli terms $L$, as well as the total 1-norm of the coefficients $\lambda$, reported in Table~\ref{table:qchem-Ham}. Notably, all three systems exhibit extremely large $L$ (hundreds of thousands of terms) but only moderate $\lambda$ (a few hundred Hartree). This separation between $L$ and $\lambda$, makes them particularly well-suited for randomized approaches since our algorithm uses qDRIFT, which scales with $\lambda^2$ (instead of $L\lambda$ in Trotter methods or standard QSVT). Moreover, a novel use of extrapolation allows us to obtain exponentially improved circuit depth as compared to vanilla qDRIFT \cite{campbell2019random}.
\begin{figure}[htbp]
    \centering
    
    \includegraphics[width=\linewidth]{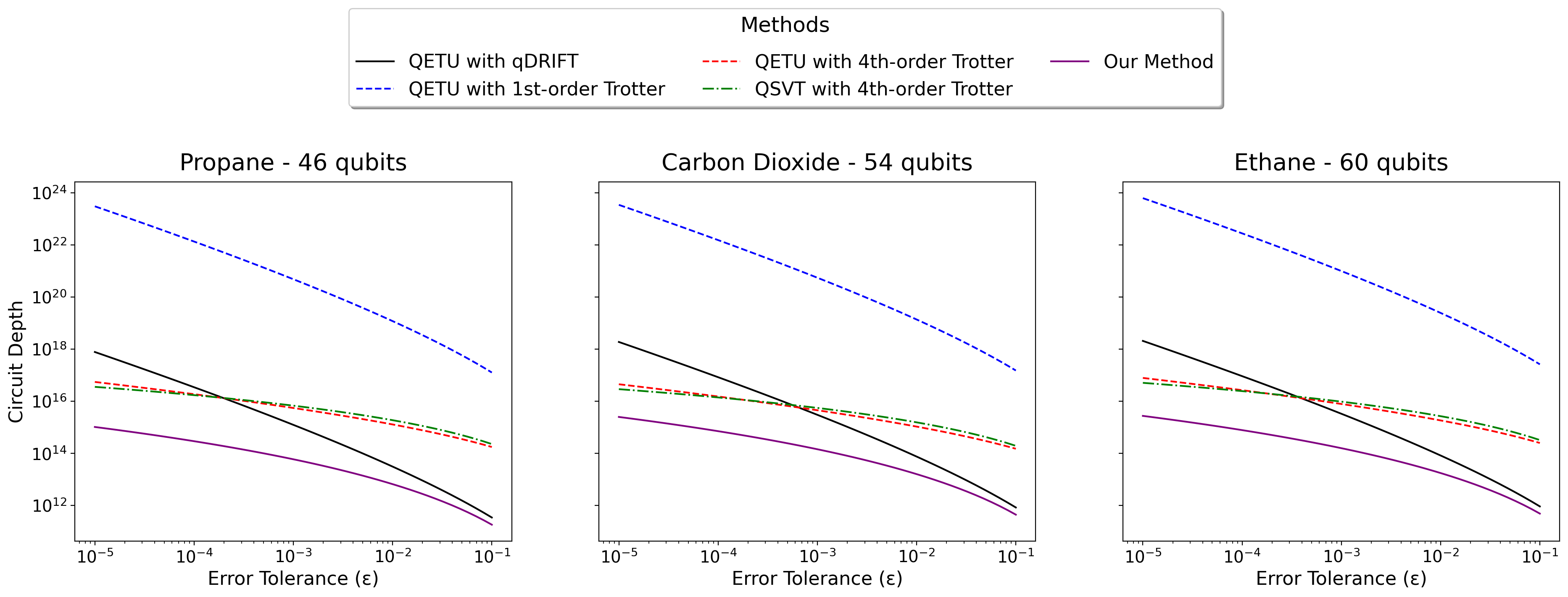}
     \caption{Comparison of the circuit depth of the different quantum algorithms for ground state property estimation when applied to the electronic structure Hamiltonians of propane (in STO-3G basis), carbon dioxide (in 6-31G basis), and ethane (in 6-31G basis). We assume that the initial ``guess" state has an overlap of $\gamma=0.1$ with the ground state of the Hamiltonian. We plot the circuit depth per coherent run as a function of the error tolerance $\varepsilon$. The circuit depth of randomized QSVT is significantly shorter than the other approaches considered here, for all three molecules.}
    \label{fig:eps_chem_error}
\end{figure}
In order to compare with Trotter-based methods, note that these Hamiltonians do not have structure and hence commutator bounds are unavailable. We therefore set the commutator prefactor $\lambda_{\rm comm}'=\lambda$ for QETU with Trotter, and the commutator prefactor $\lambda_{\rm comm}=\lambda$ for QSVT with Trotter. Here, $\lambda_{\rm comm}'$ denotes $C_{\rm Trotter}^{1/(2k+1)}$ defined in \cite[Appendix~C]{dong2022ground} for $2k$-order Trotterization, so that the complexity of QETU with Trotter scales as $O((\lambda_{\rm comm}')^{1+1/2k})$.  In addition, we fix the spectral gap to $\Delta=0.25$ Hartree, which serves as a lower bound for the vertical excitation energies of these molecules, following standard references~\cite{nakatsuji1983cluster, richartz1978calculation, johnson1979electron}. Finally, all Hamiltonians were generated using OpenFermion~\cite{mcclean2020openfermion} and the PySCF package~\cite{sun2015libcint, sun2018pyscf, sun2020recent}, which provide one- and two-electron integral evaluations and fermion-to-qubit mappings.
\begin{figure}[htbp]
    \centering
    
    \includegraphics[width=\linewidth]{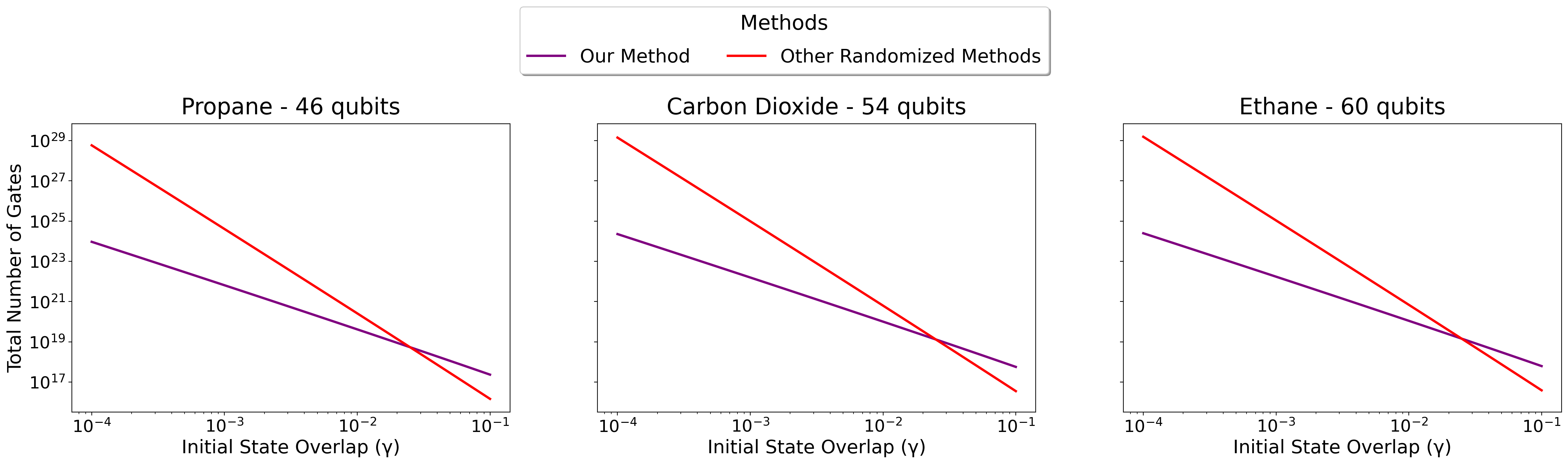}
     \caption{Comparison of the overall gate complexity (total number of single and two-qubit gates needed) of our method with other randomized methods \cite{chakraborty2024implementing, wang2024qubit} for ground state property estimation when applied to the electronic structure Hamiltonians of propane (in STO-3G basis), carbon dioxide (in 6-31G basis), and ethane (in 6-31G basis). We fix the error tolerance $\eps=0.01$, and vary $\gamma$, the overlap of the initial ``guess" state with the ground state of the Hamiltonian. Our method requires substantially fewer gates than prior randomized methods for all three molecules.}
    \label{fig:eps_chem_overlap}
\end{figure}
Figure~\ref{fig:eps_chem_error} compares the circuit depth per coherent run of randomized QSVT with that of competing methods. With the overlap of the initial ``guess" state with the ground states fixed at $\gamma=0.1$, we vary the error tolerance $\eps$. This choice of $\gamma$ is motivated by the fact that for small molecules such as the ones we consider here, there are several classical and quantum approaches to prepare a state with a constant overlap with the target ground state \cite{wecker2014gate,mcardle2020quantum, reiher2017elucidating, babbush2015chemical}. 

Across all three molecular Hamiltonians, our approach consistently achieves significantly lower depths than the QETU algorithm of Dong, Lin, and Tong~\cite{dong2022ground}, whether the required time-evolution oracle is implemented by Trotterization or by qDRIFT. We also outperform the QSVT-with-Trotterization algorithm of Chakraborty et al.~\cite{chakraborty2025quantum}, even when using fourth-order Suzuki–Trotter formulas. 
The quantitative advantage is significant: for $\eps=10^{-5}$, randomized QSVT yields depths shorter by roughly nine orders of magnitude compared to QETU with first-order Trotter, by about three orders of magnitude compared to QETU with qDRIFT, and by about two orders of magnitude even relative to fourth-order Trotterized QSVT. The advantage is similar with respect to QSVT using Trotterization. 

The circuit depth of the ground-state property estimation algorithm based on randomized LCU exhibits scaling behavior similar to the variant of our method that does not employ fixed-point amplitude amplification. However, as discussed in Sec.~\ref{subsec:gspe}, once combined with quantum amplitude amplification, our approach achieves a quadratically better dependence on the overlap parameter $\gamma$ in terms of the overall complexity. This advantage is evident in Fig.~\ref{fig:eps_chem_overlap}, which reports the total gate complexity (i.e., the number of single- and two-qubit gates aggregated over all runs) for the three molecular Hamiltonians. Here we fix the precision to $\eps=0.01$, and vary the overlap $\gamma$ of the initial guess state. As $\gamma$ decreases, the gate complexity of randomized LCU~\cite{chakraborty2024implementing, wang2024qubit} increases sharply, whereas our algorithm scales much more favorably. For example, at $\gamma=10^{-4}$, our method requires fewer gates by roughly five orders of magnitude.

Although we do not explicitly plot the circuit depth of ground-state property estimation using standard block-encoding-based QSVT~\cite{lin2020nearoptimalground}, our method achieves strictly shorter depths for all three molecular Hamiltonians. The key distinction lies in the scaling: standard QSVT requires circuit depth proportional to $L\lambda$, whereas our randomized approach depends only on $\lambda^2$, which is much smaller in this case as $L\gg \lambda$. This comparison already favors our method before accounting for the additional overhead of decomposing the multi-qubit controlled operations inherent to block encodings into elementary single- and two-qubit gates. Furthermore, block-encoding QSVT necessitates access to an explicit block encoding of $H$, which itself is resource-intensive: it involves intricate controlled operations over $O(\log L)$ ancilla qubits. For instance, for propane, this translates to roughly 21 ancilla qubits and gates controlled across each of them, an overhead entirely avoided by our approach.  

Overall, these results demonstrate that randomized QSVT provides a significant practical advantage for ground-state property estimation on realistic molecular Hamiltonians. By eliminating the explicit dependence on $L$, avoiding costly block encodings, obtaining improved circuit depths using classical techniques such as extrapolation, and requiring only a constant number of ancilla qubits, our approach substantially reduces circuit depth and gate overhead. These features make randomized QSVT especially well-suited for early fault-tolerant quantum devices, where qubit counts and coherence times remain limited.

\subsection{Transverse-field Ising model}
\label{subsec:spin-hamiltonians}
As a complementary benchmark to molecular Hamiltonians, we also evaluate the performance of ground state property estimation by randomized QSVT on spin-chain models that are central to quantum many-body physics and quantum simulation experiments. In particular, we study two one-dimensional systems: (i) the transverse-field Ising chain with algebraically decaying ZZ interactions and (ii) a hybrid extension that combines short-range nearest-neighbor XX couplings with additional long-range ZZ terms. 

The first model is a classic setting for studying quantum phase transitions, correlation spread, and anomalous dynamical phenomena in long-range interacting systems \cite{hauke2013spread, defenu2021longrange}; the second is a natural yet less-explored extension that contains both nearest-neighbor as well as long-range interactions. Because both types of interactions can be engineered or approximated in several quantum technological platforms, such as trapped ions \cite{blatt2012quantum, monroe2021programmable}, and Rydberg atom arrays \cite{zeiher2016many, browaeys2020many}, these Hamiltonians serve as meaningful testbeds for benchmarking the performance of ground state property estimation by randomized QSVT, with other methods.

\subsubsection{Long-range Transverse Ising field Ising chain}
\label{subsubsec:TFIM-1}
We consider the $n$-qubit one-dimensional Hamiltonian
\begin{equation}
\label{eq:Ham-TFIM-long-range}
H = -J \sum_{1 \le i < j \le n} \frac{Z_i Z_j}{|i-j|^\alpha} - h \sum_{i=1}^n X_i,  
\end{equation}
with $J, h > 0$ and $\alpha>1$. This is an extension of the paradigmatic nearest-neighbor transverse-field Ising model \cite{sachdev2011quantum}, which is of deep theoretical interest: it interpolates between mean-field behavior ($\alpha\rightarrow 0$), and short-range neighbor ($\alpha\rightarrow\infty$). Their relevance extends beyond theory, as long-range Ising interactions can be engineered in several experimental platforms. Because trapped ions naturally mediate spin–spin interactions via phonon modes, with effective coupling exponents tunable by laser detunings, implementing a long-range Ising Hamiltonian is relatively direct in that platform, with $\alpha$ varying between $0$
 and $3$ \cite{blatt2012quantum, zhang2017observation, monroe2021programmable}. Similarly, Rydberg atom arrays exhibit van der Waals $(\alpha=3)$ or dipolar $(\alpha=6)$ interactions that yield long-range Ising models in one or two dimensions \cite{zeiher2016many, browaeys2020many}. Dipolar quantum gases of magnetic atoms and ultracold polar molecules confined in optical lattices also realize spin Hamiltonians with $\alpha=3$ through magnetic or electric dipole–dipole couplings \cite{yan2013observation, depaz2013nonequilibrium}. Owing to the relevance of such systems, quantum simulation algorithms for Hamiltonians with power law interactions have been extensively studied \cite{tran2019locality, childs2021theory, haah2023quantum}.

 The Hamiltonian $H$ in Eq.~\eqref{eq:Ham-TFIM-long-range} is already in the desired form, i.e.\ $H=\sum_{k=1}^{L}\lambda_k P_k$, with the number of Pauli terms
 \begin{align*}
    L = n+ \binom{n}{2} = \frac{n^2+n}{2} = O(n^2).
\end{align*}
The one-norm coefficient of the coefficients of $H$, 
\begin{align*}
    \lambda = hn+J\sum_{r = 1}^{n-1} \frac{n-r}{r^{\alpha}} \le (h+J~\zeta(\alpha))n = \Theta(n),
\end{align*}
whenever the interaction strength $J,~ h$ are constants. Here, $\zeta(s)=\sum_{r=1}^{\infty} r^{-s}$ is the Riemann zeta function, which is a constant when $\alpha>1$. Thus the Hamiltonian satisfies $\lambda \ll L$, a feature that strongly favors randomized methods.

To benchmark our ground-state property estimation algorithm, it is necessary to assume a nontrivial bound on the spectral gap $\Delta$. Obtaining an exact analytical estimate of $\Delta$ is notoriously difficult, but in the Appendix \ref{sec-app:tfim}, we show evidence that
$$
\Delta \sim 2\left|h-J~\zeta(\alpha)\right|,
$$
which remains constant for a broad range of parameter values. This scaling is further supported by numerical calculations up to $n=26$ qubits (See Fig.~\ref{fig:spectral-gap} in Appendix \ref{sec-app:tfim}).

For small spin systems, classical methods can be used to efficiently prepare an initial “guess” state with nontrivial overlap with the ground state, using tools such as tensor networks (especially matrix product states)~\cite{white1992density,schollwock2005density,verstraete2008matrix} or variational algorithms \cite{huggins2019towards,dborin2022matrix,rudolph2023synergistic}. As a result, we set $\gamma=0.1$ in our experiments. In our setting, we are in fact primarily concerned with how the complexity of our randomized QSVT algorithm scales with the system size $n$. Assuming the availability of a “guess” state with constant (albeit possibly small) overlap with the true ground state, the asymptotic circuit depth of randomized QSVT is therefore $\widetilde{O}(n^2)$. By contrast, the algorithm of Lin and Tong \cite{lin2020nearoptimalground}, which employs block-encoding-based QSVT, scales as $\widetilde{O}(n^3)$, while additionally requiring $O(\log n)$ ancilla qubits and multi-qubit controlled operations. For comparisons against Trotter-based algorithms, specifically, QSVT with Trotterization~\cite{chakraborty2025quantum} and QETU with Trotterization~\cite{dong2022ground}, we must also estimate the commutator pre-factor, $\lambda_{\rm comm}$ and $\lambda_{\rm comm}'$. 

In Appendix~\ref{sec-app:tfim} (Lemma~\ref{lem:commutator-prefactor-scaling}), we establish that
$$
\lambda_{\rm comm}=O\left(n^{1/(2k+1)}\right),\quad \lambda_{\rm comm}'=O\left(n^{1/(2k+1)}\right)
$$
for the $2k$-th order Suzuki–Trotter formula. Our proof provides both upper and lower bounds through a fine-grained analysis that explicitly tracks constant factors. This level of precision is essential for the subsequent numerical benchmarking, where we compare Trotter-based approaches with randomized QSVT in regimes where the system size $n$ is only moderately large (on the order of a few thousand qubits). Importantly, although the asymptotic scaling is similar to prior work on spin Hamiltonians with power law decay \cite{childs2021theory, tran2019locality, haah2023quantum}, our derivation does not follow from existing techniques in these references. To the best of our knowledge, ours is a new method for obtaining tight commutator pre-factor scalings for spin Hamiltonians with a combination of short and long-range interactions.

Consequently, the circuit depth of the algorithm in~\cite{chakraborty2025quantum} scales as $\widetilde{O}(n^{2+1/(2k)})$, while that of~\cite{dong2022ground} scales as $O(n^2 (n/\eps)^{1/(2k)})$. In practice, only low-order Trotter methods $(k=1,2)$ are feasible, making the asymptotic advantage of randomized QSVT especially significant. Finally, we validate these theoretical findings numerically. As shown in Fig.~\ref{fig:combined-n} (left), for representative parameter values of $H$, randomized QSVT consistently achieves lower circuit depth than all prior approaches, thereby demonstrating concrete practical advantages in realistic settings.
\begin{figure}[htbp]
    \centering
    \includegraphics[width=\linewidth]{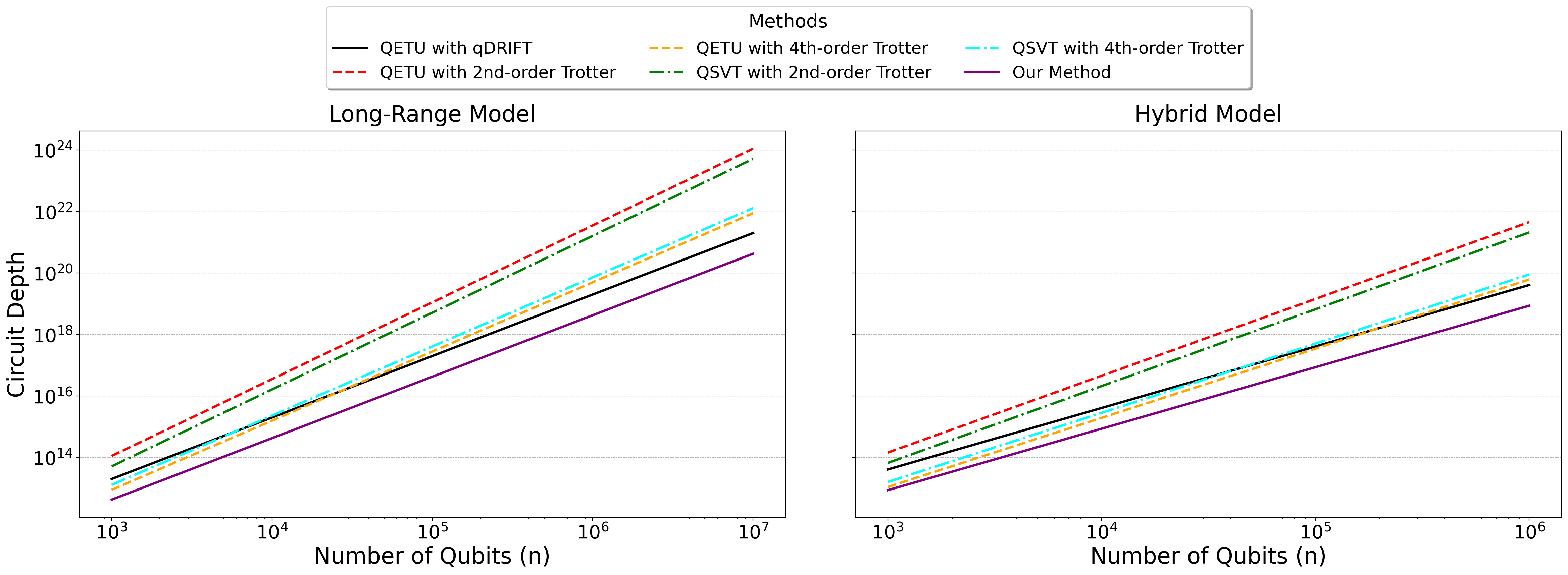}
     \caption{Comparison of the circuit depth of different quantum algorithms for ground state property estimation as a function of the number of qubits $n$ for the two Transverse Field Ising Model Hamiltonians considered in this work. (To the left) Transverse-field Ising model with long-range interactions, with the Hamiltonian parameters set to $h=3$, $J=1$, $\alpha=3$, $\eps = 0.01$, $\gamma=0.1$ (See Eq.~\eqref{eq:Ham-TFIM-long-range}). The spectral gap is fixed to $\Delta=3$. (To the right) Transverse-field Ising model with hybrid interactions, with the Hamiltonian parameters $h=3$, $J=1$, $g=0.1$, $\alpha=3$, $\eps = 0.01$, $\gamma=0.1$ (See Eq.~\eqref{eq:Ham-TFIM-hybrid}). In both cases, the circuit depth of ground state property estimation by randomized QSVT is shorter as compared to the other methods.}
    \label{fig:combined-n}
\end{figure}

\subsubsection{Transverse field Ising model with hybrid XX-ZZ interactions}
\label{subsubsec:TFIM-2}
The second model we consider is a one-dimensional hybrid spin chain that combines short- and long-range Ising interactions. Consider $n$ qubits on a one-dimensional chain with periodic boundary conditions. For $\alpha>1$ and couplings $h>0$, $J>0$, $g\ge 0$, we define the following $n$-qubit Hamiltonian
\begin{equation}
\label{eq:Ham-TFIM-hybrid}
H=-h\sum_{i=1}^{n} Z_i -J\sum_{i=1}^{n} X_i X_{i+1}
-\dfrac{g}{n}\sum_{1\leq i<j\leq n}\dfrac{1}{|i-j|^{\alpha}}~Z_i Z_j.
\end{equation}
This ``hybrid'' Hamiltonian merges ingredients that are each well established in quantum many-body physics: the nearest-neighbor transverse-field Ising chain, which is exactly solvable and serves as a paradigmatic model of quantum phase transitions~\cite{sachdev2011quantum}, and the long-range Ising model with algebraically decaying couplings, which exhibits anomalous criticality and dynamics~\cite{defenu2021longrange}. While the exact combination of nearest neighbor XX interactions, transverse field, and long-range ZZ terms has not been studied as extensively as these limiting cases, it represents a natural and experimentally motivated extension. Indeed, both short and long-range interactions can be engineered in a variety of technological platforms, such as ion traps and Rydberg atoms. Even theoretically, a (normalized) long-range interaction model has been studied \cite{kastner2011diverging}.

Beyond their physical relevance, our motivation for studying this hybrid model is that it inherits the desirable structural features of the long-range transverse-field Ising Hamiltonian discussed in Sec.~\ref{subsubsec:TFIM-1}, while at the same time permitting a rigorous lower bound on the spectral gap. In particular, the number of Pauli terms scales quadratically with system size, $L=\Theta(n^2)$, and  
\begin{align}
\lambda
&= h n + J n+ \frac{g}{n}\sum_{r=1}^{n-1}\frac{n-r}{r^{\alpha}}
\leq (h+J)n + g\zeta(\alpha)= \Theta(n),
\label{eq:lambda-upper-bound-hybrid}
\end{align}
for constant interaction strengths $J,~g,~h$. Thus, as in the long-range TFIM, the hybrid Hamiltonian lies in the favorable regime where $\lambda \ll L$, making it suitable for randomized QSVT.

For the spectral gap, we exploit the fact that the Hamiltonian can be decomposed as $H_0+V$, where
$$
H_0=-h\sum_{i=1}^{n} Z_i -J\sum_{i=1}^{n} X_i X_{i+1},
$$
is the standard one-dimensional transverse-field Ising chain, and $V$ contains long-range interactions. 
The spectral gap of $H_0$ is well-known: diagonalizing via a Jordan–Wigner transformation and minimizing the single-quasiparticle dispersion relation yields $\Delta_0=2|h-J|$. The second part, $V$ satisfies
$$
\|V\|
\le\dfrac{g}{n}\sum_{i<j}\dfrac{1}{|i-j|^{\alpha}}
=\dfrac{g}{n}\sum_{r=1}^{n-1}\dfrac{n-r}{r^{\alpha}}\leq g~\zeta(\alpha).
$$
By the min--max (Weyl) inequalities for Hermitian matrices,
$$
\Delta\geq \Delta_0-2\|V\|=2|h-J|-2g~\zeta(\alpha).
$$
In particular, if 
$$
g< \dfrac{|h-J|}{2~\zeta(\alpha)},
$$
we have $\Delta \geq |h-J|$, which is constant whenever $|h-J|$ is. Thus, provided one can prepare a “guess” state with constant overlap with the ground state, the overall circuit depth of our randomized QSVT algorithm scales as $\widetilde{O}(n^2)$, identical to the scaling in the long-range Ising model analyzed in the previous subsection. The asymptotic advantage over standard block-encoding-based QSVT therefore carries over unchanged to this hybrid setting.  

Finally, in Lemma~\ref{lem:commutator-prefactor-scaling} of Appendix~\ref{sec-app:tfim}, we establish that the commutator prefactor scales as $\lambda_{\rm comm}=\Theta(n^{1/(2k+1)})$, $\lambda_{\rm comm}'=\Theta(n^{1/(2k+1)})$ for a $2k$-th order Trotter decomposition. Consequently, even for the hybrid Hamiltonian considered here, randomized QSVT achieves a concrete asymptotic advantage over the Trotterized QSVT algorithm of~\cite{chakraborty2025quantum} and the QETU-based approach of~\cite{dong2021efficient}.

\begin{table}[h!]
\centering
\resizebox{0.6\columnwidth}{!}{
\renewcommand{\arraystretch}{1.5} 
\begin{tabular}{| c | c | c |} 
 \hline
Algorithm & Ancilla qubits & Circuit depth \\
 \hline\hline
Standard QSVT \cite{lin2020nearoptimalground} & $O(\log n)$ & $\widetilde{O}(n^3)$\\ \hline
QETU with $2k$-order Trotter \cite{dong2022ground} & $2$  & $\widetilde{O}\left(n^{2}(n/\eps)^{1/(2k)}\right)$ \\ \hline 
QETU with qDRIFT \cite{dong2022ground} & $2$  & $\widetilde{O}(n^2/\eps)$ \\ \hline
QSVT with $2k$-order Trotter \cite{chakraborty2025quantum} & $2$ & $\widetilde{O}(n^{2+1/(2k)})$\\ \hline 
This work (Theorem \ref{thm:Ground state property estimation using GQSP with qDRIFT}) & $2$ & $\widetilde{O}(n^2)$\\ \hline 
\end{tabular}
}
\caption{The circuit depth of different algorithms for estimating the ground state properties of $H$ in Eq.~\eqref{eq:Ham-TFIM-hybrid}, for parameter regimes where the spectral gap is constant. We assume that the initial guess state has a constant overlap with the ground state.}
\label{table:qspin-Ham}
\end{table}

The asymptotic circuit depth of randomized QSVT for the hybrid Hamiltonian is summarized in Table~\ref{table:qspin-Ham}, and the corresponding numerical benchmarking is presented in Fig.~\ref{fig:combined-n} (right). For the simulations, we fix the parameters to $h=3,~J=1,~g=0.1,~\alpha=3,~\eps=0.01$ and $\gamma=0.1$. In this parameter regime, the spectral gap of $H$ remains constant. 
The results clearly show that randomized QSVT achieves concrete reductions in circuit depth compared to prior approaches, and that these advantages persist across a practically relevant range of system sizes.

\section{Discussion}
\label{sec:conclusion}

In this work, we introduced the first randomized quantum algorithms for quantum singular value transformation (QSVT). Our methods simultaneously address two central bottlenecks of standard QSVT: the reliance on costly block encodings and the explicit linear dependence on the number of Hamiltonian terms $L$. Concretely, given a Hamiltonian $H=\sum_{k=1}^{L}\lambda_k H_k$ written as a linear combination of unitary (or normalized Hermitian) operators, and an initial state $\ket{\psi_0}$, our algorithms estimate expectation values with respect to $f(H)\ket{\psi_0}$ for any function $f$ that admits a bounded polynomial approximation of degree $d$. Both algorithms use a single ancilla qubit and achieve circuit depth scaling as $\widetilde{O}(\lambda^2d^2)$, independent of $L$. These features make them particularly well suited to early fault-tolerant devices, where ancilla qubits are scarce, multi-qubit controls remain expensive, and circuit depth is the most critical resource.

Our first algorithm can be viewed as a direct randomization of standard QSVT~\cite{gilyen2019quantum}, preserving its structure but also inheriting inefficiencies of prior randomized methods~\cite{chakraborty2025quantum, wang2024qubit}. An interesting open question is whether one can design a randomized variant of QSVT that still permits amplitude amplification, thereby improving asymptotic performance. 

The second algorithm embeds qDRIFT into generalized quantum signal processing (GQSP) and uses Richardson extrapolation to achieve improved precision scaling. This approach is more flexible, exhibits superior asymptotic performance, and yields strong advantages in concrete applications. For quantum linear systems and ground-state property estimation, it achieves substantial polynomial improvements over prior randomized LCU-based methods~\cite{chakraborty2024implementing, wang2024qubit}. Beyond these specific tasks, our techniques open the door to new randomized algorithms for the wide range of problems accessible via QSVT. One shortcoming of randomized algorithms~\cite{campbell2019random, nakaji2024qswift, faehrmann2022randomizingmulti, wan2022randomized, chakraborty2025quantum, wang2024qubit} is their incompatibility with existing variants of quantum amplitude estimation~\cite{grinko2021iterative, rall2023amplitude}, making it impossible to achieve Heisenberg scaling. Reconciling this gap remains an interesting direction for future work.

To complement our theoretical analysis, we numerically benchmarked randomized QSVT on ground-state property estimation for Hamiltonians from both quantum chemistry and condensed matter physics. For electronic-structure Hamiltonians of propane, carbon dioxide, and ethane, as well as Ising Hamiltonians with long-range and hybrid interactions, our algorithm consistently outperforms standard block-encoding-based QSVT, early fault-tolerant methods, and prior randomized algorithms, often by several orders of magnitude. These simulations confirm that the asymptotic advantages of randomized QSVT translate directly into meaningful reductions in resource requirements across realistic parameter regimes. An important next step will be to refine these benchmarks into explicit hardware-level resource estimates, including optimized CNOT and T gate counts once our algorithms are compiled into fault-tolerant circuits. This would facilitate a more direct comparison with different experimental platforms and provide a clearer assessment of the practicality of randomized QSVT for near-term applications in quantum chemistry and many-body physics.

Alongside these algorithmic contributions, we established a matching lower bound in the Pauli sample access model. Specifically, we proved that any randomized algorithm for Hamiltonian simulation requires circuit depth $\Omega(t^2)$. This shows that generic randomized algorithms implementing polynomial transformations must necessarily incur a quadratic dependence on the polynomial degree. This lower bound establishes the essential optimality of not only our algorithms but also other randomized Hamiltonian simulation methods~\cite{campbell2019random, nakaji2024qswift, wan2022randomized} and randomized LCU approaches~\cite{chakraborty2025quantum, wang2024qubit}. A key open question in this direction is whether stronger lower bounds can be obtained for broader classes of functions. One might also ask whether, for any continuous function $f(x)$, randomized quantum algorithms require complexity at least $\Omega(\widetilde{\deg}(f)^2)$, where $\widetilde{\deg}(f)$ is the approximate degree of $f(x)$. Such a result would establish a fundamental complexity-theoretic limitation for randomized algorithms, in contrast to the linear lower bound $\Omega(\widetilde{\deg}(f))$ known for standard QSVT~\cite{montanaro2024quantum}.

Another promising avenue is to explore structured Hamiltonians. Our bounds apply in the worst case, and it remains unclear whether properties such as commutativity, sparsity, or geometric locality can be leveraged to circumvent the quadratic barrier. One intriguing possibility is the development of hybrid methods that interpolate between block encodings and randomized sampling, potentially combining the best of both frameworks. Finally, motivated by recent progress showing that Trotterization and randomization can be fruitfully combined for Hamiltonian simulation~\cite{childs2019faster}, it would be natural to investigate analogous hybrid strategies within QSVT.

\section*{Acknowledgments}
CS and YZ are supported by the National Key Research and Development Project of China under Grant No. 2020YFA0712300. SC and SH acknowledge funding from the Ministry of Electronics and Information Technology (MeitY), Government of India, under Grant No. 4(3)/2024-ITEA. SC also acknowledges support from Fujitsu Ltd, Japan, and IIIT Hyderabad via the Faculty Seed Grant. TL and XW were supported by the National Natural Science Foundation of China (Grant Numbers 62372006 and 92365117). XW thanks the University of California, Berkeley, for its hospitality during his visit, where a part of this work was conducted. We thank Andrew M. Childs and Ronald de Wolf for valuable feedback on this work.

\begin{appendices}

\section{Some polynomial approximation results}
\label{sec:poly-approx}

\begin{lem}[Rectangle function, Lemma 29 of \cite{gilyen2019quantum}]
\label{lem:rec-approx}
    Let $\delta^{\prime}, \varepsilon^{\prime} \in\left(0, \frac{1}{2}\right)$ and $t \in$ $[-1,1]$. There exist an even polynomial $P^{\prime} \in \mathbb{R}[x]$ of degree $\mathcal{O}\left(\log \left(\frac{1}{\varepsilon^{\prime}}\right) / \delta^{\prime}\right)$, such that $\left|P^{\prime}(x)\right| \leq 1$ for all $x \in[-1,1]$, and
    \[
    \begin{cases}
    \;P^{\prime}(x) \in  {\left[0, \varepsilon^{\prime}\right]}  &\text { for all } x \in\left[-1,-t-\delta^{\prime}\right] \cup\left[t+\delta^{\prime}, 1\right] \\
    \;P^{\prime}(x) \in[1-\varepsilon',1]  & \text { for all } x \in\left[-t+\delta^{\prime}, t-\delta^{\prime}\right] \\ \end{cases}.
    \]
\end{lem}

\begin{lem}
\label{lem:poly-approx-1}
Let \( f(z) = (1 - z^2)^{n/2} \) with \( n \in \mathbb{N} \), \( s \in (0, 1/\sqrt{n}) \), and \( \varepsilon > 0 \). Define the truncated even polynomial \( P_k(z) \) by keeping the first \( k \) terms of the power series expansion of \( f(z) \). Then:
\begin{enumerate}
    \item To ensure the approximation error \( | f(z) - P_k(z) | \leq \varepsilon \) for all \( z \in [-s, s] \), the truncation order \( k \) must satisfy:
    \[
    k=\Omega(\log(1/\varepsilon)) .
    \]
    \item The maximum of \( | P_k(z) | \) on \( z \in [-1, 1] \) satisfies:
    \[
    M=O(n^k) .
    \]
    \item Moreover, there is an even polynomial $S_1(z)$ of degree $O(\sqrt{n}\log(1/\varepsilon)\log(n))$ such that 
    $|f(z)-S_1(z)|\leq \varepsilon$ for all $z\in[-s,s]$ and $|S_1(z)| \leq 1$ for all $z\in[-1,1]$.
\end{enumerate}
\end{lem}

\begin{proof}
The function \( f(z) = (1 - z^2)^{n/2} \) can be expanded as:
\[
f(z) = \sum_{m=0}^{\infty} (-1)^m \binom{ n/2 }{ m } z^{2m}.
\]
The remainder (error) \( E(z) \) is:
\[
E(z) = f(z) - P_k(z) = \sum_{m=k}^{\infty} (-1)^m \binom{ n/2 }{ m } z^{2m}.
\]
For \( z \in [-s, s] \),  we have
\[
| E(z) | \leq \sum_{m=k}^{\infty} \left| \binom{ n/2 }{ m } \right| s^{2m}\le \sum_{m=k}^{\infty}\Big(\frac{ns^2}{2}\Big)^m\frac{1}{m!}\le\Big(\frac{ns^2}{2}\Big)^{k}\frac{1}{k!}e^{ns^2/2}.
\]
Taking $k=\Omega(\log(1/\varepsilon))$ suffices to bound the truncation error by $\varepsilon$.

Now suppose $k\le n/4$. The maximum value of $|P_k(z)|$ for $z\in[-1,1]$ can be bounded by 
\[
\sum_{m=0}^{k-1} \binom{n/2}{m}\le k\binom{n/2}{k-1}\le k\Big(\frac{en}{2k}\Big)^{k} = O(n^k).
\]

Finally, we can further reduce the upper bound in Claim 2 to 1 by composing $P_k(z)$ with a rectangle function. More precisely, let $\tilde{S}_1$ be the truncated Taylor expansion of $f(z)$ over the interval $[-2s,2s]$, and $R(z)$ be the polynomial in Lemma~\ref{lem:rec-approx} with 
$$(\varepsilon',t,\delta')=(O(n^{-\log(1/\varepsilon)}),3s/2,s/2),
$$ which means $\deg(R)=O(\log(1/\varepsilon)\log(n)/s)$. Since $s=O(1/\sqrt{n})$, we have $$
|\tilde{S}_1(z)|\approx|(1-z^2)^{n/2}|\le1,
$$ for all $z\in[-2s,2s]$. Taking $S_1(z)=\tilde{S}_1(z)R(z)$, we observe that $S_1(z)$ is still an $O(\varepsilon)$-approximation of $f(z)$ for $z\in[-s,s]$ and is bounded by $1$ for $z\in[-1,1]$. The degree of $S_1(z)$ is 
$$
\deg(R)+\deg(\tilde{S}_1)=O(\sqrt{n}\log(1/\varepsilon)\log(n)).
$$
\end{proof}

\begin{lem}
\label{lem:poly-approx-2}
Let \( f(z) = \dfrac{ z }{ \sqrt{ 1 - z^2 } } \), \( s \in (0, 1/2) \), and \( \varepsilon > 0 \). Define the truncated odd polynomial $S_2(z)$ by keeping the first \( k \) terms of the power series expansion of \( f(z) \). Then:
\begin{enumerate}
    \item To ensure the approximation error \( | f(z) - S_2(z) | \leq \varepsilon \) for all \( z \in [-s, s] \), the truncation order \( k \) must satisfy:
    \[
    k\ge  \frac{\log(1/\varepsilon)}{2\log(1/s)}.
    \]
    \item The maximum of \( | S_2(z) | \) on \( z \in [-1, 1] \) satisfies:
    \[
    M\le 2\sqrt{\frac{k}{\pi}}.
    \]
\end{enumerate}
\end{lem}

\begin{proof}
The function \( f(z) = \dfrac{ z }{ \sqrt{ 1 - z^2 } } \) can be expanded as:
\[
f(z) = \sum_{ m = 0 }^{ \infty } \frac{ (2m)! }{ 4^m ( m! )^2 }z^{ 2m + 1 }.
\]
Note that \[
 \frac{ (2m)! }{ 4^m ( m! )^2 }=\frac{1}{4^m}\binom{2m}{m}\le \frac{1}{4^m}\frac{4^m}{\sqrt{\pi m}}=\frac{1}{\sqrt{\pi m}},
\]
so the truncation error for $z\in[-s,s]$ can be bounded by 
\[
|f(z)-S_2(z)|\le \sum_{m=k}^\infty \frac{s^{2m+1}}{\sqrt{\pi m}}\le \frac{s^{2k+1}}{\sqrt{\pi k}}\sum_{m=0}^{\infty} s^{2m}\le \frac{s^{2k+1}}{\sqrt{k}}\le s^{2k+1}. 
\] 
Therefore, taking $k\ge \frac{\log(1/\varepsilon)}{2\log(1/s)}$ suffices to make the truncation error bounded by $\varepsilon$.

The maximum value of $S_2(z)$ in $[-1,1]$ can be bounded by 
\[
\sum_{m=0}^{k-1}\frac{1}{\sqrt{\pi m}}\le \int_{0}^{k} \frac{1}{\sqrt{\pi m}} \,\d m=2\sqrt{\frac{k}{\pi}}.
\]
This completes the proof.
\end{proof}

The following proposition provides an approximate decomposition of a polynomial under normalization in a certain sense.

\begin{prop}
\label{prop_implement f(sx)}
Let $\varepsilon\in (0,1/2)$ and $P(x)$ be a degree-$d$ polynomial satisfying $|P(x)|\le 1/2$ for any $x\in[-1,1]$. 
Then there exists a degree-$m$ polynomial $\widetilde{P}(x)$ and $c,s, \alpha \in \mb R$ such that 
\be
\label{eq: app-some parameters}
m = \Theta\!\left(d \log^2\!\left(d/\varepsilon\right) \log^2\!\left(\log\left(d/\eps\right)\right)\right),  
\quad  
s = 1/\sqrt{m},  
\quad  
c = \sqrt{1-1/m},  
\quad  
\alpha = \sqrt{m-1}\,\log\!\big(d^2/\varepsilon\big),
\ee
    and 
    \be \label{eq: approximation}
        \left|\left(\sqrt{c^2+s^2 x^2}\right)^m\widetilde{P}\left(\dfrac{sx}{\sqrt{c^2+s^2 x^2}}\right)-P\left(x/\alpha\right)\right|\le \varepsilon
    \ee
    for all $x\in[-1,1]$. 
    Moreover, if $f(x)$ is even or odd, then so is $\tilde{f}(x)$.
\end{prop}

\begin{proof}
From the choice of $c$, we see that $c^m = (1-1/n)^{m/2} \geq 1/2$ when $m\geq 2$. So we have $|P(x)| \leq c^m$ in $[-1,1]$. 
Substituting 
$$
z=\dfrac{sx}{\sqrt{c^2+s^2x^2}},
$$ 
into Eq.~\eqref{eq: approximation}, we obtain 
\bes
    \left|\widetilde{P}(z) - \dfrac{(\sqrt{1-z^2})^{m}}{c^m} \cdot P\left(\frac{ c}{\alpha s}\frac{z}{\sqrt{1-z^2}}\right) \right|
    \leq \dfrac{(\sqrt{1-z^2})^{m}}{c^m} \cdot \varepsilon
\ees
for all $z\in[-s,s]$. To ensure the above holds, it suffices to find $\widetilde{P}(z)$ such that for all $z\in[-s,s]$
\bes
    \left|\widetilde{P}(z) - \dfrac{(\sqrt{1-z^2})^{m}}{c^m} \cdot P\left(\dfrac{ c}{\alpha s}\dfrac{z}{\sqrt{1-z^2}}\right) \right|
    \leq \varepsilon.
\ees

Let $S_1(z), S_2(z)$ be the two polynomials specified in Lemmas \ref{lem:poly-approx-1} and \ref{lem:poly-approx-2}, that $\varepsilon$-approximates $(\sqrt{1-z^2})^{m}$ and $z/\sqrt{1-z^2}$ for $z\in[-s,s]$, and is  bounded by $M_1=1, M_2=\log(1/\varepsilon)$ for $z\in[-1,1]$, respectively. 
Take 
\bes
\widetilde{P}(z) :=  \dfrac{S_1(z)}{c^m}\cdot P\left(\dfrac{S_2(z)}{M_2}\right)
\ees 
and $\alpha = cM_2/s = 
 \sqrt{m-1}\log(1/\varepsilon)$. Then, we have 
\bes
    |\widetilde{P}(z)|\le M_1 = 1
\ees  
for all $z\in [-1,1]$ since $|S_2(z)/M_2|\le1$ and $f(z)\le c^m$ for $|z|\le 1$. 
By the Markov brothers' inequality, the derivative of $P(x)$ is bounded by $d^2|P(x)|\le d^2/2$ on $[-1,1]$. 
Moreover, for $z\in [-s,s]$ we have 
\beas
&& \left|\widetilde{P}(z)-\dfrac{(\sqrt{1-z^2})^{m}}{c^m}\cdot P\left(\dfrac{c}{\alpha s}\dfrac{z}{\sqrt{1-z^2}}\right)\right|\\
&=& \left| \dfrac{S_1(z)}{c^m}\cdot P\left(\frac{S_2(z)}{M_2}\right) -\dfrac{(\sqrt{1-z^2})^{m}}{c^m} \cdot P\left(\frac{z}{M_2\sqrt{1-z^2}}\right)\right|\\
&\leq& \left| \dfrac{S_1(z)}{c^m}-\dfrac{(\sqrt{1-z^2})^{m}}{c^m} \right|\cdot \left| P\left(\dfrac{S_2(z)}{M_2}\right) \right|
+ \left|\dfrac{(\sqrt{1-z^2})^{m}}{c^m} \right| \cdot \left| P\left(\dfrac{S_2(z)}{M_2}\right) - P\left(\dfrac{z}{M_2\sqrt{1-z^2}}\right) \right|
\\
&\leq& \varepsilon + \dfrac{d^2 \varepsilon }{2 M_2 \cdot c^m} \\
&\leq& \varepsilon \left(1+ \dfrac{d^2}{\log(1/\varepsilon)}\right)
= O\left(\varepsilon d^2\right),
\eeas
where the second inequality follows from the facts that 
$$\|P'(x)\|_{[-1,1]}\le d^2/2,
$$ 
and $S_2(z)$ is an $\varepsilon$-approximation of $z/\sqrt{1-z^2}$ on $[-s,s]$. To ensure the desired accuracy, we need to update $\varepsilon$ with $\varepsilon/d^2$.

The degree of $\widetilde{P}(x)$ is $m=d \cdot \deg(S_2)+\deg(S_1)$, which is
\[
m  \gtrsim \frac{d\log(d^2/\varepsilon)}{2\log(\sqrt{m})}+\sqrt{m} \log(m) \log(d^2/\varepsilon).
\]
This shows that $m = \Theta(d\log^2(d/\varepsilon)\log^2(\log(d/\varepsilon)))$.
\end{proof}

\section{Properties of the Transverse-Field Ising Model Hamiltonian}
\label{sec-app:tfim}
In this section, we derive two results about the Hamiltonians of the transverse-field Ising chains we considered in Sec.~\ref{subsec:spin-hamiltonians}. 

\subsection{Spectral gap of the Transverse-field Ising Hamiltonian with long-range interactions}

Consider the one-dimensional long-range transverse-field Ising model on a chain of $n$ spins with Hamiltonian (defined in Eq.~\eqref{eq:Ham-TFIM-long-range})
$$
H =  \underbrace{- h \sum_{i=1}^n X_i}_{H_0} -\underbrace{J \sum_{1 \le i < j \le n} \frac{1}{|i-j|^\alpha}Z_i Z_j}_{V}
$$
where $J, h > 0$ and $\alpha > 1$. We argue using first-order perturbation theory, based on the decomposition $H = H_0 + V$, that the spectral gap is bounded below by 
$$
\Delta\geq 2(h - J\zeta(\alpha)),
$$ 
when $h/J > \zeta(\alpha)$. While this result is not a rigorous proof for the full Hamiltonian $H$, as contributions from higher-order perturbation terms have not been bounded, it provides strong evidence for the existence of a large gap, independent of the system size. Standard perturbation theory suggests that for a sufficiently large ratio $h/J$, where $V$ is demonstrably a small perturbation to $H_0$, the true spectral gap is expected to remain open and close to this first-order estimate.

\begin{figure}[htbp]
    \centering
    
    \includegraphics[width=0.55\linewidth]{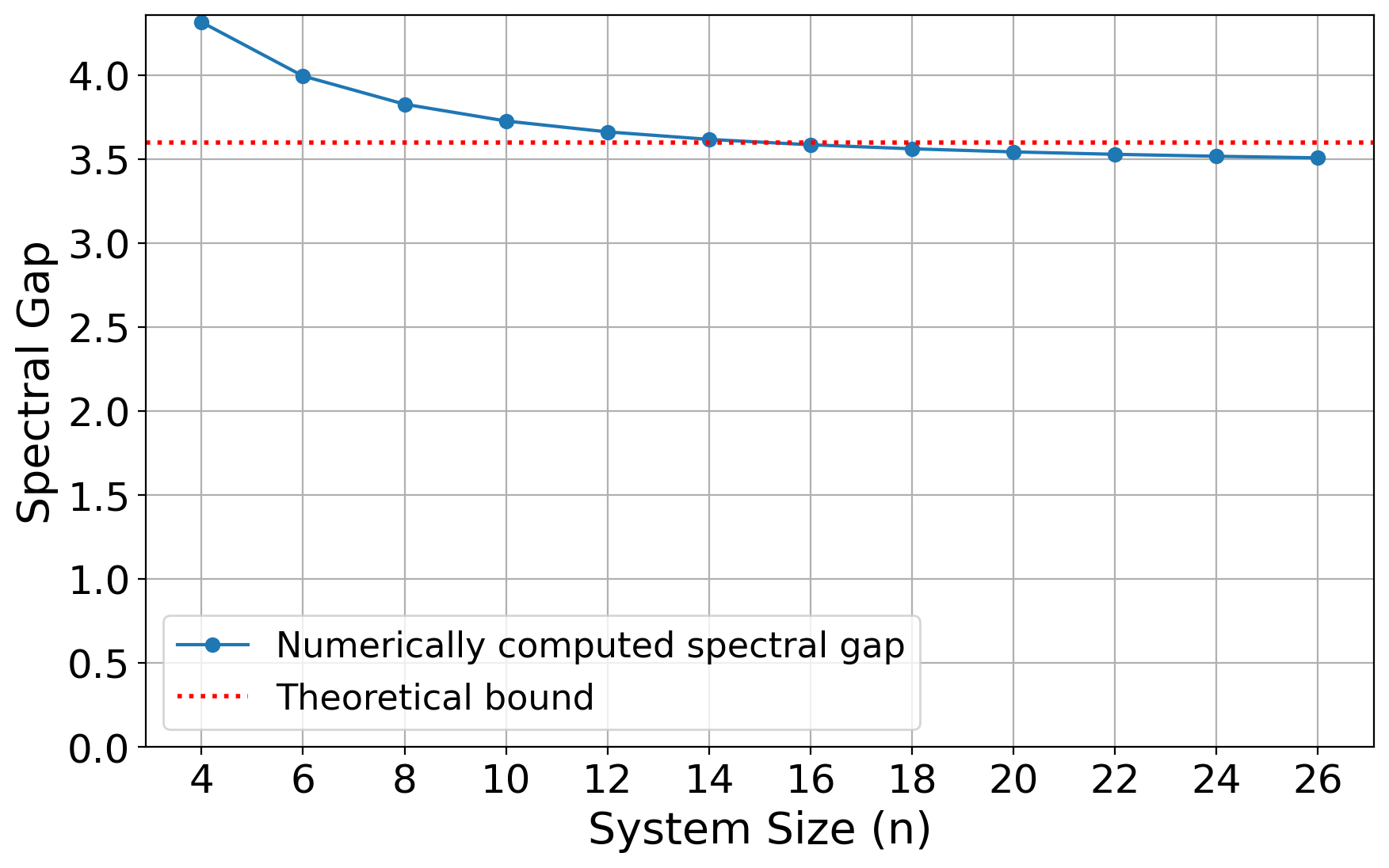}
    \caption{Spectral gap of the one-dimensional long-range transverse-field Ising model with $J=1$, $h=3$, and $\alpha=3$, as a function of the number spins $n$.}
    \label{fig:spectral-gap}
\end{figure}

The spectrum of $H_0$ is determined by the eigenstates of the Pauli $X$ operator. The ground state $|\psi_0^{(0)}\rangle = \bigotimes_{i=1}^n |+\rangle_i$ has all spins polarized in the $+x$ direction with energy $E_0^{(0)} = -nh$. For any integer $m \in \{1, \dots, n\}$, the $m$-th excited subspace $\mathcal{H}_m$ is spanned by the $\binom{n}{m}$ states formed by flipping exactly $m$ spins from $|+\rangle$ to $|-\rangle$. These states are degenerate with energy $E_m^{(0)} = -(n-2m)h$.

We apply first-order degenerate perturbation theory to a general $m$-excited subspace $\mathcal{H}_m$ for any $m \ge 1$. The first-order energy corrections are the eigenvalues $\lambda^{(m)}$ of the $\binom{n}{m} \times \binom{n}{m}$ matrix $M^{(m)}$ with elements $M^{(m)}_{kj} = \langle \psi_k^{(m)} | V | \psi_j^{(m)} \rangle$. The diagonal elements 
$$
M^{(m)}_{kk} = \langle \psi_k^{(m)} | V | \psi_k^{(m)} \rangle = 0,
$$ 
since $\langle\pm|Z|\pm\rangle=0$. The eigenvalues $\lambda^{(m)}$ are bounded by the Gershgorin Circle Theorem, which implies $|\lambda^{(m)}| \le \max_k \sum_{j \neq k} |M^{(m)}_{kj}|$. An off-diagonal element $M^{(m)}_{kj}$ is non-zero if $|\psi_j^{(m)}\rangle$ is reachable from $|\psi_k^{(m)}\rangle$ by a $Z_a Z_b$ operator, which requires swapping a $|+\rangle$ spin at site $a$ with a $|-\rangle$ spin at site $b$. For any state $|\psi_k^{(m)}\rangle$ with $m$ flipped spins, the row sum is bounded by:
\begin{align*}
\max_k \sum_{j \neq k} |M^{(m)}_{kj}| &= \max_k \sum_{\substack{a \text{ unflipped in } k \\ b \text{ flipped in } k}} \frac{J}{|a-b|^\alpha} \\
&\le \sum_{b \text{ flipped in } k} \left( \sum_{a \neq b} \frac{J}{|a-b|^\alpha} \right)\\
&< \sum_{b \text{ flipped in } k} 2J\zeta(\alpha) = 2mJ\zeta(\alpha).
\end{align*}
This uniform bound $|\lambda^{(m)}| < 2mJ\zeta(\alpha)$ holds for all $m \ge 1$. Thus, the first-order corrected energies $E_m'$ for the $m$-th excited subspace lie in the band $[E_m^{(0)} - 2mJ\zeta(\alpha), E_m^{(0)} + 2mJ\zeta(\alpha)]$. The lowest energy in this band is bounded below by 
$$
E_{m,\min}' \ge -(n-2m)h - 2mJ\zeta(\alpha).
$$

To establish a lower bound on the true spectral gap, we use the variational principle for the ground state energy, $E_0 \le \langle \psi_0^{(0)} | H | \psi_0^{(0)} \rangle = -nh$. The gap between the bottom of the $m$-th excited band and the upper bound on the ground state energy is therefore
\begin{align*}
E_{m,\min}' - E_0 &\ge \left( -(n-2m)h - 2mJ\zeta(\alpha) \right) - (-nh) \\
&= 2m(h - J\zeta(\alpha)).
\end{align*}
This holds for all $m \ge 1$. The overall first-order spectral gap $\Delta'$ is bounded by the minimum of these gaps over all $m$:
\[
\Delta' \ge \min_{m \ge 1} \{2m(h - J\zeta(\alpha))\}.
\]
Provided $h/J > \zeta(\alpha)$, the term in the parenthesis is positive, and the minimum is attained at $m=1$. Thus, the spectral gap is governed by the first excited subspace, with a lower bound of $\Delta' \ge 2(h - J\zeta(\alpha))$. This ensures that no excited energy crosses the ground state energy at first order.

Our numerical computation of the spectral gap for systems up to $n=26$ is presented in Fig.~\ref{fig:spectral-gap}. The data shows that the gap converges to an asymptotic value that is slightly below the first-order theoretical prediction of $2(h-J\zeta(\alpha))$. This suggests that while first-order theory provides a very good estimate, the net effect of higher-order perturbative corrections in this regime is to lower the spectral gap slightly.

\subsection{Pre-factor scaling of higher order Trotterization for Ising Hamiltonians with hybrid interactions}
In this section, we derive $\lambda_{\rm comm}$ and $\lambda_{\rm comm}'$ for the Hamiltonians of the spin models considered in this work. We explicitly derive a lower and an upper bound for $H$ defined in Eq.~\eqref{eq:Ham-TFIM-hybrid}. The derivation is similar for the standard long-range Ising Hamiltonian in Eq.~\eqref{eq:Ham-TFIM-long-range}. Formally, we state the following lemma
\begin{lem} 
\label{lem:commutator-prefactor-scaling}
Fix $h,J>0$, $g\ge 0$, and $\alpha>1$. For $n\ge 3$ consider the $n$-qubit Hamiltonian
\[
H = -h\sum_{i=1}^{n} Z_i - J\sum_{i=1}^{n} X_iX_{i+1} -\frac{g}{n}\sum_{1\le i<j\le n}\dfrac{1}{\left|i-j\right|^{\alpha}}\,Z_iZ_j,
\]
with periodic boundary conditions and let $\mathcal S$ be the fine-grained set of its Pauli-string summands,
\[
\mathcal S = \{hZ_i:1\leq i\leq n\}\ \cup\ \{J~X_iX_{i+1}:1\le i\le n\}\ \cup\ \left\{\tfrac{g}{n}|i{-}j|^{-\alpha}~Z_iZ_j:1\le i<j\le n\right\}.
\]
For a fixed integer $k\ge 1$, define the (universal) $(2k{+}1)$-fold commutator prefactor
$$
\alpha_{\rm comm}^{(2k+1)}
=\sum_{(S_1,\dots,S_{2k+1})\in\mathcal S^{\,2k+1}}
\left\|\operatorname{ad}_{S_1}\circ\cdots\circ \operatorname{ad}_{S_{2k}}(S_{2k+1})\right\|,
\qquad \operatorname{ad}_A(B):=[A,B].
$$
Then there exist constants $C_1, C_2$, independent of $n$, such that
\begin{equation}
\label{eq:ub-lb-alpha-comm}
    n C_1  \leq \alpha_{\rm comm}^{(2k+1)}\leq n C_2,
\end{equation}
for all $n\geq 3$. So,
$$
\lambda_{\rm comm}=\Theta(n^{1/(2k+1)}), \quad \lambda_{\rm comm}' =\left(\alpha^{(2k+1)}_{\rm comm}\right)^{1/(2k+1)}= \Theta(n^{1/(2k+1)}).
$$
\end{lem}
\begin{proof}  
By definition of $\lambda_{\rm comm}'$, we have \begin{align*}
    \lambda_{\rm comm}' = \left(\alpha^{(2k+1)}_{\rm comm}\right)^{1/(2k+1)}.
\end{align*}
The fact that
$$
\lambda_{\rm comm}=O(n^{1/(2k+1)}),
$$
follows from Eq.~\eqref{eq:ub-lb-alpha-comm} and \cite[Remark~3.2]{chakraborty2025quantum}. Ref.~\cite[Remark~3.2]{chakraborty2025quantum} also shows that \begin{align*}
    \lambda_{\rm comm} \ge \left(\alpha^{(2k+1)}_{\rm comm}\right)^{1/(2k+1)}.
\end{align*} 
Thus, all we need to prove is Eq.~\eqref{eq:ub-lb-alpha-comm}. This is what we will show in the subsequent steps of this proof.

For simplicity, let us define
\[
A_i:=hZ_i,\qquad B_i:=JX_iX_{i+1},\qquad C_{ab}:=\dfrac{g}{n}\cdot\frac{Z_aZ_b}{|a-b|^\alpha}\quad (1\leq a<b\leq n),
\]
so that we have
\[
\mathcal S=\{A_i:1\le i\le n\}\ \cup\ \{B_i:1\le i\le n\}\ \cup\ \{C_{ab}:1\le a<b\le n\}.
\]
So, using elementary Pauli algebra, we immediately obtain,
\begin{equation}
\label{eq:basic-comm-rel}
[A_i,A_j]=0,\quad [B_i,B_j]=0\ \text{ for all }i,j,\quad [A_\ell,C_{ab}]=0\ \text{ for all }\ell,a,b,
\end{equation}
and
\begin{equation}
\label{eq:overlap-rule}
[C_{ab},B_i]\neq 0\quad\Longleftrightarrow\quad |\{a,b\}\cap\{i,i+1\}|=1,
\end{equation}
while $[C_{ab},C_{cd}]=0$ always. Moreover,
\begin{equation}
\label{eq:AZZXX}
[A_i,B_i]=2ihJ\,Y_iX_{i+1},\qquad [A_{i+1},B_i]=2ihJ\,X_iY_{i+1},
\end{equation}
and $[A_j,B_i]=0$ for $j\notin\{i,i+1\}$. By definition, we have
\[
\alpha_{\rm comm}^{(2k+1)}
=\sum_{(S_1,\dots,S_{2k+1})\in\mathcal S^{\,2k+1}}
\left\|~\operatorname{ad}_{S_1}\circ \cdots\circ \operatorname{ad}_{S_{2k}}(S_{2k+1})~\right\|,
\qquad \operatorname{ad}_A(B)=[A,B],
\]
is a sum of norms of $(2k{+}1)$-fold nested commutators over \emph{all} ordered $(2k{+}1)$-tuples from $\mathcal S$. We prove a lower bound of order $n$ by exhibiting $n$ explicit ordered tuples whose contributions are strictly positive and identical up to the constant $2^{2k}\min\{h^{k+1}J^{k},h^{k}J^{k+1}\}$, and then we prove an upper bound of order $n$ by showing that every nonzero nested commutator is \emph{anchored} on a unique bond and contributes at most a bond-dependent constant independent of $n$.

\medskip\noindent\emph{Lower bound.}
Fix a bond index $i\in\{1,\dots,n\}$. Consider the two alternating words of length $2k{+}1$
\begin{equation}
\label{eq:ladder-words}
\mathbf w^{(A)}_{i,k}=(A_i,B_i,A_i,B_i,\dots,A_i,B_i,A_i),
\qquad
\mathbf w^{(B)}_{i,k}=(B_i,A_i,B_i,A_i,\dots,B_i,A_i,B_i),
\end{equation}
in which $A_i$ and $B_i$ alternate and the first/last entries are as indicated. Define the associated nested commutators
\[
\mathcal L^{(A)}_{i,k}:=\operatorname{ad}_{A_i}\circ\operatorname{ad}_{B_i}\circ\cdots\circ\operatorname{ad}_{A_i}(B_i),
\qquad
\mathcal L^{(B)}_{i,k}:=\operatorname{ad}_{B_i}\circ\operatorname{ad}_{A_i}\circ\cdots\circ\operatorname{ad}_{B_i}(A_i),
\]
where the number of $\operatorname{ad}$'s is $2k$ in each case. We claim that, for every $k\ge 1$,
\begin{equation}
\label{eq:ladder-norms}
\big\|\mathcal L^{(A)}_{i,k}\left\|=2^{2k}h^{k+1}J^{k},\qquad
\right\|\mathcal L^{(B)}_{i,k}\big\|=2^{2k}h^{k}J^{k+1}.
\end{equation}
We prove Eq.~\eqref{eq:ladder-norms} by induction on $k$. For $k=1$, using Eq.~\eqref{eq:AZZXX} and $[Z_i,Y_i]=-2iX_i$ we compute
\[
[A_i,B_i]=2ihJ\,Y_iX_{i+1},\qquad
[A_i,[A_i,B_i]]=2ihJ\,[Z_i,Y_i]X_{i+1}=2ihJ\cdot(-2i)X_iX_{i+1}=4h^2J\,X_iX_{i+1},
\]
whence $\|\mathcal L^{(A)}_{i,1}\|=4h^2J=2^{2}\,h^{2}J$. Similarly, using $[X_i,Y_i]=2iZ_i$,
\[
[B_i,[B_i,A_i]]=[JX_iX_{i+1},-2ihJ\,Y_iX_{i+1}]=-2ihJ^2\,[X_iX_{i+1},Y_iX_{i+1}]
=-2ihJ^2\,[X_i,Y_i]=4hJ^2\,Z_i,
\]
so $\|\mathcal L^{(B)}_{i,1}\|=4hJ^2=2^{2}\,hJ^{2}$. This establishes Eq.~\eqref{eq:ladder-norms} for $k=1$.

Assume Eq.~\eqref{eq:ladder-norms} holds for some $k\ge 1$. Consider $\mathcal L^{(A)}_{i,k+1}$. By definition,
\[
\mathcal L^{(A)}_{i,k+1}
=\operatorname{ad}_{A_i}\circ\operatorname{ad}_{B_i}\big(\mathcal L^{(A)}_{i,k}\big).
\]
By the induction hypothesis, $\mathcal L^{(A)}_{i,k}$ is, up to a phase, a single Pauli string supported on $\{i,i+1\}$ with norm $2^{2k}h^{k+1}J^{k}$. The commutator with $B_i=JX_iX_{i+1}$ acts nontrivially only on site $i$ (the site $i{+}1$ Pauli in $\mathcal L^{(A)}_{i,k}$ is either $X_{i+1}$ or $Z_{i+1}$, but $[X_{i+1},X_{i+1}]=0$ and $[Z_{i+1},X_{i+1}]=2iY_{i+1}$ gives again a single-site Pauli), and in any case $\|[B_i,\cdot]\|=2J\|\cdot\|$ because the commutator of $X_i$ with the site-$i$ Pauli flips $X_i\leftrightarrow Y_i$ (or $X_i\leftrightarrow Z_i$) and contributes a factor $2$, while the $X_{i+1}$ factor either commutes or also flips, contributing a factor of $1$ in norm. Thus,
\[
\left\|\operatorname{ad}_{B_i}\left(\mathcal L^{(A)}_{i,k}\right)\right\|
=2J~\left\|\mathcal L^{(A)}_{i,k}\right\|
=2^{2k+1}h^{k+1}J^{k+1}.
\]
Applying next $\operatorname{ad}_{A_i}$ multiplies the norm by a factor $2h$ for the same reason, hence
\[
\left\|\mathcal L^{(A)}_{i,k+1}\right\|
=\left\|\operatorname{ad}_{A_i}\circ\operatorname{ad}_{B_i}\left(\mathcal L^{(A)}_{i,k}\right)\right\|
=2h\cdot 2^{2k+1}h^{k+1}J^{k+1}
=2^{2(k+1)}h^{(k+1)+1}J^{k+1}.
\]
This proves the first identity in Eq.~\eqref{eq:ladder-norms} for $k+1$. The proof for $\mathcal L^{(B)}_{i,k+1}=\operatorname{ad}_{B_i}\circ\operatorname{ad}_{A_i}(\mathcal L^{(B)}_{i,k})$ is identical upon exchanging the roles of $A_i$ and $B_i$; one obtains
\[
\|\mathcal L^{(B)}_{i,k+1}\|=2J\cdot(2h)\,\|\mathcal L^{(B)}_{i,k}\|
=2^{2(k+1)}h^{k+1}J^{(k+1)+1}.
\]
This completes the induction and proves Eq.~\eqref{eq:ladder-norms} for all $k\ge 1$.

Since $\alpha_{\rm comm}^{(2k+1)}$ is a sum over \emph{all} ordered $(2k{+}1)$-tuples from $\mathcal S$, the two words $\mathbf w^{(A)}_{i,k}$ and $\mathbf w^{(B)}_{i,k}$ appear among the summands for each bond $i$. Therefore,
\begin{equation}
\label{eq:LB-final}
\alpha_{\rm comm}^{(2k+1)}
\ \ge\ \sum_{i=1}^n \min\{\|\mathcal L^{(A)}_{i,k}\|,\ \|\mathcal L^{(B)}_{i,k}\|\}
\ =\ 2^{2k}\min\{h^{k+1}J^{k},h^{k}J^{k+1}\}\, n.
\end{equation}
This yields the desired lower bound with $C_1=2^{2k}\min\{h^{k+1}J^{k},h^{k}J^{k+1}\}$.

\medskip\noindent\emph{Upper bound.}
We show that every nonzero $(2k{+}1)$-fold nested commutator in the sum is \emph{anchored} on a unique nearest-neighbor bond and contributes at most a bond-dependent constant. Let
\[
T:=\operatorname{ad}_{S_1}\circ\cdots\circ\operatorname{ad}_{S_{2k}}(S_{2k+1}),
\qquad S_\ell\in\mathcal S,
\]
be such that $T\neq 0$. If none of the $S_\ell$ is a bond operator $B_i$, then all $S_\ell$ are $Z$-type strings ($A_i$ or $C_{ab}$), which commute by Eq.~\eqref{eq:basic-comm-rel}, and hence $T=0$, a contradiction. Thus some $B_i$ must appear. Choose one such $B_i$ and call $i$ the \emph{anchor index}. We claim the following support property: for each depth $d=0,1,\dots,2k$, if we define recursively
\[
T_0:=S_{2k+1},\qquad T_{d+1}:=[S_{d+1},T_d],
\]
then $T_d$ is either $0$ or a single Pauli string supported in the set $\{i,i+1\}\cup \mathcal J_d$, where $\mathcal J_d$ is a subset of $\{1,\dots,n\}\setminus\{i,i+1\}$ consisting of those indices that appear as the \emph{far} endpoint in $C_{ab}$'s that touch exactly one of the anchor sites at the steps $1,\dots,d$. Moreover, if $S_{d+1}$ is a local operator among $\{A_i,A_{i+1},B_i\}$, then $S_{d+1}$ must act nontrivially on at least one of the anchor sites and anticommute there with the site-Pauli of $T_d$ for the commutator not to vanish; if $S_{d+1}$ is long-range of the form $C_{ab}$, then necessarily $|\{a,b\}\cap\{i,i+1\}|=1$ (cf.~Eq.~\eqref{eq:overlap-rule}), and its far endpoint $j\in \{a,b\}\setminus\{i,i+1\}$ is adjoined to $\mathcal J_{d}$ to produce $\mathcal J_{d+1}$.

The support claim follows by induction on $d$. For $d=0$ it is obvious. Assume it holds at depth $d$, so $T_d$ is a Pauli string supported on $\{i,i+1\}\cup \mathcal J_d$. If $S_{d+1}\in\{A_s:1\le s\le n\}$, then $[A_s,T_d]\neq 0$ only if $s\in\{i,i+1\}$ (since $A_s$ commutes with any $Z$ on sites in $\mathcal J_d$ and with any $X$ or $Y$ off site $s$), and moreover the commutator is nonzero only if the site-$s$ Pauli in $T_d$ anticommutes with $Z_s$; in that case $[A_s,T_d]$ is again a single Pauli string supported on $\{i,i+1\}\cup\mathcal J_d$. If $S_{d+1}=B_i$, then $[B_i,T_d]\neq 0$ only if at least one of the site-$i$ or site-$(i+1)$ Paulis in $T_d$ anticommutes with $X$; in that case $[B_i,T_d]$ is again a Pauli string supported on $\{i,i+1\}\cup\mathcal J_d$ (the $X_{i+1}$ factor commutes through any $Z_j$, $j\in\mathcal J_d$). If $S_{d+1}=C_{ab}$, then $[C_{ab},T_d]\neq 0$ only if $|\{a,b\}\cap\{i,i+1\}|=1$ by Eq.~\eqref{eq:overlap-rule}, and then the far endpoint $j\in\{a,b\}\setminus\{i,i+1\}$ is appended to the set of remote indices: $[C_{ab},T_d]$ is again a single Pauli string supported on $\{i,i+1\}\cup(\mathcal J_d\cup\{j\})$. This proves the induction step and hence the support property.

As a consequence, every nonzero nested commutator is \emph{localized} at a unique anchor bond $i$ and can be generated by an ordered word of length $2k{+}1$ in which each letter is either \emph{local} ($A_i$, $A_{i+1}$, or $B_i$) or \emph{long-range} of the form $C_{i,j}$ or $C_{i+1,j}$ with $j\notin\{i,i+1\}$. Let us upper bound the total contribution of all such words anchored at bond $i$.

For a fixed anchor $i$, consider first purely local words (no long-range letters). For any such word $\mathbf u=(U_1,\dots,U_{2k+1})$ with $U_\ell\in\{A_i,A_{i+1},B_i\}$, repeated application of the Pauli commutators shows that the nested commutator $T(\mathbf u):=\operatorname{ad}_{U_1}\circ\cdots\circ\operatorname{ad}_{U_{2k}}(U_{2k+1})$ is either zero or (when nonzero) a single Pauli string on $\{i,i+1\}$ multiplied by a scalar of the form $2^{2k} h^{p}J^{2k+1-p}$ for some $p\in\{0,1,\dots,2k+1\}$ (each bracket by $A$ multiplies by $2h$ and each bracket by $B$ multiplies by $2J$ in norm). Therefore
\begin{equation}
\label{eq:local-per-bond}
\sum_{\mathbf u\in \{A_i,A_{i+1},B_i\}^{2k+1}}\big\|T(\mathbf u)\big\|
\ \le\ 3^{2k+1}\,2^{2k}\,\max\{h,J\}^{2k+1}
=:C_k^{\mathrm{loc}}(h,J),
\end{equation}
a constant depending only on $k,h,J$.

Next consider words that contain exactly $m\ge 1$ long-range letters. Any such letter must be either $C_{i,j}$ or $C_{i+1,j}$ with $j\notin\{i,i+1\}$ by the support property. The coefficient of $C_{i,j}$ is $(g/n)|i-j|^{-\alpha}$; summing this over all admissible $j$ yields
\[
\sum_{j\neq i,i+1} \dfrac{g}{n}\frac{1}{|i-j|^\alpha}
\ \le\ \frac{g}{n}\left(\sum_{r\ge 1}\frac{1}{r^\alpha}+\sum_{r\ge 1}\frac{1}{r^\alpha}\right)
 \leq\ \dfrac{2g}{n}\zeta(\alpha).
\]
Thus, after summing over all allowed choices of the far endpoints for the $m$ long-range letters in positions where they yield a nonzero commutator, the total weight contributed by the long-range coefficients is bounded by
\[
\left(\frac{2g}{n}\zeta(\alpha)\right)^{m}.
\]
The remaining $2k{+}1-m$ letters are local and contribute a factor bounded by $2^{2k}\max\{h,J\}^{2k+1-m}$ in norm, exactly as in the local case. The number of operator-type patterns with exactly $m$ long-range letters is at most ${2k+1\choose m}\cdot 3^{2k+1-m}$ (choose the $m$ positions and fill the rest with three local choices). Hence the total contribution of all words with exactly $m$ long-range letters anchored at bond $i$ is bounded by
\[
{2k+1\choose m}\, 3^{2k+1-m}\, 2^{2k}\,\max\{h,J\}^{2k+1-m}\,\left(\frac{2g}{n}\zeta(\alpha)\right)^{m}.
\]
Summing over $m\ge 1$ and using the binomial theorem we obtain the bound
\begin{equation}
\label{eq:lr-per-bond}
\sum_{\substack{\text{words anchored at }i\\ \text{with }m\ge 1\text{ long-range letters}}}\!\!\!\!\!\!\!\!\!\!\!\!\!\!\!\|T(\cdot)\|
\ \le\ 2^{2k}\,\left(3\max\{h,J\}+\frac{2g}{n}\zeta(\alpha)\right)^{2k+1}
- C_k^{\mathrm{loc}}(h,J)
=:C_k^{\mathrm{lr}}(h,J,g,\alpha;n),
\end{equation}
where $C_k^{\mathrm{lr}}(h,J,g,\alpha;n)$ is bounded as the term involving $(2g/n)\zeta(\alpha)$ vanishes as $n\to\infty$ and the difference subtracts the purely local part already counted in Eq.~\eqref{eq:local-per-bond}. 

In particular, there exists a constant $C_k^{\mathrm{lr}}(h,J,g,\alpha)$ such that $C_k^{\mathrm{lr}}(h,J,g,\alpha;n)\le C_k^{\mathrm{lr}}(h,J,g,\alpha)$ for all $n$. Then, combining Eq.~\eqref{eq:local-per-bond} and Eq.~\eqref{eq:lr-per-bond}, the total contribution of all words anchored at a fixed bond $i$ is at most
\[
C_k^{\mathrm{loc}}(h,J)+C_k^{\mathrm{lr}}(h,J,g,\alpha),
\]
a constant independent of $n$. Finally, summing over the $n$ anchor bonds yields
\begin{equation}
\label{eq:UB-final}
\alpha_{\rm comm}^{(2k+1)}
\ \le\ n\,C_k^{\mathrm{loc}}(h,J)+ n~\sup_{n}C_k^{\mathrm{lr}}(h,J,g,\alpha;n)
\ \leq C_2 n.
\end{equation}

The lower bound \eqref{eq:LB-final} and the upper bound \eqref{eq:UB-final} together imply that $\alpha_{\rm comm}^{(2k+1)}=\Theta(n)$. This completes the proof.
\end{proof}

\section{Randomized QSVT for density operators}

We prove a result analogous to Theorem \ref{thm: randomised QSVT}, but for density operators. Formally, we state the following result:

\begin{thm}
\label{thm: randomised QSVT in the general case}
Consider the same assumptions as in Theorem \ref{thm: randomised QSVT}. Let $O$ be an observable and $\rho$ be a density operator. Then there is a procedure using $T=O(\varepsilon^{-2} \log(1/\delta))$ repetitions of the quantum circuit in Proposition \ref{thm_SVT of non-singular A}, outputs $\mu$ such that
    \[
    \left| \mu- 
     {\rm Tr}\left[
    O \,\,
    {P(H')} \,
    \rho \,
    {P(H')}^\dagger  
    \right] 
     \right| \leq \varepsilon \|O\|
    \]
    with probability at least $1-\delta$. Each repetition has circuit depth $\widetilde{O}(d)$ and uses two ancilla qubits. 
\end{thm}

\begin{proof}
We consider the case when $P(x)$ is an even polynomial.
Denote $\rho'=\ket{1}\bra{1}\otimes \rho$.\footnote{When $f(x)$ is odd, we cannot directly assume $\rho'=\ket{0}\bra{1}\otimes \rho$ as it is not a density operator. However, we can multiply $U_\Phi$ by $X\otimes I$ to change the position of $P(H')$ so that we can now use $\rho'=\ket{0}\bra{0}\otimes \rho$.} Also denote
$\tilde{\rho}=\ket{+}\bra{+}\otimes \rho'$. 
Let
$$\tilde{U}^{(R)}_\Phi
=\ket{0}\bra{0}\otimes U^{(R_0)}_\Phi
+ \ket{1}\bra{1}\otimes U^{(R_1)}_\Phi,
$$ 
be a random quantum circuit with two independent applications of $U^{(R)}_\Phi$. Then it is easy to check that
\beas
\tilde{U}^{(R)}_\Phi \, \tilde{\rho} \, (\tilde{U}^{(R)}_\Phi)^\dagger 
&=& 
\frac{1}{2}\ket{0}\bra{0}\otimes U^{(R_0)}_\Phi \rho' (U^{(R_0)}_\Phi)^\dagger
+
\frac{1}{2}\ket{0}\bra{1}\otimes U^{(R_0)}_\Phi \rho' (U^{(R_1)}_\Phi)^\dagger \\
&& 
+\,
\frac{1}{2}\ket{1}\bra{0}\otimes U^{(R_1)}_\Phi \rho' (U^{(R_0)}_\Phi)^\dagger
+
\frac{1}{2}\ket{1}\bra{1}\otimes U^{(R_1)}_\Phi \rho' (U^{(R_1)}_\Phi)^\dagger.
\eeas
Denote $O'=\ket{1}\bra{1}\otimes O$, then we have
\beas
{\rm Tr}\left[ (X \otimes O') \,\, \tilde{U}^{(R)}_\Phi \, \tilde{\rho} \, (\tilde{U}^{(R)}_\Phi)^\dagger  \right]
= \frac{1}{2}
{\rm Tr}\left[ O' \,\,  U^{(R_0)}_\Phi \rho' (U^{(R_1)}_\Phi)^\dagger  \right]
+
\frac{1}{2}
{\rm Tr}\left[ O' \,\, U^{(R_1)}_\Phi \rho' (U^{(R_0)}_\Phi)^\dagger  \right] .
\eeas
This means
\bes
\begin{aligned}
\E\left[ {\rm Tr}\left[ (X \otimes O') \,\, \tilde{U}^{(R)}_\Phi \, \tilde{\rho} \, (\tilde{U}^{(R)}_\Phi)^\dagger  \right] \right]
&=
\E\left[ {\rm Tr}\left[ O' \,\,  U^{(R_0)}_\Phi \rho' (U^{(R_1)}_\Phi)^\dagger  \right] \right] \\
&=
 {\rm Tr}\left[ O' \,\,  U_\Phi \rho' U_\Phi  \right]  \\
 &\approx_\eps {\rm Tr}\left[ O'\, \begin{bmatrix}
     * & * \\ * & {P(H')} \,
    \rho \,
    {P(H')}^\dagger
 \end{bmatrix}\right]
\\ &\approx_\eps
{\rm Tr}\left[
    O \,\,
    {P(H')} \,
    \rho \,
    {P(H')}^\dagger  
    \right].
\end{aligned}
\ees
Now we can apply Hoeffding's inequality to estimate the above mean value to obtain the claimed result by a similar argument to the proof of Theorem \ref{thm: randomised QSVT}. The trace can be computed by measuring $(X\otimes O')$ on the state $\tilde{U}^{(R)}_\Phi \, \tilde{\rho} \, (\tilde{U}^{(R)}_\Phi)^\dagger$, so there is no need to introduce a new ancilla. In total, 2 ancilla qubits are used.
\end{proof}

\end{appendices}

\bibliographystyle{unsrt}
\bibliography{qsvt.bib}

\end{document}